\documentclass[11pt]{article}

\RequirePackage[colorlinks,citecolor=blue,urlcolor=black,backref=true,linktocpage=true]{hyperref}
\usepackage{hyperref}
\usepackage{graphicx,amsmath,latexsym,amssymb}
\usepackage{float,epsfig,multirow,rotating,times}
\usepackage{url}
\usepackage{amsthm}
\usepackage{relsize}
\usepackage{subfigure}

\newdimen\dummy
\dummy=\oddsidemargin
\newtheorem{theorem}{Theorem}[section]
\newtheorem{lemma}[theorem]{Lemma}
\newtheorem{proposition}[theorem]{Proposition}
\newtheorem{corollary}[theorem]{Corollary}
\newtheorem{result}[theorem]{Result}
\newcommand{\abs}[1]{|#1|}

\begin{document}

\normalsize

\title{
\textbf{Gaussian Random Functional Dynamic Spatio-Temporal Modeling of Discrete-Time Spatial Time Series Data}}
\author{Suman Guha\thanks{Suman Guha is a Phd student in Interdisciplinary Statistical Research Unit, Indian Statistical Institute, 
203, B. T. Road, Kolkata 700108. His research is supported by CSIR SPM Fellowship, Govt. of India. \emph{Corresponding e-mail address:} 
sumanguha\_r@isical.ac.in.} \ and Sourabh Bhattacharya\thanks{Sourabh Bhattacharya is an Associate Professor in Interdisciplinary Statistical Research Unit, Indian Statistical Institute, 203, B. T. Road, Kolkata 700108.
\emph{Corresponding e-mail address:} sourabh@isical.ac.in.}}
\maketitle

\begin{abstract}
Discrete-time spatial time series data arise routinely in meteorological and environmental studies. Inference and 
prediction associated with them are mostly carried out using any of the several variants of the linear state space 
model that are collectively called linear dynamic spatio-temporal models (LDSTMs). However, real world environmental 
processes are highly complex and are seldom representable by models with such simple linear structure. Hence, 
nonlinear dynamic spatio-temporal models (NLDSTMs) based on the idea of nonlinear observational and evolutionary 
equations have been proposed as an alternative. However, in that case, the caveat lies in selecting the specific 
form of nonlinearity from a large class of potentially appropriate nonlinear functions. Moreover, modeling by 
NLDSTMs requires precise knowledge about the dynamics underlying the data. In this article, we address 
this problem by introducing the Gaussian random functional dynamic spatio-temporal model (GRFDSTM). 
Unlike the LDSTMs or NLDSTMs, in GRFDSTM both the functions governing the observational and evolutionary equations 
are composed of Gaussian random functions. 
We exhibit many interesting theoretical properties of the GRFDSTM and 
demonstrate how model fitting and prediction can be carried out coherently in a Bayesian framework. 
We also conduct an extensive simulation study and apply our model to a real, $\mathrm{SO_{2}}$ pollution data over Europe. 
The results are highly encouraging.   
 \\[2mm]
\\[2mm]

\end{abstract}

\section{Introduction}
\label{sec:intro}
Spatio-temporal modeling has received much attention in recent years. Particularly, the rise in global temperature being a major environmental concern, scientists are now taking keen interest in developing appropriate statistical models to study such climatic phenomena \cite{Furrer:Sain,Jun:Knutti,Sain:Furrer:Cressie,Sang:Jun,Smith:Tebaldi:Nychka,Tebaldi:Sanso}. Other closely 
related meteorological phenomena, that are also drawing much attention of modelers, are rainfall \cite{Cox:Isham,Sanso:Guenni} and precipitation (mist, snowfall, sulphate, nitrate \cite{Sahu:Gelfand:Holland} etc.) 
across different regions. Apart from meteorology, challenging spatio-temporal data also arise from environmental and 
ecological science. To mention a few, studies on ground level concentration of ozone \cite{Bruno:Guttorp:Sampson:Cocchi,Dou:Le:Zidek,Guttorp:Meiring:Sampson,Huerta:Sanso:Stroud}, $\mathrm{SO_{2}}$ \cite{Giannitrapani:Bowman:Scott:Smith,Holland:Oliveira:Cox:Smith}, $\mathrm{NO_{2}}$ \cite{Amini et al.} and PM \cite{Paciorek et al.} related air pollution, 
species distribution over a region \cite{Chakraborty et al.}, change in land usage pattern over time \cite{Deng:Wang:Hong:Qi}, etc. 

While the term spatio-temporal data encompasses the above few and includes many other different types of data, the above-mentioned ones belong to a particularly interesting class called discrete-time spatial time series data. Typically, they are obtained by collecting observations at arbitrary but fixed spatial locations at times $t=1, 2, 3,\cdots, T$. The data $Y(s_{i},t)$ can be represented as $Y_{s_{i}}(t)$, and then we have a whole set of time series, one for each spatial location. Although in reality, the data is available only at finitely many spatial locations (generally called monitoring sites), it is conceptually always useful to assume the existence of time series at every spatial location. An alternative view is to consider it as a time varying spatial random field. The problem facing the statistician is to develop an appropriate model, infer about the spatio-temporal process $Y(s,t)$ and possibly predict at new sites based on this partial realization. 

In what follows, we review some existing models and discuss the advantages and problems associated with them. Then we develop the Gaussian random functional dynamic spatio-temporal model (GRFDSTM) and study its theoretical properties in much detail. We demonstrate how model fitting and prediction can be done coherently using a Bayesian approach. Finally, we examine its performance in an extensive simulation study and apply it to an $\mathrm{SO_{2}}$ pollution dataset over Europe, and obtain very encouraging results. The proofs of all our mathematical results are deferred to the Appendix.

\section{Existing Approaches for Discrete-Time Spatial Time Series Data}
\label{sec:existing_approaches}
As already mentioned, discrete-time spatial time series data $y(\bold{s}_{i},t)$ consists of observations taken at some arbitrary but fixed spatial locations $\bold{s}_{1},\bold{s}_{2},\bold{s}_{3}\cdots, \bold{s}_{n}$ at times $t=1, 2, 3,\cdots, T$. Note that any method applicable to them can easily be extended to data collected at non equispaced time points. Hence, for the sake of simplicity we confine ourselves only to equispaced time points.

There exist two different perspectives from which one can develop a model for such data. One is the marginal approach, 
in which one develops a joint distribution for $y(\bold{s}_{i},t)$. In the case of the Gaussian model, this amounts to 
specifying a spatio-temporal mean function and covariance function. The other approach relies on specifying the 
conditional distribution of the current process realizations given the past process realizations. Although theoretically 
it is equivalent to specifying the conditional distribution or the marginal distribution, one being derivable from the other, 
the conditional approach being more closer to the etiology of the phenomena under study, is preferred over the marginal one. For more discussions on this issue interested reader may look into the book by \cite{Cressie10}. In fact, \cite{Cressie10} covers very wide range of materials on dynamic spatio-temporal modeling and also serves as the main resource for many of the models that we consider here.  

Although there exists a vast literature on marginal models, here we mention only a few of them. \cite{Reinsel et al.} proposed a regional-effects model for the analysis of stratospheric ozone data. Later, \cite{Bloomfield et al.} considered the spatio-temporal random effect $\eta(\bold{s},t)$ and extended it to the following model (see equation 6.16 on the page 306 of \cite{Cressie10})
\begin{align*}
     Y(\bold{s},t)=\gamma(t)+\kappa_{k}(t)+\eta(\bold{s},t)+\epsilon(\bold{s},t),\ \ \  \text{for all}\  \bold{s},
\end{align*}
where $\gamma(t)$ is the random time effect common to all spatial locations, $\kappa_{k}(t)$ is the random regional time effect 
common to all locations belonging to the $k$th region, $\eta(\bold{s},t)$ is the spatio-temporal random 
effect and $\epsilon(\bold{s},t)$ is a spatio-temporal white noise process representing the microscale 
spatio-temporal variability. Another interesting marginal model that hinges not only on the spatio-temporal 
effect but also incorporates useful covariate information is the following spatio-temporal hedonic model for house prices 
proposed by \cite{Gelfand:Ecker:Knight:Sirmans}: 
\begin{align*}
     Y(\bold{s},t)=\bold{X}(\bold{s},t)^{T}\boldsymbol{\beta}(\bold{s},t)+\alpha(t)+w(\bold{s})+\epsilon(\bold{s},t).
\end{align*}
Here, $Y(\bold{s},t)$ is the log selling price, $\alpha(t)$ is the common time effect for all the locations, 
$w(\bold{s})$ is the spatial effect, $\epsilon(\bold{s},t)$ is the Gaussian white noise and $\bold{X}(\bold{s},t)$ is the design matrix containing useful covariate information. This form allows spatio-temporally varying coefficients, which is perhaps more than what is required. Hence, $\boldsymbol{\beta}(\bold{s},t)=\boldsymbol{\beta}$ is frequently adopted. Setting $\boldsymbol{\beta}(\bold{s},t)=\boldsymbol{\beta}(t)$ yields an extension of the model proposed by \cite{Knight:Dombrow:Sirmans}.

Akin to the marginal approach several models for the conditional distribution have been proposed and most of them are based on the idea of linear state space models in time series. These models are generally referred to as linear dynamic spatio-temporal models (LDSTMs). A common representative form for them would be
\begin{align*}
     &Y(\bold{s},t)=\mu_{t}(\bold{s})+\epsilon_{t}(\bold{s}),\ \ \  \text{for all}\  \bold{s};\\
     &\mu_{t}(\bold{s})=\bold{X}_{t}(\bold{s})^{T}\boldsymbol{\beta}_{t};\\
     &\boldsymbol{\beta}_{t}=\bold{G}_{t}\boldsymbol{\beta}_{t-1}+\boldsymbol{\eta}_{t},\ \ \  \text{for all}\  \bold{s},
\end{align*}
where $\bold{X}_{t}(\bold{s})^{T}\boldsymbol{\beta}_{t}$ controls the temporal effect and $\epsilon_{t}(\bold{s})$ handles the spatial variation. The design matrix $\bold{X}_{t}(\bold{s})$ either consists of deterministic functions of locations or important covariate information associated with the locations. The observed data arise from the process $Y(\bold{s},t)$, which is temporally driven by the latent (or, state) process $\boldsymbol{\beta}_{t}$ and the temporal evolution of $\boldsymbol{\beta}_{t}$ is given by some suitable conditional model, preferably a Markovian one. \cite{Stroud:Muller:Sanso} applied the above model in 
the environmental context, where they chose to define $\bold{X}_{t}(\bold{s})=\bold{X}(\bold{s})$ as a linear combination of basis functions of the location $\bold{s}$. Their model is attractive in being applicable to any dataset that is continuous in space and discrete in time. 
\cite{Wikle:Cressie} used the same model in their dimension reduction approach. However, they obtained the model from a more 
general 
underlying process. \cite{Sanso:Schmidt:Nobre} also used this model, but without the error term.
The scope of the model is widened further as \cite{Lopes:Salazar:Gamerman} allowed the columns of $\bold{X}(\bold{s})$ to vary as random functions of locations. Although developed primarily for modeling Gaussian spatial time series data, LDSTMs are easily amenable to non-Gaussian data. Particularly, it can be modified as follows, so as to be suitable for spatio-temporal data that 
are associated with an exponential family model 
\begin{align*}
     Y(\bold{s},t)   \sim & f(y(\bold{s},t)|\mu_{t}(\bold{s}))\\ 
     \text{where}\ \ \  f(y(\bold{s},t)|\mu_{t}(\bold{s})) = & 
     \exp\left[\gamma\left\{y(\bold{s},t)\mu_{t}(\bold{s})-\zeta(\mu_{t}(\bold{s}))\right\}+h(y(\bold{s},t),\gamma)\right].
\end{align*}
Here $\zeta(\cdot)$ is a known, twice differentiable function, $h(\cdot,\cdot)$ is a known function and $\gamma$ is a scalar
dispersion parameter.

All these models are dynamically specified but the dynamics remain invariant across the space. Significant extension can be achieved by allowing the state variable $\boldsymbol{\beta}_{t}$ to vary across space. This obviously implies a substantial increase in the number of parameters and may lead to identifiability problem. Putting restrictions on the parameter space keeps the problem manageable. \cite{Paez:Gamerman:Landim:Salazar} proposed such a spatially varying coefficient dynamic model where they took $\boldsymbol{\beta}_{t}(\bold{s})=\bar{\boldsymbol{\gamma}}_{t}+\boldsymbol{\gamma}_{t}(\bold{s})$. The common trend $\bar{\boldsymbol{\gamma}}_{t}$ evolves dynamically with time but the spatio-temporal part $\boldsymbol{\gamma}_{t}(\bold{s})$ are 
iid Gaussian processes on space. \cite{Banerjee:Gamerman:Gelfand} considered a more general model in the 
environmental context where both $\bar{\boldsymbol{\gamma}}_{t}$ and $\boldsymbol{\gamma}_{t}(\bold{s})$ vary dynamically 
with respect to time. They used this model in 
the 
context of 
regression. \cite{Huerta:Sanso:Stroud} used a similar model additionally equipped with seasonal components to analyze a spatio-temporal ozone dataset. The spatially varying state variables for all these LDSTMs can be represented as
\begin{align*}
     &\boldsymbol{\beta}_{t}(\bold{s})=\bar{\boldsymbol{\gamma}}_{t}+\boldsymbol{\gamma}_{t}(\bold{s}),\ \ \  \text{for all}\  \bold{s};\\ 
      &\bar{\boldsymbol{\gamma}}_{t}=\bar{\bold{G}}_{t}\bar{\boldsymbol{\gamma}}_{t-1}+\bar{\boldsymbol{\omega}}_{t};\\
      &\boldsymbol{\gamma}_{t}(\bold{s})=\bold{G}_{t}\boldsymbol{\gamma}_{t-1}(\bold{s})+\boldsymbol{\omega}_{t}(\bold{s}),\ \ \  \text{for all}\  \bold{s}.
\end{align*}
However, a number of authors proposed an apparently different form for the evolution of $\boldsymbol{\beta}_{t}(\bold{s})$. They assumed that
\begin{align*}
     \beta_{t}(\bold{s})=\int K(\bold{u},\bold{s}){\beta}_{t-1}(\bold{u})d\bold{u}+\eta_{t}(\bold{s}),\ \ \  \text{for all}\  \bold{s},
\end{align*} where $K(\bold{u},\bold{s})$ is a redistribution kernel that determines how the state variables at the previous time point influence the state variables at present. Typically, spatially nearer observations get more weight. Although, it appears to be very different from the form considered earlier, a discrete convolution based approximation to the integral would yield that usual form. \cite{Wikle:Cressie} considered the above formulation in their dimension reduction approach and \cite{Storvik:Frigessi:Hirst} considered it in the context of temperature data modeling.

Model fitting for LDSTMs can be carried out either by the hierarchical Bayesian approach or by classical Kalman filtering. 
\cite{Wikle:Cressie,Huang:Cressie,Mardia:Goodall} used classical Kalman filter or its variations whereas 
\cite{Stroud:Muller:Sanso,Sanso:Guenni,Tonellato,Huerta:Sanso:Stroud,Banerjee:Gamerman:Gelfand} considered the 
hierarchical Bayesian approach. For more on LDSTMs an excellent reference would be \cite{Gelfand:Diggle:Fuentes:Guttorp}.

So far we have considered dynamic spatio-temporal models whose temporal evolution can be described by linear equations. However, real life environmental processes are complex and require much more sophistication in model specification. In particular, processes like precipitation, deposition, etc. are driven by complex interactions among atmospheric processes and are best represented by nonlinear models. An LDSTM model as considered above would be simply unsuitable in this situation. Apart from atmospheric processes, sometimes sea-surface temperature data is modeled by nonlinear dynamic spatio-temporal models (NLDSTMs) owing to the complex dynamics of sea waves, that influence it heavily. Also, many processes in the context of growth curve modeling exhibit state-dependent or density-dependent growth, e.g., $\partial Y/\partial t = Yg(Y;\boldsymbol{\theta})$ for some nonlinear growth function $g(\cdot)$ (e.g., logistic, Ricker, Beverton-Holt, etc.). In addition, many processes exhibit what is 
sometimes referred to as nonlinear advection, e.g., in one spatial dimension, $\partial Y/\partial t = Y\partial Y/\partial s_{1}$ (see \cite{Cressie10}). More general nonlinear dynamic spatio-temporal models are required to accommodate such processes, among others.
\cite{Sanso:Guenni} considered such a truncation based NLDSTM for modeling a Venezuelan rainfall dataset, given as follows:
\begin{align*}
     Y(\bold{s},t)=\begin{cases} 
                    {\boldsymbol{\beta}_{t}(\bold{s})}^{b_{t}}\ \ \  &\text{if}\  \boldsymbol{\beta}_{t}(\bold{s})>0;\\
                    0\ \ \  &\text{if}\  \boldsymbol{\beta}_{t}(\bold{s})\leq 0, 
                    \end{cases}
\end{align*} where $\boldsymbol{\beta}_{t}(\bold{s})$ is the state process and $b_{t}$ is a time varying parameter associated with the truncation equation. A more general formulation for this model can be written as (see equation 7.39 on page 380 of \cite{Cressie10})
\begin{align*}
 Y(\bold{s},t)=a_{t}(\bold{s})+h_{t}(\bold{s}){\boldsymbol{\beta}_{t}(\bold{s})}^{b_{t}(\bold{s})}+\epsilon_{t}(\bold{s}),\ \ \  \text{for all}\  \bold{s}, 
\end{align*} where $\epsilon_{t}(\bold{s})$ is a spatial-temporal noise process. The above model assumes that the observational equation that connects the observed process $Y(\bold{s},t)$ with the state or latent process $\boldsymbol{\beta}_{t}(\bold{s})$, is nonlinear. Instead of that, nonlinearity can be introduced into the evolutionary equation. In that case, we would have 
\begin{align*}
 \boldsymbol{\beta}_{t}(\bold{s})=\varPsi(\boldsymbol{\beta}_{t-1}(\bold{s}))+\eta_{t}(\bold{s}),
\end{align*} where $\varPsi(\cdot)$ is some appropriate nonlinear function. \cite{Cressie10} discussed nonlinear state-dependent models $\boldsymbol{\beta}_{t}(\bold{s})=\varPsi_{t}(\boldsymbol{\beta}_{t-1}(\bold{s}))\boldsymbol{\beta}_{t-1}(\bold{s})+\epsilon_{t}(\bold{s})$ which are time varying versions of it. Often the nonlinear function $\varPsi_{t}(\boldsymbol{\beta}_{t-1}(\bold{s}))\boldsymbol{\beta}_{t-1}(\bold{s})$ is taken to be a threshold function as follows:
\begin{align*}
 \boldsymbol{\beta}_{t}(\bold{s})=\varPsi_{t}(\boldsymbol{\beta}_{t-1}(\bold{s}))\boldsymbol{\beta}_{t-1}(\bold{s})+\eta_{t}(\bold{s})=
 \begin{cases}
    G_{1}\boldsymbol{\beta}_{t-1}(\bold{s})+\eta_{1,t}(\bold{s})\ \ \  &\text{if}\  f_{1}(\gamma_{t}) \in c_{1};\\ 
    \vdots\\
    G_{K}\boldsymbol{\beta}_{t-1}(\bold{s})+\eta_{K,t}(\bold{s})\ \ \  &\text{if}\  f_{K}(\gamma_{t}) \in c_{K},\\      
 \end{cases}
\end{align*} where $f_{k}(\gamma_{t})$ is a function of a time varying parameter $\gamma_{t}$, and 
$c_{k}; k=1,2,\cdots,K$, is the condition under which the $k$th equation is to be followed (see equation 7.69 on page 406 of \cite{Cressie10}). An example of such a model is given by \cite{Berliner:Wikle:Cressie} with regard to long-lead forecasting of tropical Pacific sea surface temperature. \cite{Hughes:Guttorp} used such a model in atmospheric application and \cite{Hooten:Wikle} employed them for the analysis of an ecological dataset. Another interesting class of NLDSTM, that \cite{Wikle:Hooten} (also see \cite{Cressie10}) referred to as the General Quadratic Nonlinear (GQN) Model, is given by
\begin{align*}
 \beta_{t}(\bold{s}_{i})=\sum_{j=1}^{n}a_{ij}\beta_{t-1}(\bold{s}_{j})+\sum_{k=1}^{n}\sum_{l=1}^{n}b_{i,kl}\beta_{t-1}(\bold{s}_{k})g(\beta_{t-1}(\bold{s}_{l});\boldsymbol{\theta}^G)+\eta_{t}(\bold{s}_{i});\ \ \  \text{for}\  i=1,2,\cdots,n,
\end{align*} where $\sum_{j=1}^{n}a_{ij}\beta_{t-1}(\bold{s}_{j})$ is a linear combination of the process at the previous time and \newline $\sum_{k=1}^{n}\sum_{l=1}^{n}b_{i,kl}\beta_{t-1}(\bold{s}_{k})g(\beta_{t-1}(\bold{s}_{l});\boldsymbol{\theta}^G)$ contains quadratic interactions of the process and potentially some transformation of the lagged process, at the previous time. 
The model is flexible enough to accommodate nonlinear transformations of the process through the function $g(\cdot)$, 
which might depend upon the unknown parameter vector $\boldsymbol{\theta}^G$. GQN constitutes a very rich class 
of models and many complex process models including the one considered by \cite{Berliner:Milliff:Wikle} 
for the so-called quasi-geostrophic flow in the ocean, are special cases of it. \cite{Wikle:Holan} considered an 
extension of GQN by incorporating higher order polynomial interaction terms in the following way:
\begin{align*}
 \beta_{t}(\bold{s}_{i})=&\sum_{j_{1}=1}^{n}a_{i,j_{1}}^{(1)}\beta_{t-1}(\bold{s}_{j_{1}})+\sum_{j_{2}=1}^{n}\sum_{j_{1}=1}^{n}a_{i,j_{1}j_{2}}^{(2)}\beta_{t-1}(\bold{s}_{j_{2}})g(\beta_{t-1}(\bold{s}_{j_{1}});\boldsymbol{\theta}_{g_{1}})+\eta_{t}(\bold{s}_{i})\\
 &+\sum_{j_{3}=1}^{n}\sum_{j_{2}=1}^{n}\sum_{j_{1}=1}^{n}a_{i,j_{1}j_{2}j_{3}}^{(3)}\beta_{t-1}(\bold{s}_{j_{3}})\beta_{t-1}(\bold{s}_{j_{2}})g(\beta_{t-1}(\bold{s}_{j_{1}});\boldsymbol{\theta}_{g_{2}})\\
 & \ \ \ \ \ \ \ \ \ \ \ \ \ \ \ \ \ \ \ \ \vdots\\
 &+\sum_{j_{p}=1}^{n}\cdots\sum_{j_{2}=1}^{n}\sum_{j_{1}=1}^{n}a_{i,j_{1}j_{2}\cdots j_{p}}^{(p)}\beta_{t-1}(\bold{s}_{j_{p}})\cdots \beta_{t-1}(\bold{s}_{j_{2}})g(\beta_{t-1}(\bold{s}_{j_{1}});\boldsymbol{\theta}_{g_{p}})\\
 &+\eta_{t}(\bold{s}_{i}).\\
\end{align*}
It is called the General Polynomial Nonlinear (GPN) model. For an excellent overview about such NLDSTMs the reader 
may look into chapter 7 of \cite{Cressie10}.

NLDSTMs provide us with a way out when the usual LDSTMs turn out to be too naive for the phenomena under study, but that too comes with a cost. The issues of dimensionality and efficient parametrization present the most significant challenges for statistical modeling of LDSTMs and these issues get even more critical for NLDSTM's. However, what is more daunting is that without 
very precise knowledge of the underlying dynamics, it is almost impossible to elicit an appropriate nonlinear model from a large class of probable nonlinear functions. A selected nonlinear model, that is unsuitable for the physical process under study, would show grossly poor predictive performance. Seemingly irrelevant departure from reality at the level of model specification cumulates over time and the outcome may be devastating.

\section{Gaussian Random Functional Dynamic Spatio-Temporal Model}
\label{sec:our_proposal}

As already discussed, a focal issue for the NLDSTMs is the selection of the form of nonlinearity, and 
such a task is highly non-trivial. We propose a dynamic spatio-temporal model where both the observational and evolutionary equations are random functions. Hence, we no longer have to decide about the specific functional forms; instead, all we need to do is to ensure that the probabilistic law for the random function is so chosen that it gives enough probability to a set of functions that seem potentially appropriate for the data at hand. Moreover, unlike the LDSTMs or NLDSTMs where the functional form is fixed, a random functional form is more adaptable to the data and expected to represent the true underlying process, which may be complex and highly nonlinear, more accurately. As the probabilistic law for those random functions, we specify 
appropriate Gaussian processes. The motivation for choosing Gaussian process comes from the fact that 
Gaussian processes are good natural priors for 
nonparametric regression and classification problems and under increasingly dense observations, the true shape of the arbitrary function or the classifier can be captured accurately, a posteriori \cite{Rasmussen:Williams}. Later we shall see how the proposed Gaussian random functional dynamic spatio-temporal model (we refer to our model as GRFDSTM) is connected to a Bayesian nonparametric function estimation problem. Besides, already there is a very rich mathematical theory for Gaussian process \cite{Adler07} which can be utilized to explore properties of GRFDSTM. Many elegant mathematical results, which hold true for Gaussian random functions, may be lost if we consider some other random functions, i.e., Levy random functions, Elliptic random functions etc. 

We assume that the observed spatio-temporal process $Y(\bold{s},t)$ (equivalently, $Y_{t}(\bold{s})$) 
is driven by an unobserved spatio-temporal state process $X(\bold{s},t)$ which itself is evolving in time.
Our model can be described in the following way:

\begin{align}
     Y(\bold{s},t)&=f(X(\bold{s},t))+\epsilon(\bold{s},t);
     \label{eqn:npr1}\\
     X(\bold{s},t)&=g(X(\bold{s},t-1))+\eta(\bold{s},t),
     \label{eqn:npr2}\\
     X(\cdot,0) & \sim \text{GP}(\mu_{0}(\cdot),c_{0}(\cdot,\cdot)); \ \ f(\cdot), g(\cdot) \sim \text{GRF}(\cdot,\cdot),
\end{align}
where $ \bold{s}\in \mathbb{R}^2 $ and $ t\in \{1,2,3,\ldots\}$. In the above, ``GP" stands for ``Gaussian process" and ``GRF" stands for ``Gaussian random function". Here $X(\cdot,0)$ is a spatial Gaussian process on $\mathbb{R}^2$;
$\epsilon(\cdot,t)$ and $\eta(\cdot,t)$ 
are temporally independent and identically distributed 
spatial Gaussian processes on $\mathbb{R}^2$, and $f(\cdot)$ and $g(\cdot)$ are Gaussian random functions on $\mathbb{R}$.
They are all independent of each other. Note that, in the LDSTMs and NLDSTMs discussed in the previous section, the state process is denoted by $\beta_{t}(\bold{s})$ or $\beta(\bold{s},t)$ and $\bold{X}(\bold{s},t)$ commonly refer to the covariate information, whereas for GRFDSTM we choose to denote the state process by $X(\bold{s},t)$, which resembles the notation used for state space models in classical time series literature.

The specification (\ref{eqn:npr2}) is very crucial since it determines the nonlinear dynamics of the model. In nature we encounter such nonlinear function $g(.)$ in population ecology where $Y(\bold{s},t)$ measures vegetation density at location $\bold{s}$ at time $t$. Then $Y(\bold{s},t)$ is linked to the population size $X(\bold{s},t)$ of a herbivorous species that feed on it, through some complex nonlinear function $f(.)$. The population size of the herbivorous species $X(\bold{s},t)$ on the other hand, dynamically changes according to a discretized spatial Lotka-Voltera type model. 

\begin{align*}
     X(\bold{s},t+\Delta t)&=X(\bold{s},t)+\Delta t(\alpha X(\bold{s},t)-\beta X(\bold{s},t)Z(\bold{s},t))\\
     Z(\bold{s},t+\Delta t)&=Z(\bold{s},t)+\Delta t(\delta X(\bold{s},t)Z(\bold{s},t)-\gamma Z(\bold{s},t))
\end{align*}
where $Z(\bold{s},t)$ denotes the population size of a carnivorous species that hunt the herbivorous species. So, in this case $X(\bold{s},t+\Delta t)=g(X(\bold{s},t))$ through a very complex nonlinear function $g(.)$ which is not even writable in closed form. Another instance of such nonlinear evolutionary transformation $g(.)$ can be found in oceanography where the sea surface temperature $Y(\bold{s},t)$ depends on the ocean stream function $X(\bold{s},t)$. The ocean stream function varies dynamically through some complex nonlinear equation referred to as quasi-geostrophic (QG) equation (see \cite{Berliner:Milliff:Wikle})
\begin{align*}
     (\nabla^{2}-\frac{1}{r^{2}})\frac{\partial X(\bold{s},t)}{\partial t}=-J(X(\bold{s},t),\nabla^{2}X(\bold{s},t))-\beta\frac{\partial X(\bold{s},t)}{\partial s_{1}}+\frac{1}{\rho H}curl_{s_{3}}\tau-\gamma\nabla^{2}X(\bold{s},t)+a_{h}\nabla^{4}X(\bold{s},t)
\end{align*}
Again, discretized version of this equation leads to a very complex nonlinear evolutionary transformation $g(.)$.

The GRFDSTM can be connected to a Bayesian nonparametric function estimation problem in the following way. Suppose, we want to study a spatio-temporal process $X(\bold{s},t)$. However, all that we have is data $Y(\bold{s},t)$, which is a masked version of $X(\bold{s},t)$. Such masked data are routinely encountered in the study of signals, images, videos, etc. Suppose that $X(\bold{s},t)$ is dynamically changing with time and the dynamics can be modeled by some complex, possibly nonlinear function $g(\cdot)$. Then continuity is the minimal requirement for such a function $g(\cdot)$. We further assume that the masking is also described by some continuous transformation $f(\cdot)$. Now, the problem facing the statistician is to estimate the two unknown functions (or alternatively infinite dimensional parameters) and recover the true signal $X(\bold{s},t)$. To solve the problem in the 
Bayesian way we must elicit priors on the infinite dimensional function space of all real-valued continuous functions on the real line, i.e., $C(\mathbb{R})$. Suppose, we specify a Gaussian process prior on the parameter space $C(\mathbb{R})$. 
Then, if we consider the conditional distribution of $[Y(\bold{s},t),X(\bold{s},t)|\boldsymbol{\theta}]$, 
where $\boldsymbol{\theta}$ is the hyper-parameter associated with the Gaussian process priors, it is the same as the GRFDSTM. 
Although the purpose of this hierarchical Bayesian model for function estimation is to identify $f(\cdot)$ and $g(\cdot)$ 
and recover the original signal $X(\bold{s},t)$, which is quite different than the purpose the GRFDSTM serves, i.e., 
prediction possibly at a new spatio-temporal location $(\bold{s}^*,t^*)$, the following duality between the two models
\begin{align*}
 \hat{f}(\cdot), \hat{g}(\cdot) \ \ \text{and}\ \  \hat{X}(\bold{s},t)   \ \ \ \ \ \ \ \ \ \ \ \ \ \ \ \ \ \ \ \ \ \ \ \ \ \  & \ \mathlarger{\Leftrightarrow}  \ \ \ \ \ \ \ \ \ \ \ \ \ \ \ \ \ \ \ \ \ \ \ \ \  \ \ \   \hat{Y}(\bold{s}^*,t^*)\\
 \text{[good estimators of $f(\cdot),g(\cdot)$ and $X(\bold{s},t)$} \ \ \ \ \ \ \ \ \ \ \ \ & \ \ \ \ \ \ \ \ \ \ \ \ \text{[good predictors of $Y(\bold{s}^*,t^*)$ at the spatio-}\\
 \text{for the hierarchical Bayesian model]} \ \ \ \ \ \ \ \ \ \ \ \ \ \ \ \ & \ \ \ \ \ \ \ \ \ \ \ \ \text{temporal location $(\bold{s}^*,t^*)$ for the GRFDSTM]}
\end{align*} 
gives us the intuition that a GRFDSTM might work well.

Finally, to completely specify the GRFDSTM, we need to describe the parameters associated with the Gaussian random functions and the spatial Gaussian processes. 
We assume that the Gaussian random function $f(x)$ has the mean function of the form $\beta_{0f}+\beta_{1f}x$ (where
$\beta_{0f}, \beta_{1f}$ are suitable parameters) and has isotropic covariance kernel of the form $c_{f}(x_{1},x_{2})=\gamma(\|x_1-x_2\|)$,
where $\gamma$ is a positive definite function. 
It again consists of parameters that determine the smoothness of the sample paths of $f(\cdot)$. Moreover, $\gamma$ is such that the centered Gaussian random function with covariance kernel $\gamma$ has continuous sample paths. This assumption of 
sample path continuity is required due to some technical reasons and is very minimal in the sense that all popular isotropic covariance kernels satisfy it. Typical examples of $\gamma$ are exponential, powered exponential, Gaussian, Mat\'{e}rn etc. (see, for example, Table 2.1 of \cite{Banerjee04} for other examples of such covariance kernels).
%
Similarly, parameters $\beta_{0g}, \beta_{1g}$ and $c_{g}(x_{1},x_{2})$ are associated with the Gaussian random function $g(x)$. 

The zero mean spatial Gaussian processes $\epsilon(\bold{s},t)$ and $\eta(\bold{s},t)$ have covariance kernels $c_{\epsilon}(\bold{s},\bold{s}^{\prime})$ 
and $c_{\eta}(\bold{s},\bold{s}^{\prime})$ respectively, which are also of similar form. 
 
Regarding the spatial Gaussian process associated with $X(\cdot,0)$, we assume a continuous mean process of the form $\mu_{0}(\cdot)$ and 
isotropic covariance kernel $c_{0}(\cdot,\cdot)$. For convenience, we introduce separate notations for the mean vector and the 
covariance matrix associated with $(X(\bold{s}_{1},0),X(\bold{s}_{2},0),\cdots X(\bold{s}_{n},0))$, where 
$\bold{s}_{1}, \bold{s}_{2}, \bold{s}_{3}\cdots, \bold{s}_{n}$ are the spatial locations where the data is observed. 
We denote them by ${\boldsymbol{\mu}}_{0}$ and ${\boldsymbol{\Sigma}}_{0}$ respectively.

From (\ref{eqn:npr1}) and (\ref{eqn:npr2}) it is not difficult to see that our model boils down to a simple 
LDSTM with no covariate if the process variance associated with the Gaussian random functions (denoted by $\sigma_{f}^{2}$ and $\sigma_{g}^{2}$) become 0. The model equations (\ref{eqn:npr1}) and (\ref{eqn:npr2}) then reduce to the following form
\begin{align}
     Y(\bold{s},t) &=\beta_{0f}+\beta_{1f}X(\bold{s},t)+\epsilon(\bold{s},t);
     \label{eqn:DLM1}\\
     X(\bold{s},t) &=\beta_{0g}+\beta_{1g}X(\bold{s},t-1)+\eta(\bold{s},t);
     \label{eqn:DLM2}\\
     X(\cdot,0) & \sim \text{GP}(\mu_{0}(\cdot),c_{0}(\cdot,\cdot)),
\end{align}
where $ \bold{s}\in \mathbb{R}^2 $ and $ t\in\{1,2,3,\ldots\}$. Although we develop the GRFDSTM for equispaced time points, 
simple modification of this model can handle non equispaced time points as well. 
However, for the sake of 
simplicity and brevity, in this article, we shall consider only equispaced time points. Lastly, it is also possible to consider nonlinear mean function for $f(x)$ and $g(x)$. This formulation would work well, provided one already has strong prior knowledge about the form of the nonlinear dynamics behind the process. However, under little or no knowledge about the process, we strongly suggest sticking to the linear mean functions. 

The GRFDSTM is in spirit very similar to spatio-temporal generalized additive models (STGAM) \cite{Giannitrapani:Bowman:Scott:Smith,Holland:Oliveira:Cox:Smith}. A typical STGAM looks like the following
\begin{align*}
 Y(\bold{s},t) =\mu+m_{\bold{s}}(s_{1},s_{2})+m_{t}(t)+\epsilon(\bold{s},t).
\end{align*}
Like the GRFDSTM, the STGAM also model the data without assuming linear or specific nonlinear form for $m_{\bold{s}}(.)$ and $m_{t}(.)$, thereby modeling the spatial and temporal trend nonparametrically. However, the main difference between the GRFDSTM and the STGAM is that the former is a conditional approach and assumes flexible nonparametric form for the dynamic (conditional) structure of the space-time model whereas the latter is a marginal approach and stresses flexible modeling of the trend in space and time and the interaction (marginal and joint) terms.

\subsection{Some Measurability and Existential Issues}

Before we proceed to explore the properties of GRFDSTM, we need to ensure that a family of valid (measurable) spatio-temporal 
stochastic processes is induced by the proposed model. Only then it can be used as a statistical model for real world physical processes. 
In general, such measurability issues are almost always trivially satisfied by statistical models and so, never discussed in detail. But in this case, we 
need to show that $f(X(\bold{s}_{1},t)),f(X(\bold{s}_{2},t)),\cdots,f(X(\bold{s}_{n},t))$ are jointly measurable for any $n$ 
and any set of spatial locations $\bold{s}_{1},\bold{s}_{2},\cdots,\bold{s}_{n}$, and this is not a trivial problem. 
The difficulty is that when $f$ and $X$ both are random, $f(X)$ need not be a measurable or valid random variable. 
It is the sample path continuity of $f(\cdot)$, which compels $f(X(\bold{s},t))$ to be measurable.

\begin{proposition}
\label{propn:measurable}
The \emph{GRFDSTM} defines a family of valid (measurable) spatio-temporal processes on $\mathbb R^{2}\times\mathbb Z^+$.
\end{proposition}

Once it is ensured that the proposed model induces a family of valid spatio-temporal processes, we look 
into some other important aspects like the joint distributions of state variables and observed variables, covariance structure 
of the observed process etc. Note that although we develop the GRFDSTM assuming that the spatial dimension is 
$2$, the construction and all the subsequent results go through for $\mathbb{R}^d$ $(d>2)$ as well.

\subsection{Joint Distribution of the Variables}
Due to the implicit hierarchical structure of the GRFDSTM, the joint distribution of the observed variables is non-Gaussian. But even before considering that, we want to find out the joint distribution of the state variables. A closed form joint pdf for the state variables is necessary for MCMC based posterior inference.

\begin{theorem}
\label{thm:state}Suppose that the spatio-temporal process is observed at locations 
$\bold{s}_{1}, \bold{s}_{2}, \bold{s}_{3}\cdots, \bold{s}_{n}$ for times $t=1, 2, 3,\cdots, T $.
Then the joint distribution of the state variables is non-Gaussian and has the pdf
\vspace{3mm}

$\mathlarger{\mathlarger{\frac{1}{(2\pi)^\frac{n}{2}}\frac{1}{|{\mathbf{\Sigma}}_{0}|^\frac{1}{2}}}}\exp\left[-\frac{1}{2}{\begin{pmatrix}x(\bold{s}_{1},0)-\mu_{01}\\x(\bold{s}_{2},0)-\mu_{02}\\ \vdots\\x(\bold{s}_{n},0)-\mu_{0n}\end{pmatrix}}^{\prime}{{\mathbf{\Sigma}}_{0}}^{-1}{\begin{pmatrix}x(\bold{s}_{1},0)-\mu_{01}\\x(\bold{s}_{2},0)-\mu_{02}\\ \vdots\\x(\bold{s}_{n},0)-\mu_{0n}\end{pmatrix}}\right] \mathlarger{\mathlarger{\frac{1}{(2\pi)^\frac{nT}{2}}\frac{1}{|\tilde{\mathbf{\Sigma}}|^\frac{1}{2}}}}\times$

\vspace{3mm}

$\exp\left[-\frac{1}{2}{\begin{pmatrix}x(\bold{s}_{1},1)-\beta_{0g}-\beta_{1g}
x(\bold{s}_{1},0)\\x(\bold{s}_{2},1)-\beta_{0g}-\beta_{1g}
x(\bold{s}_{2},0)\\ \vdots\\x(\bold{s}_{n},T)-\beta_{0g}-\beta_{1g}
x(\bold{s}_{n},T-1)\end{pmatrix}}^{\prime}{\tilde{\mathbf{\Sigma}}}^{-1}{\begin{pmatrix}x(\bold{s}_{1},1)-\beta_{0g}-\beta_{1g}
x(\bold{s}_{1},0)\\x(\bold{s}_{2},1)-\beta_{0g}-\beta_{1g}
x(\bold{s}_{2},0)\\ \vdots\\x(\bold{s}_{n},T)-\beta_{0g}-\beta_{1g}
x(\bold{s}_{n},T-1)\end{pmatrix}}\right], $

\vspace{3mm}

where ${\boldsymbol{\mu}}_{0}=(\mu_{01},\mu_{02},\cdots,\mu_{0n})^\prime$ and ${\mathbf{\Sigma}}_{0}$ 
are already defined to be the mean vector and the covariance matrix of
$(X(\bold{s}_1, 0), X(\bold{s}_2, 0),\ldots, X(\bold{s}_n , 0))$,
and

\vspace{3mm}
\begin{center}
$\tilde{\mathbf{\Sigma}}=\begin{pmatrix}1&0&\cdots &0\\0&1&\cdots &0\\ \vdots \\0&0&\cdots &1\\ \end{pmatrix}\bigotimes{\mathbf{\Sigma}}_{\eta}+\mathbf{\Sigma},$
\vspace{3mm}\\
\end{center}
where the elements of ${\mathbf{\Sigma}}_{\eta}$ are obtained from the purely spatial covariance function $c_{\eta}$ and the elements of $\mathbf{\Sigma}$ are obtained from the covariance function $c_{g}$ in the following way : \\
\\
the $(i,j)$ th entry of ${\mathbf{\Sigma}}_{\eta}$ is $c_{\eta}(\bold{s}_{i},\bold{s}_{j})$ and the $((t_{1}-1)n+i,(t_{2}-1)n+j)$ th entry of $\mathbf{\Sigma}$ is $c_{g}(x(\bold{s}_{i},t_{1}-1),x(\bold{s}_{j},t_{2}-1))$ where $1 \leq t_{1},t_{2} \leq T$ and $1\leq i,j \leq n$ .
\end{theorem}

Although the appearance of the probability density function resembles that of a multivariate Gaussian density, the involvement of 
$x(\bold{s}_{i},t)$ in $\tilde{\mathbf{\Sigma}}$ renders it non-Gaussian. In the extreme case when the process 
variance $c_{g}(0,0)$ ($=\sigma_{g}^2$) of the Gaussian random function $g(\cdot)$ is 0, $\tilde{\mathbf{\Sigma}}$ 
becomes a block diagonal matrix with identical blocks and the joint density becomes Gaussian. In the formation of $\tilde{\mathbf{\Sigma}}$, the ${\mathbf{I}}_{T\times T}\bigotimes\mathbf{\Sigma_{\eta}}$ part corresponds to linear evolution and Gaussianity whereas the 
component $\mathbf{\Sigma}$ corresponds to departure from linearity. In fact, it is also responsible for making the pdf a non-Gaussian one.

Moreover, it is also clear from the form of the density function that the temporal aspect is imposed on the model through both the location function and the scale function associated with the latent process, making the GRFDSTM a very flexible spatio-temporal model. 

The interesting property that the observed spatio-temporal process is also non-Gaussian, is a consequence of both of the facts that the state variables are non-Gaussian and the GRFDSTM has an implicit hierarchical structure. We have the following
theorem in this regard:

\begin{theorem}
\label{thm:observe}
Suppose that the spatio-temporal process is observed at locations $\bold{s}_{1}, \bold{s}_{2}, \bold{s}_{3}\cdots, \bold{s}_{n}$ for times $t=1, 2, 3,\cdots, T $.
Then the following hold true:
\vspace{3mm}\\
(a) The joint distribution of the observed variables is a Gaussian mixture and has the following density

\vspace{3mm}

$\mathlarger{\mathlarger{\int}_{\mathbb R^{nT}}\frac{1}{(2\pi)^\frac{nT}{2}}\frac{1}{|{\mathbf\Sigma}_{f,\epsilon}|^\frac{1}{2}}}\times $

\vspace{3mm}

$\exp\left[-\frac{1}{2}{\begin{pmatrix}y(\bold{s}_{1},1)-\beta_{0f}-\beta_{1f}x(\bold{s}_{1},1)\\y(\bold{s}_{2},1)-\beta_{0f}-\beta_{1f}x(\bold{s}_{2},1)\\ \vdots\\y(\bold{s}_{n},T)-\beta_{0f}-\beta_{1f}x(\bold{s}_{n},T)\end{pmatrix}}^{\prime}{{\mathbf\Sigma}_{f,\epsilon}}^{-1}{\begin{pmatrix}y(\bold{s}_{1},1)-\beta_{0f}-\beta_{1f}x(\bold{s}_{1},1)\\y(\bold{s}_{2},1)-\beta_{0f}-\beta_{1f}x(\bold{s}_{2},1)\\ \vdots\\y(\bold{s}_{n},T)-\beta_{0f}-\beta_{1f}x(\bold{s}_{n},T)\end{pmatrix}}\right]\mathlarger{ h(\mathbf{x}) \,d\mathbf{x}} $
\vspace{3mm}\\
where the mixing density $h(\mathbf{x})$ is obtained by marginalizing the pdf derived at Theorem \ref{thm:state} with respect to 
the variables $x(\bold{s}_{1},0),x(\bold{s}_{2},0),\cdots,x(\bold{s}_{n},0)$ and the $((t_{1}-1)n+i,(t_{2}-1)n+j)$ 
th entry of $\mathbf{\Sigma}_{f,\epsilon}$ is given by $c_{f}(x(\bold{s}_i,t_1),x(\bold{s}_j,t_2))+c_{\epsilon}(\bold{s}_i,\bold{s}_j)\delta(t_1-t_2)$ where $1\leq t_{1},t_{2},\leq T$ and $1\leq i,j\leq n$.
\vspace{3mm}\\
(b) In the extreme case, when the process variance $c_{f}(0,0)$ ($=\sigma_{f}^2$) and $c_{g}(0,0)$ ($=\sigma_{g}^2$) 
of each of the Gaussian random function $f(\cdot)$ and $g(\cdot)$ are 0, 
the joint distribution turns into Gaussian.
\end{theorem}
The implicit hierarchical structure implies that the observed spatio-temporal process is a mixture of Gaussian processes which subsequently implies that GRFDSTM can flexibly accommodate both Gaussian and non-Gaussian distributions for spatio-temporal data. In fact several authors have already used the mixture distribution approach to produce non-Gaussian models for spatial and spatio-temporal data \cite{Fonseca:Steel,Palacios:Steel}.  

\subsection{Dependence and Covariance Structure} As is evident from Theorem \ref{thm:observe}, the observed process $Y(\bold{s},t)$ is non-Gaussian for all but a few special cases. Hence, the covariance function no longer characterizes the dependence structure completely. However, the covariance function, which characterizes the linear dependence structure, still might give us valuable insight about the process, particularly in the case of small departure from linearity ($\sigma_{f}^2$ and $\sigma_{g}^2$ are small) and so it is worthwhile to take a deeper look into that.
However, even before delving into a deeper study of the covariance function and some related issues like nonstationarity 
and nonseparability, a more basic question is whether the process is light-tailed, that is, whether 
or not all the coordinate variables have finite variance. 
From Theorem \ref{thm:observe} we see that the observed variables are distributed as 
a Gaussian mixture and Gaussian mixtures sometimes can give rise to heavy-tailed distributions. 
So, an answer to the above question is not immediately available. In what follows, we show that the process $Y(\bold{s},t)$ is light-tailed and the covariance function is nonstationary and nonseparable.

\begin{theorem}
\label{thm:covariance}
(a) The observed spatio-temporal process $Y(\bold{s},t)$ is light-tailed in the sense that all the coordinate variables have finite variance.\\
(b) The covariance function $c_{y}((\bold{s},t),(\bold{s^*},t^*))$ of the observed spatio-temporal process is
nonstationary and nonseparable for any $s,s^*,t,t^*$.
\end{theorem}
Although the GRFDSTM yields a nonstationary and nonseparable covariance function, an obvious limitation is that no closed form expression for the covariance function is available. This problem, 
however, exists in many other spatio-temporal models including the convolution models, deformation models, complex DSTM models, etc., where a closed form covariance function is available only in very few special cases. In fact, the direct construction of spatio-temporal covariance function is the only approach that is always guaranteed to yield a closed form covariance function.

Besides, from the prediction point of view, this is not a serious problem since all that we need are the posterior 
predictive distributions at unmonitored locations at arbitrary points of time, which do not require 
closed form expression of the covariance function.

However, we still strive to find some closed form expression and are partially successful in the sense 
that when our model is approximately linear (in some suitable sense to be described later), the covariance function is approximately a geometric function of time lag.

\subsubsection{Approximate Form of the Covariance Function} 
Here, in Theorem \ref{thm:closed_covariance}, we show that if the process variances ($\sigma_{g}^2$) and ($\sigma_{f}^2$) are small, then under a minor assumption, the covariance function $c_{y}((\bold{s},t),(\bold{s^*},t^*))$ is approximately a geometric function of time lag. 
This result is mainly of theoretical interest.

\begin{theorem}
\label{thm:closed_covariance} 
Assume that $|\beta_{1g}|<1$. Then for given $\epsilon^{\prime\prime}>0$ arbitrarily small, 
$\exists~\delta>0$ such that for $0<\sigma_{g}^{2},\sigma_{f}^{2}<\delta$ the covariance between 
$Y(\bold{s},t)$ and $Y(\bold{s^*},t^*)$, denoted by 
$c_{y}((\bold{s},t),(\bold{s^*},t^*))$, is of following form:

\begin{align*}
\beta_{1g}^{|t-t^*|}\left[c_{0}(\bold{s},\bold{s}^*)+\left[\frac{1-\beta_{1g}^{2(t^*+1)}}{1-\beta_{1g}^2}\right]
c_{\eta}(\bold{s},\bold{s}^*)\right]-\epsilon^{\prime\prime}
&\leq c_{y}((\bold{s},t),(\bold{s^*},t^*))\\
&\leq \beta_{1g}^{|t-t^*|}\left[c_{0}(\bold{s},\bold{s}^*)+\left[\frac{1-\beta_{1g}^{2(t^*+1)}}{1-\beta_{1g}^2}\right]
c_{\eta}(\bold{s},\bold{s}^*)\right]+\epsilon^{\prime\prime}.
\end{align*}
\end{theorem}
The assumption $|\beta_{1g}|<1$ restricts the latent process $X(\bold{s},t)$ within the class of nonexplosive spatial AR models. Although this class is fairly large, it misses some interesting spatio-temporal processes like the spatial random walk. Spatial random walk, which is temporally nonstationary, is used extensively in econometric applications.

Note that no such closed form expression is available if we consider the behaviour of 
$c_{y}((\bold{s},t),(\bold{s^*},t^*))$ with respect to increasing spatial lag, i.e., 
$\|\bold{s}-\bold{s}^*\| \rightarrow \infty$. Empirical simulations, however, suggest that  
$c_{y}((\bold{s},t),(\bold{s^*},t^*))$ decays to $0$ as $\|\bold{s}-\bold{s}^*\| \rightarrow \infty$ and the rate of decay depends on the specific form of the covariance kernels used in the GRFDSTM. Specifically, if we use the squared exponential covariance kernels for all the associated Gaussian random functions and processes, then $c_{y}((\bold{s},t),(\bold{s^*},t^*))$ exhibits short range dependence and decays very fast to $0$ as $\|\bold{s}-\bold{s}^*\| \rightarrow \infty$.

\subsection{Sample Path Properties} 
So far, we have discussed the finite dimensional properties of the 
spatio-temporal process $Y(\bold{s},t)$. But finite dimensional properties alone are not 
sufficient to characterize any arbitrary general stochastic process. Two stochastic processes 
with completely different sample path behavior may have identical finite dimensional 
distributions and properties. We demonstrate this through the following simple example (see also \cite{Adler07}):

Let us consider two spatio-temporal processes $Y(\bold{s},t)(\omega)$ and $Y^{*}(\bold{s},t)(\omega)$ 
defined on the same probability space $\Omega=[0,1]^2$ in the following way:
$Y(\bold{s},t)(\omega)=0$ for all $\bold{s},t,\omega$, and
$Y^{*}(\bold{s},t)(\omega)=1$ for all $\bold{s}=\omega$ and $=0$ otherwise. 
Then one can show that for any fixed $t$, $Y(\bold{s},t)(\omega)$ has continuous sample path 
with probability $1$ whereas $Y^{*}(\bold{s},t)(\omega)$ has discontinuous sample path with probability $1$. However, both 
$Y(\bold{s},t)(\omega)$ and $Y^{*}(\bold{s},t)(\omega)$ have exactly the same finite-dimensional distributions.

In the light of the above, we decide to explore the path properties of $Y(\bold{s},t)$. The first part of 
the following theorem states that $Y(\bold{s},t)$ has continuous 
sample paths and the second part says that moreover, under additional 
smoothness assumptions regarding the covariance functions, it will have smooth sample paths.

\begin{theorem}
\label{thm:sample_path}
(a) The spatio-temporal process $Y(\bold{s},t)$ has continuous sample paths.
\vspace{2mm}\\
(b) Assume that the covariance functions $c_{f}(\cdot,\cdot),c_{g}(\cdot,\cdot),
c_{\epsilon}(\cdot,\cdot),c_{\eta}(\cdot,\cdot),c_{0}(\cdot,\cdot)$ satisfy the 
additional smoothness assumption that the centered Gaussian processes with these covariance 
functions have $k$ times differentiable sample paths. Then the non-Gaussian
spatio-temporal process $Y(\bold{s},t)$ also have $k$ times differentiable sample paths. 
\end{theorem}

So, the spatial surface generated by the GRFDSTM at any time point $t$ is continuous and unless the spatial surface interpolating the data points is extremely jagged indicating multiple points of discontinuity, any spatio-temporal data can be modeled reasonably adequately by the GRFDSTM.
A stronger statement, however, is made in the second part of the theorem. It says that if 
the covariance functions $c_{f}(\cdot,\cdot),c_{g}(\cdot,\cdot),c_{\epsilon}(\cdot,\cdot),
c_{\eta}(\cdot,\cdot),c_{0}(\cdot,\cdot)$ are sufficiently smooth then the sample paths 
of the process $Y(\bold{s},t)$ are also smooth and their degrees of smoothness depend on the degree of smoothness of the covariance functions. Immediate to the above theorem we have the following corollary associated with two very popular classes of covariance functions.

\begin{corollary}
 \label{corr:sample_path}
 (a) If all the covariance kernels associated with GRFDSTM, belong to the Mat\'{e}rn family whose smoothness parameter 
 is $\nu$, then $Y(\bold{s},t)$ will have  $\lceil \nu-1 \rceil$ many time differentiable sample paths.
Here, for any $x$, $\lceil x\rceil$ denotes the smallest integer greater than or equal to $x$. 
 \vspace{2mm}\\
 (b) If all the covariance kernels associated with GRFDSTM, are chosen to be squared exponential, 
 then $Y(\bold{s},t)$ will have infinitely many times differentiable sample paths. 
\end{corollary}

It is a well known fact that the sample paths of a centered Gaussian process, whose covariance function is Mat\'{e}rn with smoothness parameter $\nu$, are $\lceil \nu-1 \rceil$ times differentiable (see page 23 of \cite{Gelfand:Diggle:Fuentes:Guttorp}). 

Hence, part (a) of Corollary \ref{corr:sample_path} follows. Part (b) follows from the facts that squared exponential covariance function is essentially Mat\'{e}rn with smoothness parameter $\nu \rightarrow \infty$ and sample paths of a centered Gaussian process with squared exponential covariance function are infinitely many times differentiable. Hence, when we have strong evidence from the data or prior knowledge that the data generating spatial surface is neither too rough nor too smooth, and also have some idea regarding the degree of smoothness of the spatio-temporal process, then we may choose all the covariance functions $c_{f}(\cdot,\cdot),c_{g}(\cdot,\cdot),
c_{\epsilon}(\cdot,\cdot),c_{\eta}(\cdot,\cdot),
c_{0}(\cdot,\cdot)$ from the Mat\'{e}rn family with some appropriate specific value of $\nu$.

\section{Identifiability Issues and Sharpening the Model Description}
\label{section:Identifiability}
The GRFDSTM is more like an umbrella term used for a general modeling strategy, rather than a single model. Hence, a deeper investigation of the identifiability issue requires more specific model description. First recall that, associated with the mean function of the GRFs $f(.)$ and $g(.)$, we have four parameters $\beta_{0g},\beta_{1g},\beta_{0f},\beta_{1f}$. The vector parameter ${\boldsymbol{\mu}}_{0}$ is associated with the initial state process. The dependence structure is specified through some isotropic covariance kernels. Note that, till now we have not specified any particular form for the isotropic covariance kernel. However, to address identifiability we need to fix the covariance functions. Although any reasonable isotropic covariance kernel that satisfies the mild regularity conditions mentioned in Section \ref{sec:our_proposal} can be used in the GRFDSTM, for the sake of simplicity 
we consider the squared exponential covariance kernel with the representation 
$c(\bold{u},\bold{v})=\sigma^2 e^{-\lambda||\bold{u}-\bold{v}||^2}$. Associated with five covariance 
kernels  $c_{f}(\cdot,\cdot),c_{g}(\cdot,\cdot),c_{\epsilon}(\cdot,\cdot),c_{\eta}(\cdot,\cdot),
c_{0}(\cdot,\cdot)$ we have five scale parameters $\sigma_{f}^2,\sigma_{g}^2,
\sigma_{\epsilon}^2,\sigma_{\eta}^2,\sigma_{0}^2$ and five smoothness parameters $\lambda_{f},\lambda_{g},
\lambda_{\epsilon},\lambda_{\eta},\lambda_{0}$. Among them we fix the values of $\lambda_{f},\lambda_{g}$. Why would we choose to fix the value of $\lambda_{f},\lambda_{g}$ keeping $\lambda_{\epsilon},\lambda_{\eta}$ free to vary? The reason behind that is, although all four of them are smoothness parameters associated with squared exponential covariance kernels, they play entirely different roles in the model. Note that $\lambda_{\epsilon},\lambda_{\eta}$ determine the spatial variation of $Y(\bold{s},t)$ and its effective range. 
On the other hand, $\lambda_{f},\lambda_{g}$ are needed essentially for specifying Gaussian processes that are supported on 
substantially large classes of continuous functions, and fixing the values of $\lambda_{f},\lambda_{g}$ doesn't 
restrict the scope of that. More importantly, fixing their values help us get rid of 
certain identifiability problems. To illustrate, let us consider the following proposition:

\begin{proposition}
\label{propn:identifiability}
 Consider the following two sets of parameter values for GRFDSTM : 
 \begin{center}
 $\boldsymbol{\theta}_{1}=[\beta_{0f},\beta_{1f},\beta_{0g},\beta_{1g},{\boldsymbol{\mu}}_{0},\sigma_{f}^2,\sigma_{g}^2,\sigma_{\epsilon}^2,\sigma_{\eta}^2,\sigma_{0}^2,\lambda_{f},\lambda_{g},
\lambda_{\epsilon},\lambda_{\eta},\lambda_{0}]$\hspace{2cm}\mbox{and}
\end{center}
\begin{center}
$\boldsymbol{\theta}_{2}=[\beta_{0f},\frac{\beta_{1f}}{c},c\beta_{0g},\beta_{1g},{c\boldsymbol{\mu}}_{0},\sigma_{f}^2,c^2\sigma_{g}^2,\sigma_{\epsilon}^2,c^2\sigma_{\eta}^2,c^2\sigma_{0}^2,\frac{\lambda_{f}}{c^2},\frac{\lambda_{g}}{c^2},\lambda_{\epsilon},\lambda_{\eta},\lambda_{0}]$ \ \ for any $c\neq 0$.
\end{center}
\begin{center}
Then $[Y(\bold{s}_{1},1),\cdots, Y(\bold{s}_{n},T)|\boldsymbol{\theta}_{1}] \stackrel{d}{=} [Y(\bold{s}_{1},1),\cdots, Y(\bold{s}_{n},T)|\boldsymbol{\theta}_{2}]$.
\end{center}
\end{proposition}

An easy way to break off the above identifiability problem is to fix the values of $\lambda_{f},\lambda_{g}$. We also fix the values of $\sigma_{0}^2$ and $\lambda_{0}$. The reason is that as $t$ gets larger, 
the effects of $\sigma_{0}^2$ and $\lambda_{0}$ fade away, making the data $Y(\bold{s},t)$ 
much less informative about $\sigma_{0}^2$ and $\lambda_{0}$ compared to the parameters like 
$\sigma_{\epsilon}^2,\sigma_{\eta}^2$, etc. Unfortunately, identifiability can not be mathematically established under these mild restrictions on the parameter space. Now we state some stronger restrictions on the parameter space that will be sufficient for ensuring identifiability.
\begin{theorem}
\label{thm:mathematical_identifiability}
 Consider the following restrictions on the parameter space : \\ 
 (A) Suppose $\beta_{1f},\beta_{1g},\sigma_{f}^2,\sigma_{g}^2,\sigma_{\epsilon}^2,\sigma_{\eta}^2,\sigma_{0}^2 \neq 0 $ and $\lambda_{0} \neq \lambda_{\epsilon}$.\\
 (B) Moreover, we need to fix values of some parameters: assume $\lambda_{f},\lambda_{g},\sigma_{f}^2,\sigma_{g}^2,\lambda_{0},\sigma_{0}^2,\lambda_{\eta},\sigma_{\eta}^2,\beta_{0f},\beta_{1f},{\boldsymbol{\mu}}_{0}$ are fixed.\\
 (C) Assumption on the spatio-temporal sampling design : \\
 Assume $n \geq 3$ and $T\geq 1$ and existence of at least three distinct values $d_{1},d_{2},d_{3}$ of the sampling interpoint distances; i.e. $\exists \ \bold{s}_{i_{1}},\bold{s}_{j_{1}},\bold{s}_{i_{2}},\bold{s}_{j_{2}},\bold{s}_{i_{3}},\bold{s}_{j_{3}}$ such that $d_{1}=||\bold{s}_{i_{1}}-\bold{s}_{j_{1}}||$, $d_{2}=||\bold{s}_{i_{2}}-\bold{s}_{j_{2}}||$ and $d_{3}=||\bold{s}_{i_{3}}-\bold{s}_{j_{3}}||$. \\
 Also, assume $\exists \ \bold{s}_{i}$ and $\bold{s}_{j}$ such that $\mu_{0}(\bold{s}_{i}) \neq \mu_{0}(\bold{s}_{j})$.

\begin{center}
Then the remaining parameters i.e. $\beta_{0g},\beta_{1g},\lambda_{\epsilon},\sigma_{\epsilon}^2$ are jointly identifiable.
\end{center}
\end{theorem}
Although, the restrictions stated above would considerably narrow down the flexibility of the GRFDSTM, not everything is lost. The parameters $\lambda_{\epsilon},\sigma_{\epsilon}^2$ still induce a flexible spatial structure in the model and the parameters $\beta_{0g},\beta_{1g}$ flexibly control the mean behaviour of the temporal dynamics. However, going beyond theory, we would stick to the GRFDSTM with milder restrictions (i.e. $\sigma_{0}^2,\lambda_{0},\lambda_{f},\lambda_{g}$ are fixed) for the modeling purpose. In fact, Theorem \ref{thm:mathematical_identifiability} is just a sufficient condition for identifiability. We believe that to ensure identifiability, we don't need so many restrictions on the parameter space and assuming $\sigma_{0}^2,\lambda_{0},\lambda_{f},\lambda_{g}$ to be fixed, is enough. This is also evident from the simulation studies and real data analysis where under the assumption that $\sigma_{0}^2,\lambda_{0},\lambda_{f},\lambda_{g}$ are fixed, the posteriors of all the remaining 
parameters are unimodal, essentially indicating identifiability. Moreover,
 the priors used are also beneficial in this regard as they ensure that potentially undesirable region of the parameter space like $\sigma_{\epsilon}^2,\sigma_{\eta}^2 = 0$, etc. get very less probability so that the posterior probability of those regions of parameter space (i.e. regions like $\sigma_{\epsilon}^2=0,\sigma_{\eta}^2 = 0$, etc.) $\approx 0$. Unfortunately, mathematically showing it requires dealing with highly complex system of nonlinear equations and we fall short of a rigorous proof.

\section{Prior Specification, Model Fitting and Prediction}
\label{section:Prior_specification_fitting}
In this section, we describe how the GRFDSTM can be fitted using the Bayesian approach and can be used to make predictions at new spatio-temporal locations. We shall be using the GRFDSTM with mild restrictions as described in Section \ref{section:Identifiability}, i.e. all the covariance kernels are squared exponential type and the parameters $\sigma_{0}^2,\lambda_{0},\lambda_{f},\lambda_{g}$ are fixed. Firstly, let us specify the prior structure. Unlike LDSTMs, where Gaussian-inverse gamma (or multivariate Gaussian-inverse Wishart in the case of multivariate LDSTMs) is used as the conjugate prior, a more complex structure of GRFDSTM leaves us with no hope of conjugacy. Moreover, during posterior inference, updating a high-dimensional state vector is required, and so sufficiently informative priors are needed to ensure the convergence of MCMC within feasible time.

We consider bivariate vague Gaussian priors for each of 
$(\beta_{0g},\beta_{1g})$ and $(\beta_{0f},\beta_{1f})$. The variance covariance matrix associated with the vague Gaussian priors is diagonal, with very large values (of the order $1000$) of the marginal variances which make the priors virtually non-informative. For the rest of the scale and smoothness parameters, we consider the 
lognormal prior. Also, we have taken the $N(\boldsymbol{0},\boldsymbol{\Sigma_{\mu}})$ prior for the vector parameter 
${\boldsymbol{\mu}}_{0}$, where $\boldsymbol{\Sigma_{\mu}}$ is specified by an isotropic covariance function with 
fixed parameter values. All the priors considered above are mutually independent.

Note that, although we consider lognormal priors for the scale and smoothness parameters, inverse gamma priors would also 
serve the purpose. However, lognormal distributions, which have much lighter tails compared to inverse gamma distributions, 
whose tails exhibit power law decay, provide the additional safeguard in posterior computation, in the sense that the 
corresponding MCMC algorithm doesn't travel too widely through the parameter space making the convergence time too large.

With the prior specification as above and the conditional densities 
$[\mathbf{y}|\mathbf{x},\boldsymbol{\theta}]$ and $[\mathbf{x}|\boldsymbol{\theta}]$ being explicitly available 
($\boldsymbol{\theta}$ denotes the vector consisting of all the parameters) we design a Gibbs sampler with 
Gaussian full conditionals for ${\boldsymbol{\mu}}_{0}$, $(\beta_{0g},\beta_{1g})$ and $(\beta_{0f},\beta_{1f})$ and 
update the scale parameters $\sigma_{f}^2,\sigma_{g}^2,\sigma_{\epsilon}^2,\sigma_{\eta}^2$, the smoothness parameters $\lambda_{\epsilon},\lambda_{\eta}$ and the state vector $\mathbf{x}$ together, using Transformation based Markov Chain Monte
Carlo (TMCMC) introduced by \cite{Dutta:Bhattacharya}. 
In particular, we use the additive transformation which has been shown by \cite{Dutta:Bhattacharya} to require less number of ``moves types" compared to 
other valid transformations.

The idea of TMCMC is very simple, yet a very powerful one. 
Here we briefly illustrate the idea of additive TMCMC by contrasting it with the traditional Random Walk Metropolis (RWM)
approach, assuming that we wish to update all the variables simultaneously.
Suppose that we want to simulate from the $k$-variate distribution $f_{\mathbf{U}}(\mathbf{u})$ using the RWM approach where $\mathbf{U}$ is a high dimensional vector and $f_{\mathbf{U}}(\mathbf{u})$ is the corresponding probability density or the mass function. 
Then we have to simulate $k$ independent Gaussian random variables 
$\boldsymbol{\epsilon}_{k}=(\epsilon_{1},\epsilon_{2},\cdots,\epsilon_{k})'$; assuming that the current state of the 
Markov chain is $\mathbf{u}^{(i)}$, we accept the new state $\mathbf{u}^{(i)}+\boldsymbol{\epsilon}_{k}$ 
with probability $\min\left\{1,\frac{f_{\mathbf{U}}(\mathbf{u}^{(i)}+
\boldsymbol{\epsilon}_{k})}{f_{\mathbf{U}}(\mathbf{u}^{(i)})}\right\}$. 
However, if $k$ is large then this acceptance probability will tend to be extremely small.
Hence the RWM chain sticks to a particular state for very long time, and therefore the convergence to the distribution (posterior in our case) $f_{\mathbf{U}}(\mathbf{u})$ is very slow. What additive TMCMC does is simulate only one $\epsilon>0$ 
from some arbitrary distribution left-truncated at zero, and then form the 
$k$ dimensional vector $\boldsymbol{\epsilon^*}_{k}$ setting the $l$-th element independently to $-\epsilon$ with probability 
$p_l$ and $+\epsilon$ with probability $1-p_l$. For our applications we set $p_l=1/2$ $\forall$ $l=1,2,\ldots,k$. 
Then we accept the new state 
$\mathbf{u}^{(i)}+\boldsymbol{\epsilon^*}_{k}$ 
with acceptance probability $\min\left\{1,\frac{f_{\mathbf{U}}(\mathbf{u}^{(i)}+
\boldsymbol{\epsilon^*}_{k})}{f_{\mathbf{U}}(\mathbf{u}^{(i)})}\right\}$. \cite{Dutta:Bhattacharya,Dey:Bhattacharya2016a,Dey:Bhattacharya2016b} provide details of many advantages of TMCMC (in particular, additive TMCMC)
as compared to traditional MCMC (in particular, RWM). 
However, in our setup, we have used a block TMCMC approach where separate independent $\epsilon$'s are used for each of 
the parameters $\sigma_{f}^2,\sigma_{g}^2,\sigma_{\epsilon}^2,\sigma_{\eta}^2$, $\lambda_{\epsilon},\lambda_{\eta}$, 
and each block of state vector corresponding to a particular time point $t$. This improves mixing significantly over an 
ordinary TMCMC. In our model, the dimension of the state vector is large and updating it using usual RWM would have 
been very inefficient. Block TMCMC saved us from that pitfall.

Using the above sampling-based approach it is straightforward to study the posterior distribution of the unknown quantities and 
make inferences regarding the parameters. But our main goal is to predict $y(\bold{s}^*,t^*)$ at some 
new spatio-temporal coordinate $(\bold{s}^*,t^*)$, and all summaries regarding the prediction is given by 
\begin{center}
 $[y(\bold{s}^*,t^*)|\bold{y}]=\mathlarger{\int}[y(\bold{s}^*,t^*)|\boldsymbol{\theta},\bold{y}][\boldsymbol{\theta}|\bold{y}]d\boldsymbol{\theta}$.
\end{center}
Now, to simulate from the posterior predictive distribution $[y(\bold{s}^*,t^*)|\bold{y}]$, it is enough to first simulate from $[\mathbf{x},x(\bold{s}^*,t^*),\boldsymbol{\theta}|\bold{y}]$ and then simulate from $[y(\bold{s}^*,t^*)|\mathbf{x},x(\bold{s}^*,t^*),\boldsymbol{\theta},\bold{y}]$. Now see that simulation from $[\mathbf{x},x(\bold{s}^*,t^*),\boldsymbol{\theta}|\bold{y}]$ is exactly similar to posterior simulation from $[\mathbf{x},\boldsymbol{\theta}|\bold{y}]$. The reason is that if we augment $x(\bold{s}^*,t^*)$ with the $1\times nT$ state vector
$\mathbf{x}=\left(x(\bold{s}_{1},1),\cdots,x(\bold{s}_{n},T)\right)$ and consider the conditional distribution of 
$[\mathbf{y}|\mathbf{x},x(\bold{s}^*,t^*),\boldsymbol{\theta}]$ then it is same as 
$[\mathbf{y}|\mathbf{x},\boldsymbol{\theta}]$. Once the post burn-in posterior samples\\ 
$\left\{(\mathbf{x}^{(B)},x^{(B)}(\bold{s}^*,t^*),\boldsymbol{\theta}^{(B)}),(\mathbf{x}^{(B+1)}x^{(B+1)}(\bold{s}^*,t^*),\boldsymbol{\theta}^{(B+1)}),
\cdots\right\}$ from $[\mathbf{x},x(\bold{s}^*,t^*),\boldsymbol{\theta}|\bold{y}]$ are available, it is then enough to simulate from $[y(\bold{s}^*,t^*)|\mathbf{x},x(\bold{s}^*,t^*),\boldsymbol{\theta},\bold{y}]$ plugging in them.\\
So, let us look into the form of the conditional distribution $[y(\bold{s}^*,t^*)|\mathbf{x},x(\bold{s}^*,t^*),\boldsymbol{\theta},\bold{y}]$. It is easy to see that $[y(\bold{s}^*,t^*)|\mathbf{x},x(\bold{s}^*,t^*),\boldsymbol{\theta},\bold{y}]$ is $N\left(\beta_{0f}+\beta_{1f}x(\bold{s}^*,t^*)+{\boldsymbol{\Sigma}}_{12}\left({\boldsymbol{\Sigma}}_{22}\right)^{-1}\boldsymbol{V},(\sigma_{f})^2+(\sigma_{\epsilon})^{2}-{\boldsymbol{\Sigma}}_{12}\left({\boldsymbol{\Sigma}}_{22}\right)^{-1}{\boldsymbol{\Sigma}}_{21}\right)$ where ${\boldsymbol{\Sigma}}_{12}$ is a vector of covariance values $c_{f}(x(\bold{s}^*,t^*),\mathbf{x})$, ${\boldsymbol{\Sigma}}_{22}$ is the variance covariance matrix $c_{f}(\mathbf{x},\mathbf{x})$ and $\boldsymbol{V}$ is a vector of values $ y(\bold{s}_{i},t)-\beta_{0f}-\beta_{1f}x(\bold{s}_{i},t)$, where $ i=1,\cdots,n $ and $ t=1,\cdots,T $. Hence, if we simulate $\left\{y^{(B)}(\bold{s}^*,t^*),\\y^{(B+1)}(\bold{s}^*,t^*),\cdots\right\}$ from \\
$N\left(\beta_{0f}^{(j)}+\beta_{1f}^{(j)}x^{(j)}(\bold{s}^*,t^*)+{\boldsymbol{\Sigma}}^{(j)}_{12}\left({\boldsymbol{\Sigma}}^{(j)}_{22}\right)^{-1}\boldsymbol{V}^{(j)},(\sigma_{f}^{(j)})^2+(\sigma_{\epsilon}^{(j)})^{2}-{\boldsymbol{\Sigma}}^{(j)}_{12}\left({\boldsymbol{\Sigma}}^{(j)}_{22}\right)^{-1}{\boldsymbol{\Sigma}}^{(j)}_{21}\right)$ 
where $j=B,B+1,\cdots$ and $\{\beta_{0f}^{(j)},\beta_{1f}^{(j)},\sigma_{f}^{(j)},\sigma_{\epsilon}^{(j)}\}$ 
are post burn-in posterior samples for the respective parameters, then that would give samples from the conditional distribution $[y(\bold{s}^*,t^*)|\mathbf{x},x(\bold{s}^*,t^*),\boldsymbol{\theta},\bold{y}]$. So, ultimately these \\ 
$y^{(B)}(\bold{s}^*,t^*),y^{(B+1)}(\bold{s}^*,t^*),\cdots$ are samples from the posterior predictive distribution $[y(\bold{s}^*,t^*)|\mathbf{y}]$ which are then used to calculate various summaries related to the prediction at the new spatio-temporal coordinate $(\bold{s}^*,t^*)$.

\section{Simulation Study and Real Data Analysis}

\subsection{Simulation from LDSTMs} With the model fitting and prediction method sketched in the previous section now we apply the GRFDSTM to simulated and real datasets. We consider simulated datasets generated by LDSTMs and investigate the predictive performance of the GRFDSTM on them. In fact, we consider four different datasets simulated from four different LDSTMs, ranging from a very simple model to spatio-temporally a more structured one.

I. \emph{Spatio-temporal white noise:} We consider a unit square on $\mathbb{R}^2$ and randomly generate $50$ spatial locations, where we simulate the data $Y(\bold{s},t)$ for $t=1,2,\cdots,20$, using the following spatio-temporal white noise model:
\begin{align*}
 Y(\bold{s},t)\stackrel{i.i.d}{\sim} N(0,1)\ \ \  \text{for all}\  \bold{s} \ \text{and all}\  t.
\end{align*}
This is the simplest spatio-temporal model and independent with respct to both space and time. Its analysis doesn't 
require a spatio-temporal model, but it is of interest to see how the GRFDSTM performs in this case. The jagged 
spatial surfaces obtained by spatial interpolation of the dataset, however, suggest that the GRFDSTM may not be an 
appropriate model. Independent of the data generation, we randomly generate $10$ more spatial locations, 
where we simulate $Y(\bold{s},t)$ for $t=1,2,\cdots,20$, using the same model and set aside the sample as test data. 
We compare the performance of the GRFDSTM with two different LDSTMs, the univariate LDSTM without covariate, 
proposed by \cite{Banerjee:Gamerman:Gelfand} (let us give it a name, say BGG model following the authors' surnames) 
and a variation of that (modified-BGG model). In fact, there are a number of LDSTMs in the literature, 
which can be considered for the comparative study, however, given the limitation of space and time, we consider 
only one among them. 
The LDSTM proposed by \cite{Banerjee:Gamerman:Gelfand} is very general in the sense that it accommodates space 
varying state vector; if needed can incorporate covariate information, and also admits a straightforward extension to 
multivariate spatio-temporal data. However, \cite{Banerjee:Gamerman:Gelfand} considered inverse gamma priors for the 
scale parameters and gamma prior for the smoothness parameter in the BGG model, which is entirely different from the 
lognormal prior structure used for scale and smoothness parameters of the GRFDSTM. A more comparable model is the 
modified-BGG model, for which the model specification is exactly same as the BGG model, but the prior structure is 
composed of Gaussian priors for the location parameters and lognormal priors for the scale and smoothness parameters. 
Hence, we also included it in the simulation study.

II. \emph{Temporally iid spatial process:} In this case, the basic design of the simulation remains same as the earlier one, but now we simulate from a model that has some dependence structure: 
\begin{align*}
 (Y(\bold{s}_{1},t),Y(\bold{s}_{2},t),\cdots,Y(\bold{s}_{n},t))^{\prime} \stackrel{i.i.d}{\sim} N(\bold{0},\boldsymbol{\Sigma})\ \ \  \text{for all}\ \ \  t.
\end{align*} Here, $\boldsymbol{\Sigma}$ is a spatial variance-covariance matrix induced by an exponential covariance function. 
However, the story in this case is purely spatial and no temporal dependence structure is assumed. In this case, we simulate the data $Y(\bold{s},t)$ at $60$ spatial locations for $t=1,2,\cdots,20$ and reserve the data associated with $10$ locations as the test dataset, which is used later to judge the goodness of fit.

III. \emph{Spatial random walk:} Spatial random walk models are used extensively in econometric applications (see discussions on page 260 of \cite{Banerjee04}). Here $Y(\bold{s}_{i},t)$ is distributed according to a Gaussian random walk and for two distinct spatial locations $\bold{s}_{i}$ and $\bold{s}_{j}$, $Y(\bold{s}_{i},t)$ and $Y(\bold{s}_{j},t)$ are mutually independent. 
Specifially,
\begin{align*}
  Y(\bold{s}_{i},t)&= Y(\bold{s}_{i},t-1)+\epsilon_{it} \ \ \  \text{for all}\ \ \  \bold{s}_{i};\\
  Y(\bold{s}_{i},0)& \sim N(0,1)  \ \ \ ; \  \epsilon_{it} \stackrel{iid}{\sim} N(0,1). 
\end{align*}

IV. \emph{Linear dynamic spatio-temporal model:} Finally, we simulate from a more structured linear dynamic spatio-temporal model. Unlike any of the simulation schemes I - III, here $Y(\bold{s},t)$ exhibits both spatial and temporal dependence. The spatial surfaces are smoother than those obtained in scheme I and III, owing to the spatial dependence structure. The model 
has the following form: 
\begin{align*}
     Y(\bold{s}_{i},t)&=X(\bold{s}_{i},t)+\epsilon(\bold{s}_{i},t);\\
     X(\bold{s}_{i},t)&=\rho X(\bold{s}_{i},t-1)+\eta(\bold{s}_{i},t);\\
     \{X(\bold{s}_{i},0)\}_{i=1}^n & \sim N(\bold{0},\boldsymbol{\Sigma_{0}}),
\end{align*} 
where $\{\epsilon(\bold{s}_{i},t)\}_{i=1}^n$ and $\{\eta(\bold{s}_{i},t)\}_{i=1}^n$ are temporally independent and identically distributed as $N(\bold{0},\boldsymbol{\Sigma_{\epsilon}})$ and $N(\bold{0},\boldsymbol{\Sigma_{\eta}})$, respectively. The associated variance-covariance matrices $\boldsymbol{\Sigma_{0}},\boldsymbol{\Sigma_{\epsilon}}$ and $\boldsymbol{\Sigma_{\eta}}$ are generated by exponential covariance functions of the form $c(\bold{u},\bold{v})=\sigma^2 e^{-\lambda||\bold{u}-\bold{v}||}$. We chose $\sigma_{0}^2,\sigma_{\epsilon}^2,\sigma_{\eta}^2=1$ and $\lambda_{0}=1,\lambda_{\eta}=1$ and $\lambda_{\epsilon}=0.25$. The values of the smoothness parameters are chosen carefully so that there is enough spatial dependence between the points, 
which lie within the unit square. Larger values of the smoothness parameters imply small effective range, 
rendering $Y(\bold{s},t)$ to show little spatial dependence; very small values of the smoothness parameters would imply 
almost deterministic behaviour. 
Judiciously chosen values of the 
smoothness 
parameters ensure an interesting spatial story. Finally, the value of the auto-regression coefficient $\rho$ was set to $0.8$ so that besides spatial dependence, the data exhibits strong temporal dependence as well. Note that in each of the four simulations the same spatio-temporal grid size is used to generate the training dataset ($50 \times 20$) and as well as the test dataset ($10 \times 20$). Also, the data $Y(\bold{s},t)$ consists of a single replication over space-time.

Before delving into a deeper discussion regarding the fitting of the simulated datasets, let us briefly consider the pictures presented in Figure \ref{fig:subfigures1}. The panels display spatial surfaces and time plots comprising datasets generated by simulation schemes I-IV. The plots in the 1st (leftmost) column show the spatial surfaces obtained by interpolation of the data generated by simulation schemes I to IV for a particular time slice $t$. The plots in the 2nd column show similar spatial surface plots but for a different time slice $t^{\prime}$. The plots in the 3rd (rightmost) column show the time plots obtained by interpolation of the data generated by simulation schemes I to IV for a particular spatial location $\bold{s}_{i}$. For example, panels (a), (b), (c) show the respective spatial surfaces and time plot from simulation scheme I (spatio-temporal white noise). Panels (d), (e), (f), panels (g), (h), (i) and panels (j), (k), (l) show similar spatial surfaces and time plots for simulation 
schemes  II, III and IV, 
respectively. Note that, the spatial surfaces generated by scheme I (panels (a), (b)) and scheme III (panels (g), (h)) are very jagged and the ones generated by scheme IV (panels (j), (k)) are the smoothest. This is natural as only scheme IV exhibits dependence with respect to 
both space and time, causing smoothest sample realizations. Regarding the time plots, schemes I (panel (c)) and II (panel (f)) generate very 
wiggly curves that are strongly indicative of temporal independence, whereas the curve generated by scheme IV (panel (l)) shows 
smoother behaviour. Scheme III (panel (i)) generates almost a monotonic curve, which hints towards a random walk component.

\begin{figure}[H]
     \begin{center}
        \subfigure[]{%
            \label{fig:sim1t5}
            \includegraphics[width=0.31\textwidth,height=0.16\textheight]{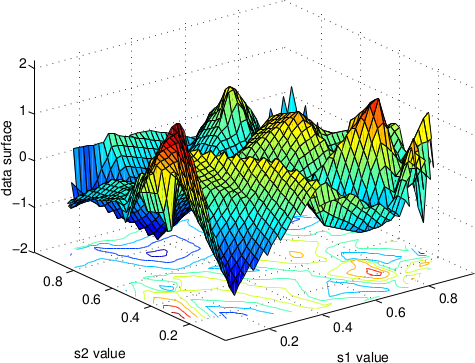}
        }%
        \subfigure[]{%
           \label{fig:sim1t9}
           \includegraphics[width=0.31\textwidth,height=0.16\textheight]{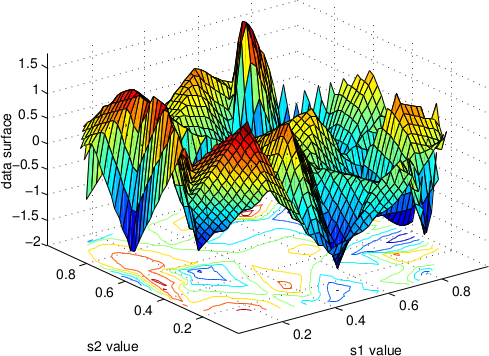}
        }%
        \subfigure[]{%
           \label{fig:sim1s5}
           \includegraphics[width=0.29\textwidth,height=0.16\textheight]{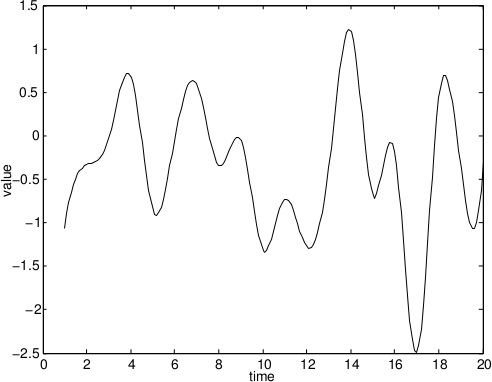}
        }\\ 

        \subfigure[]{%
            \label{fig:sim2t5}
            \includegraphics[width=0.31\textwidth,height=0.16\textheight]{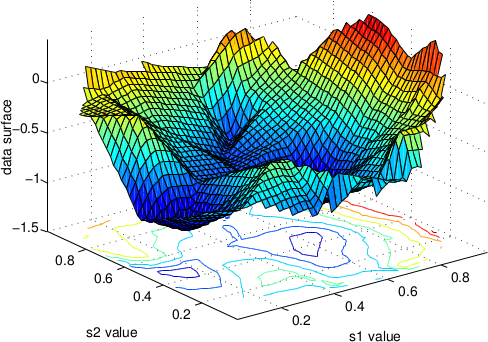}
        }%
        \subfigure[]{%
           \label{fig:sim2t9}
           \includegraphics[width=0.31\textwidth,height=0.16\textheight]{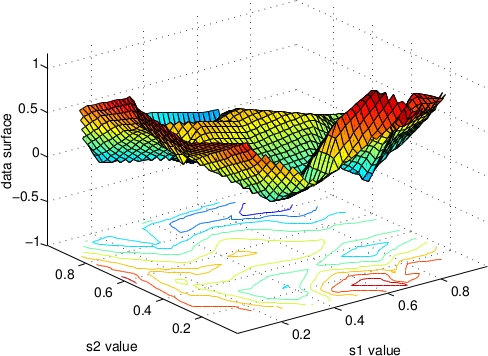}
        }%
        \subfigure[]{%
           \label{fig:sim2s3}
           \includegraphics[width=0.29\textwidth,height=0.16\textheight]{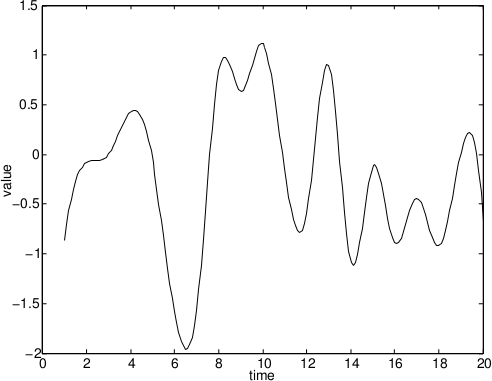}
        }\\ 

         \subfigure[]{%
            \label{fig:sim3t5}
            \includegraphics[width=0.31\textwidth,height=0.16\textheight]{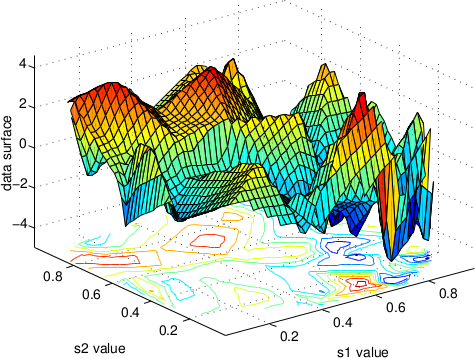}
        }%
        \subfigure[]{%
           \label{fig:sim3t9}
           \includegraphics[width=0.31\textwidth,height=0.16\textheight]{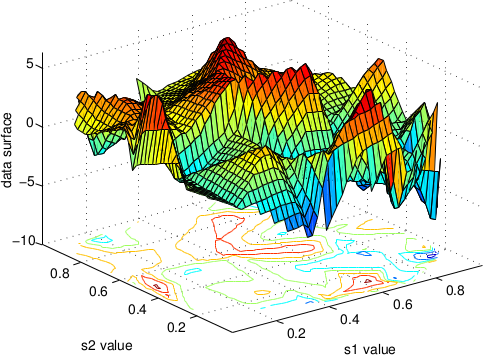}
        }%
        \subfigure[]{%
           \label{fig:sim3s35}
           \includegraphics[width=0.29\textwidth,height=0.16\textheight]{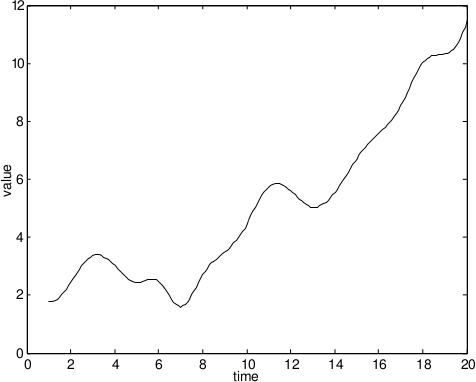}
        }\\ 

        \subfigure[]{%
            \label{fig:sim4t5}
            \includegraphics[width=0.31\textwidth,height=0.16\textheight]{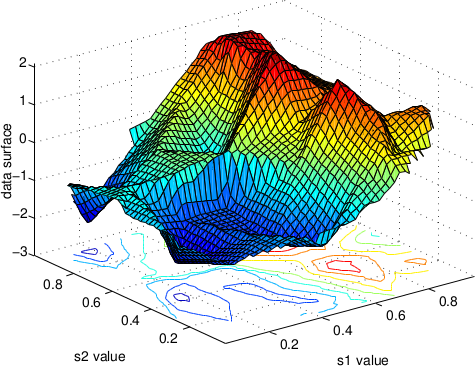}
        }%
        \subfigure[]{%
           \label{fig:sim4t9}
           \includegraphics[width=0.31\textwidth,height=0.16\textheight]{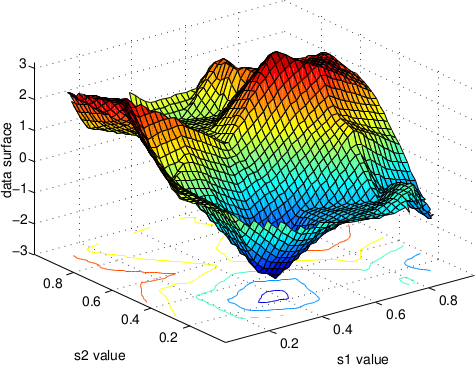}
        }%
        \subfigure[]{%
           \label{fig:sim4s19}
           \includegraphics[width=0.29\textwidth,height=0.16\textheight]{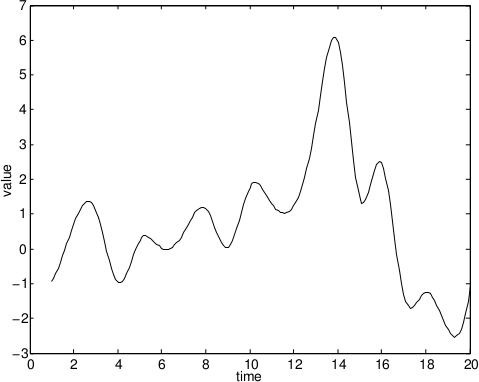}
        }\\ 
       
    \end{center}
\caption{Spatial surfaces and time plots comprising datasets generated by the simulation schemes I-IV.}	
    \label{fig:subfigures1}
\end{figure}

With datasets generated by the prescribed simulation schemes, now we fit GRFDSTM to them. We also fit the two LDSTMs, discussed above (BGG model and modified-BGG model) and compare their performances with that of the GRFDSTM. Regarding the prior elicitation and fitting of GRFDSTM, we follow the prescription provided in Section \ref{section:Prior_specification_fitting}. Hence, $\sigma_{0}^2,\lambda_{0},\lambda_{f},\lambda_{g}$ are no longer unknown parameters but some fixed values. Then the unknown parameters associated to the GRFDSTM, that we want to estimate based on the data, comprises $\beta_{0g},\beta_{1g},\beta_{0f},\beta_{1f}$ the vector parameter ${\boldsymbol{\mu}}_{0}$ associated to the initial process and the parameters associated to the squared exponential covariance kernels as specified in Section \ref{section:Identifiability}.
We specify independent vague normal priors $N(0,1000)$ for each of $\beta_{0g},\beta_{1g},\beta_{0f} \ \text{and}  \ \beta_{1f}$ and specify $N(\boldsymbol{0},\boldsymbol{\Sigma_{\mu}})$ prior for the parameter vector ${\boldsymbol{\mu}}_{0}$, where $\boldsymbol{\Sigma_{\mu}}$ is specified by an isotropic covariance function with scale parameter value $1$ and smoothness parameter value $1$. We consider the squared exponential covariance kernels with the following representation $c(\bold{u},\bold{v})=\sigma^2 e^{-\lambda||\bold{u}-\bold{v}||^2}$. Associated with them we have five covariance 
kernels  $c_{f}(\cdot,\cdot),c_{g}(\cdot,\cdot),c_{\epsilon}(\cdot,\cdot),c_{\eta}(\cdot,\cdot),
c_{0}(\cdot,\cdot)$, with five scale parameters $\sigma_{f}^2,\sigma_{g}^2,\sigma_{\epsilon}^2,\sigma_{\eta}^2,\sigma_{0}^2$ and five smoothness parameters $\lambda_{f},\lambda_{g},\lambda_{\epsilon},\lambda_{\eta},\lambda_{0}$. Recall that among them $\sigma_{0}^2,\lambda_{0},\lambda_{f},\lambda_{g}$ are just some fixed quantities. We set $\sigma_{0}^2=1,\lambda_{0}=100000,\lambda_{f}=1$ and $\lambda_{g}=1$. Fixing the value of $\lambda_{0}$ to $100000$ renders the covariance matrix $\boldsymbol{\Sigma}_{0}$ associated with $(X(\bold{s}_{1},0),X(\bold{s}_{2},0),\cdots X(\bold{s}_{n},0))$, an identity matrix. This ensures that $\boldsymbol{\Sigma}_{0}$ doesn't interfere with the spatial dependence structure of $Y(\boldsymbol{s},t)$. For the rest of the scale and smoothness parameters, we consider independent lognormal priors. For example, we consider lognormal$(0.4,1.4)$ prior for each of the scale parameters $\sigma_{f}^2,\sigma_{g}^2,\sigma_{\eta}^2$ and $\sigma_{\epsilon}^2$ for the data generated by 
simulation scheme IV. For that same dataset we consider lognormal$(0,0.2)$ prior for $\lambda_{\eta}$ and lognormal$(0,0.05)$ prior for $\lambda_{\epsilon}$. Note that we use the following representation for the 
lognormal$(
\mu,\sigma)$ pdf :
\begin{align*}
\mathlarger{f_{\mu,\sigma}(x)=\frac{1}{\sqrt{2\pi}\sigma}\frac{1}{x}e^{-\frac{1}{2}(\frac{\ln(x)-\mu}{\sigma})^2}} \ \ \  \text{for all}\ \ \  x>0.
\end{align*}
The reason behind choosing very concentrated, almost degenerate priors for $\lambda_{\eta}$ and $\lambda_{\epsilon}$ is that, smoothness parameters are notoriously difficult to handle during MCMC computation. Their erratic movement across the Markov chain state space would lead to severely ill-conditioned variance-covariance matrix, making the MCMC algorithm 
difficult to converge in feasible time. However, with that the question comes how we get these specific priors? In this case lognormal$(0,0.05)$ and lognormal$(0,0.2)$ respectively. We execute MCMC pilot runs based on a smaller subset of the entire dataset and run the MCMC for different combinations of prior hyper parameters, i.e. we try out lognormal$(a,b)$ and lognormal$(c,d)$ and choose the combination of values $a,b,c,d$ that gives the best fit in terms of a goodness of fit measure $D_{0.5}$ to be introduced later.

So far, we have considered the prior specification regarding the GRFDSTM. Now, let us consider the parameters and prior structures for the BGG model and modified-BGG model. For the BGG model, we follow the same prior structure as prescribed in \cite{Banerjee:Gamerman:Gelfand}; vague normal prior for the location parameter $\beta_{0}$, inverse gamma priors for the scale parameters $\sigma_{\eta}^2,\sigma_{\epsilon}^2,\sigma_{w}^2$ and gamma prior for the smoothness parameter $\phi$. For the modified-BGG  model, the basic structure is exactly the same; however, instead of gamma and inverse gamma, lognormal priors are taken for the scale and smoothness parameters. 

To each of the $4$ simulated datasets, we fit all the $3$ models. We consider $100,000$ iterations of Markov chains for all of them, taking the first $80,000$ iterations as burn-in and the post burn-in $20,000$ iterations have been used for posterior inference. Convergence of the MCMC chains are confirmed by using the CODA package of R. Besides, visual checking of the individual trace-plots also suggests that the convergence is satisfactory.

To assess the model fit we propose a goodness of fit type of statistic. Note that, since the main goal behind developing the GRFDSTM is to predict efficiently, the proposed statistic is so developed to measure how efficiently a model is able to predict at a new spatio-temporal location. Essentially, it is a convex combination of two statistics :
\begin{align*}
D_{\alpha}=\frac{\alpha}{m}\sum_{(\boldsymbol{s}^*,t^*)}|\hat{y}^{pred}(\boldsymbol{s}^*,t^*)-y(\boldsymbol{s}^*,t^*)|+\frac{(1-\alpha)}{m}\sum_{(\boldsymbol{s}^*,t^*)}|Q_{0.975}^{pred}(\boldsymbol{s}^*,t^*)-Q_{0.025}^{pred}(\boldsymbol{s}^*,t^*)| ; \  0<\alpha<1
\end{align*}
where $\hat{y}^{pred}(\boldsymbol{s}^*,t^*)$ is any suitable representative central value for the posterior predictive distribution $[Y(\boldsymbol{s}^*,t^*)|\bold{y}]$ at the spatio-temporal location $(\boldsymbol{s}^*,t^*)$ and $m$ is the number of spatio-temporal coordinates $(\boldsymbol{s}^*,t^*)$, where we are predicting. We recommend using the posterior median as $\hat{y}^{pred}(\boldsymbol{s}^*,t^*)$. The statistic $\sum_{(\boldsymbol{s}^*,t^*)}|\hat{y}^{pred}(\boldsymbol{s}^*,t^*)-y(\boldsymbol{s}^*,t^*)|$ measures the deviation of the point prediction value from the observed test data $y(\boldsymbol{s}^*,t^*)$, and $\sum_{(\boldsymbol{s}^*,t^*)}|Q_{0.975}^{pred}(\boldsymbol{s}^*,t^*)-Q_{0.025}^{pred}(\boldsymbol{s}^*,t^*)|$ measures the width of the $95\%$ Bayesian symmetric prediction interval w.r.t. the posterior predictive distribution $[Y(\boldsymbol{s}^*,t^*)|\bold{y}]$. The later part acts like a penalty term, so that any model for which the point prediction value agrees with the observed 
test 
data, but the reliability of the point predictor is low, the value of $D_{\alpha}$ goes up. Hence, even if the test data 
fall well within the $95\%$ Bayesian symmetric prediction interval; a wider prediction interval means $D_{\alpha}$ is large and hence the model should be discarded. The amount of penalization is specified by $\alpha$. A natural choice may be $\alpha=0.5$. The proposed statistic $D_{\alpha}$ is essentially a robust version of the predictive criteria for model selection considered in \cite{Gelfand:Ghosh}.

The findings of the simulation study are summarized in the Table \ref{tab:table1}.
\begin{center}
\begin{table}[h]
\smaller
  \begin{tabular}{      l      c      c      c      c      c      r      }
    \hline
     Simulation & \  & \ & I. Spatio-tempoal & II. Temporally iid & III. Spatial & IV. Linear\\ 
    \ schemes & \  & \  & \ white noise & \ \ \ \ \  spatial process & \ \ \ \ random walk & \ \ \ DSTM\\
    \hline
    \hline
    \\
    Percentage & GRFDSTM & \ & 85.20\% & 97.40\% & 85.90\% & 95.40\% \\
     of data & BGG model & \ & 96.40\% & 100.00\% & 100.00\% & 99.80\%\\
    falling inside & modified-BGG & \ & 95.70\% & 99.10\% & 94.50\% & 99.60\%\\
    $95\%$ Bayesian  & model\\
    symmetric\\
    prediction interval\\
    \\
     $95\%$ Bayesian symmetric &  GRFDSTM & \ & 3.65 & 0.44 & 14.63 & 1.29\\
     prediction interval & BGG model & \ & 4.22 & 1.55 & 4.60 & 3.37\\
     width & modified-BGG & \ & 3.99 & 1.30 & 12.44 & 3.07\\
     $\mathsmaller{\frac{1}{m}\sum_{(\boldsymbol{s}^*,t^*)}|Q_{0.975}^{pred}(\boldsymbol{s}^*,t^*)-Q_{0.025}^{pred}(\boldsymbol{s}^*,t^*)|}$ & model\\
     \\
     Error in & GRFDSTM & \ & 0.99 & 0.08 & 3.79 & 0.25\\
     Point prediction & BGG model & \ & 0.78 & 0.16 & 0.41 & 0.32\\
     $\mathsmaller{\frac{1}{m}\sum_{(\boldsymbol{s}^*,t^*)}|\hat{y}^{pred}(\boldsymbol{s}^*,t^*)-y(\boldsymbol{s}^*,t^*)|}$    & modified-BGG & \ & 0.78 & 0.19 & 2.35 & 0.40\\
     \ & model\\
     \\
     Value of & GRFDSTM & \ & 2.32 & 0.26 & 9.21 & 0.77\\
     $D_{0.5}$ statistic & BGG model & \ & 2.50 & 0.85 & 2.50 & 1.85\\
     \ & modified-BGG & \ & 2.39 & 0.75 & 7.39 & 1.73\\
     \ & model\\
     \hline
   \end{tabular}
   \caption{Fitting and prediction associated to simulated datasets generated by LDSTM scheme I-IV; summary of results.}
    \label{tab:table1} 
 \end{table}
\end{center}

Following Table \ref{tab:table1}, we see that for each simulation scheme I-IV, the highest percentage of test data points fall inside the $95\%$ Bayesian symmetric prediction interval given by the BGG model. This is quite likely, given the fact that among the three competing models, the BGG model also produces the widest prediction intervals. The lowest percentage of test data points fall inside the prediction intervals associated with the GRFDSTM, that also produces the narrowest prediction intervals. It is tempting to think that approximately $0.95$ proportion of the test dataset should fall inside the $95\%$ Bayesian symmetric prediction interval, in case the model assumption is adequate. But this consideration is incorrect since firstly, the prediction intervals are calculated based on the univariate posterior predictive distributions for each of the test data points separately and they are highly dependent; more importantly, this is a Bayesian prediction interval and unlike the classical prediction 
interval 
a direct relative frequency based interpretation is not available. However, this doesn't ward off its utility as a model comparison criteria and in general, a model that produces prediction intervals that has moderate widths with many of the test data points lying inside the respective prediction intervals, indicate good performance. Note that, except the spatial random walk scheme, the GRFDSTM performs uniformly better than the two other competing models in terms of the $D_{0.5}$ goodness of fit statistic. In particular, for schemes II and IV, where the spatial story is nontrivial, the GRFDSTM performs much better than the BGG model and modified-BGG model. In the case of scheme III, i.e., the spatial random walk case, however, the BGG model performs strikingly better than the GRFDSTM and modified-BGG model. An apparent explanation for that is the inherent random walk structure of the state variables in the BGG model make it the most suitable model for the spatial random walk dataset. The performances of all 
the three models are very similar in the case of scheme I, that offers no spatial and temporal story. However, as 
Figure \ref{fig:subfigures2} shows, for the simulation schemes I and III, GRFDSTM follows the data pattern 
much better than the competing two models.

\begin{figure}[H]
     \begin{center}
        \subfigure[]{%
            \label{fig:sim1s30fitted}
            \includegraphics[width=0.31\textwidth,height=0.16\textheight]{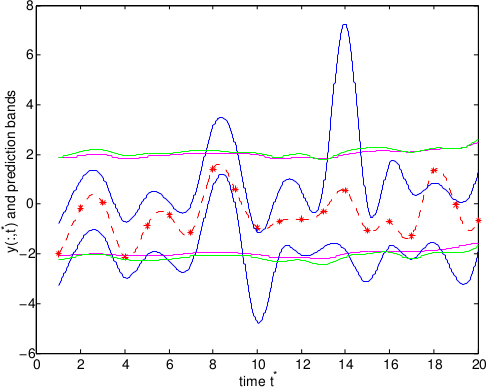}
        }%
        \subfigure[]{%
           \label{fig:sim3s35fitted}
           \includegraphics[width=0.31\textwidth,height=0.16\textheight]{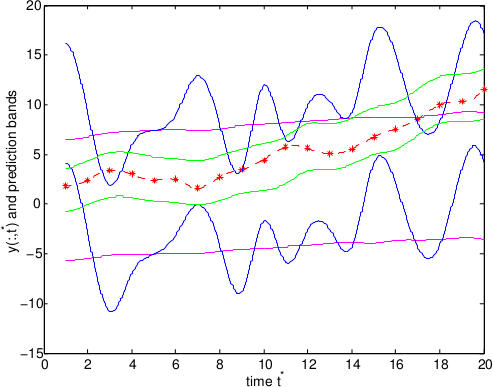}
        }%
        \subfigure[]{%
           \label{fig:sim4s42fitted}
           \includegraphics[width=0.31\textwidth,height=0.16\textheight]{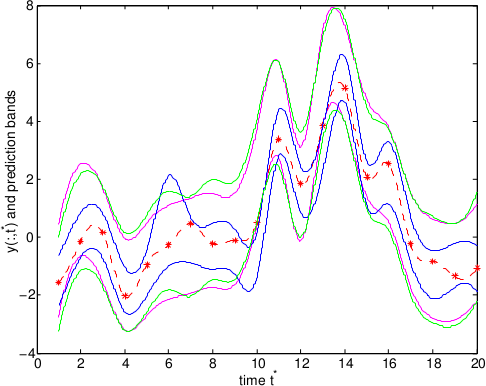}
        }\\ 

    \end{center}
\caption{Specimen pictures showing test data and prediction bands for simulation schemes I, III and IV. Panel (a) displays the test data generated by simulation scheme I (i.e. spatio-temporal white noise) and the 95\% Bayesian symmetric prediction intervals given by GRFDSTM, BGG model and modified-BGG model at a particular representative spatial test data location for $t=1,2,\cdots,20$; the red starred line represents the time series of test data at test data location $\bold{s}^{*}_{i}$ for  $t=1,2,\cdots,20$ and a smooth red curve is interpolated through that test data; the blue band, the green band, and the magenta band represent the 95\% Bayesian symmetric prediction interval associated with the GRFDSTM, BGG model and modified-BGG model, respectively. Panel (b) displays a similar plot for the spatial random walk data, i.e., the data generated by simulation scheme III and panel (c) displays a similar plot for the data generated by the structured LDSTM, i.e., simulation scheme IV.}	
    \label{fig:subfigures2}
\end{figure}

\subsection{Simulation from Nonlinear Non-Gaussian Dynamic Spatio-Temporal Models}

We have simulated spatio-temporal data using four different linear dynamic models and carried out a comparative study of GRFDSTM and two other LDSTM models on them. The LDSTM models, themselves being linear, show excellent predictive performance. What is more interesting is that the GRFDSTM's performance is equally competitive. In fact, in schemes II and IV, where there is structured spatial dependence, its performance is better than the BGG model and the modified-BGG model. However, the focal issue behind developing the GRFDSTM is to model nonlinear dynamics efficiently, bypassing the choices of specific nonlinear forms and hence checking its performance when the data is simulated by a nonlinear model, is more crucial. We consider the following two different nonlinear dynamic models and simulate spatio-temporal data using them.

V. \emph{Power transform based NLDSTM:} This model is motivated by the work of \cite{Sanso:Guenni} in the context of rainfall modeling. Rainfall is influenced by complex interactions among atmospheric processes and therefore is best modeled by nonlinear models. Similar to the previous simulation schemes, we consider a unit square on $\mathbb{R}^2$ and randomly generate $60$ spatial locations, where we simulate a single replicate of $Y(\bold{s},t)$ for $t=1,2,\cdots,20$ using the following spatio-temporal model:
\begin{align*}
      &Y(\bold{s}_{i},t)=\begin{cases} 
                    {X(\bold{s}_{i},t)}^{b},\ \ \  &\text{if}\  X(\bold{s}_{i},t)>0;\\
                    0\ \ \  &\text{if}\  X(\bold{s}_{i},t)\leq 0; 
                    \end{cases}\\
     &X(\bold{s}_{i},t)=\alpha + \beta X(\bold{s}_{i},t-1)+\eta(\bold{s}_{i},t);\\
     &\{X(\bold{s}_{i},0)\}_{i=1}^n \sim N(\bold{0},\boldsymbol{\Sigma_{0}}).
\end{align*}
We set aside the data associated with $10$ spatial location as test dataset and the rest of the data is fitted using the GRFDSTM and the two other competing LDSTMs (BGG model and modified-BGG model). Here $\{\eta(\bold{s}_{i},t)\}_{i=1}^n$ are temporally independent and identically distributed as $N(\bold{0},\boldsymbol{\Sigma_{\eta}})$. The associated variance-covariance matrices $\boldsymbol{\Sigma_{0}}$ and $\boldsymbol{\Sigma_{\eta}}$ are generated by exponential covariance functions of the form $c(\bold{u},\bold{v})=\sigma^2 e^{-\lambda||\bold{u}-\bold{v}||}$. We choose $\sigma_{0}^2=1,\sigma_{\eta}^2=1$ and $\lambda_{0}=1,\lambda_{\eta}=1$. The values of the smoothness parameters are chosen carefully so that there is an interesting spatial story. We choose $\alpha=1$ and $\beta=-0.8$, which imply nontrivial temporal evolution. Finally, the value of $b$ was set to $3$ to make sure that the process exhibits enough nonlinearity within the time span $t=1,2,\cdots,20$. The observed data $y(\bold{s}_{i},t)$ 
is non-
Gaussian. However, the nonlinearity in this model is introduced at the observational equation and the state process evolves linearly. The next model, instead, describes a situation where the nonlinearity is expressed through the evolution of the state process.

VI. \emph{Threshold NLDSTM:} Threshold models are very useful in the context of economic time series data modeling. Spatio-temporal adaptation of such model is described by \cite{Berliner:Wikle:Cressie} in the context of forecasting of sea-surface temperature. The grid size and basic design of the simulation remains same as the earlier one, but now we simulate from the following model: 
\begin{align*}
      &Y(\bold{s}_{i},t)=X(\bold{s}_{i},t)+\epsilon(\bold{s}_{i},t);\\
      &X(\bold{s}_{i},t)=\begin{cases} 
                    1-0.6X(\bold{s}_{i},t-1)+\eta(\bold{s}_{i},t),\ \ \  &\text{if}\  X(\bold{s}_{i},t)<-4;\\
                    X(\bold{s}_{i},t-1)+\eta(\bold{s}_{i},t),\ \ \  &\text{if}\ -4<X(\bold{s}_{i},t)<4;\\
                    -1-0.6X(\bold{s}_{i},t-1)+\eta(\bold{s}_{i},t),\ \ \  &\text{otherwise};\  \\
                    \end{cases}\\
     &\{X(\bold{s}_{i},0)\}_{i=1}^n \sim N(\bold{0},\boldsymbol{\Sigma_{0}}).
\end{align*}
The spatial processes $\eta(\cdot,t)$ and $\epsilon(\cdot,t)$ are temporally independent and identically distributed as centered Gaussian processes with covariance functions $c_{\eta}(\cdot,\cdot)$ and $c_{\epsilon}(\cdot,\cdot)$. We choose $\sigma_{0}^2=1,\sigma_{\eta}^2=1,\sigma_{\epsilon}^2=1$ and $\lambda_{0}=1,\lambda_{\eta}=1,\lambda_{\epsilon}=0.25$. 

\begin{figure}[H]
     \begin{center}
        \subfigure[]{%
            \label{fig:sim5t7}
            \includegraphics[width=0.31\textwidth,height=0.16\textheight]{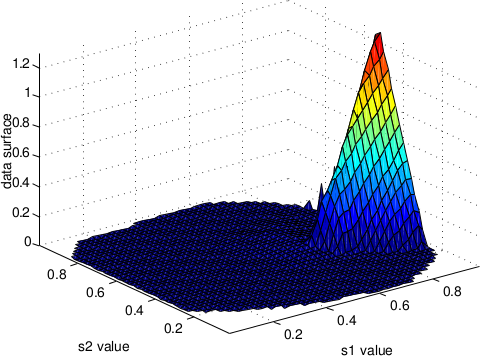}
        }%
        \subfigure[]{%
           \label{fig:sim5t16}
           \includegraphics[width=0.31\textwidth,height=0.16\textheight]{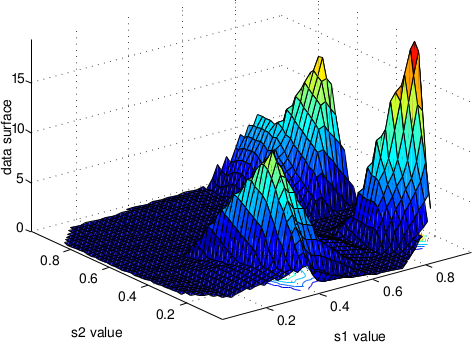}
        }%
        \subfigure[]{%
           \label{fig:sim5s21}
           \includegraphics[width=0.29\textwidth,height=0.16\textheight]{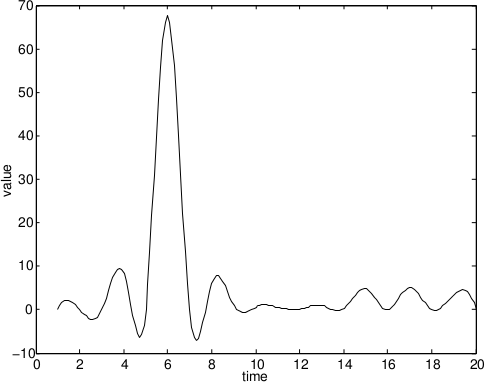}
        }\\ 

        \subfigure[]{%
            \label{fig:sim6t8}
            \includegraphics[width=0.31\textwidth,height=0.16\textheight]{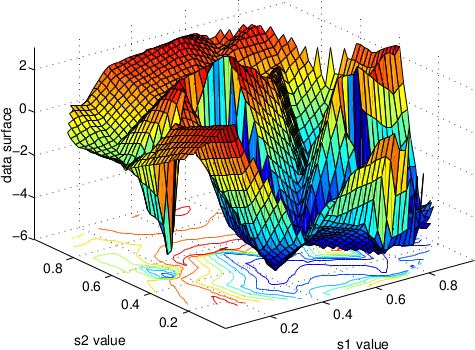}
        }%
        \subfigure[]{%
           \label{fig:sim6t9}
           \includegraphics[width=0.31\textwidth,height=0.16\textheight]{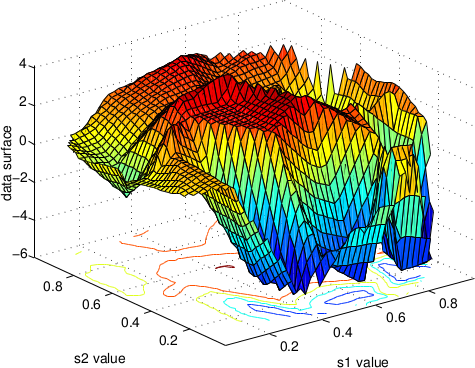}
        }%
        \subfigure[]{%
           \label{fig:sim6s13}
           \includegraphics[width=0.29\textwidth,height=0.16\textheight]{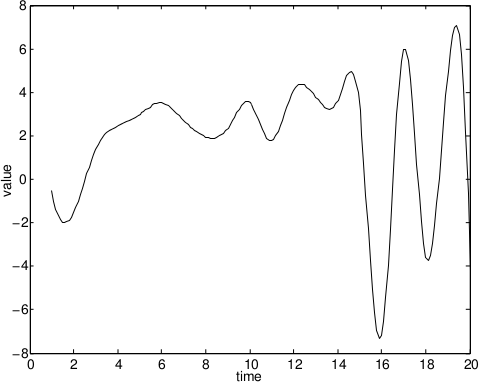}
        }\\ 
                 
    \end{center}
\caption{Spatial surfaces and time plots generated from the nonlinear non-Gaussian models. The plots are analogous to the plots in Figure \ref{fig:subfigures1}. Panels (a), (b) and (c) correspond to the data simulated from the power transform based NLDSTM. The spike and slab structures of the spatial surface plots of panels (a) and (b) and the time plot of panel (c) are indicative of nonlinear data generating process. Panels (d), (e) and (f) correspond to the data simulated from the threshold NLDSTM.}	
    \label{fig:subfigures3}
\end{figure}

Like the previous schemes, the data $y(\bold{s}_{i},t)$ generated by this model is non-Gaussian; indeed a three component Gaussian mixture.

\begin{table}[h]
\smaller
  \begin{center}
  \begin{tabular}{      l      c      c      c      c      c      r      }
    \hline
     Simulation & \  & \ & \ & V. Power transform & \ & VI. Threshold\\ 
    \ schemes & \  & \  & \ \ & \ \ \ \ \  based NLDSTM  & \ \ \ \ \ & \ \ \  NLDSTM\\
    \hline
    \hline
    \\
    Percentage & GRFDSTM & \ & \ \ & 90.80\% & \ \ & 94.30\% \\
     of data & BGG model & \ & \ \ & 97.50\% & \ \ & 98.70\%\\
    falling inside & modified-BGG & \ & \ \ & 97.50\% & \ \ & 98.20\%\\
    $95\%$ Bayesian  & model\\
    symmetric\\
    prediction interval\\
    \\
     $95\%$ Bayesian symmetric &  GRFDSTM & \ & \ & 17.18 & \ & 5.57\\
     prediction interval & BGG model & \ & \ & 31.59 & \ & 9.10\\
     width & modified-BGG & \ & \ & 26.05 & \ & 9.69\\
     $\mathsmaller{\frac{1}{m}\sum_{(\boldsymbol{s}^*,t^*)}|Q_{0.975}^{pred}(\boldsymbol{s}^*,t^*)-Q_{0.025}^{pred}(\boldsymbol{s}^*,t^*)|}$ & model\\
     \\
     Error in & GRFDSTM & \ & \ & 4.21 & \ & 1.16\\
     Point prediction & BGG model & \ & \ & 3.27 & \ & 1.25\\
     $\mathsmaller{\frac{1}{m}\sum_{(\boldsymbol{s}^*,t^*)}|\hat{y}^{pred}(\boldsymbol{s}^*,t^*)-y(\boldsymbol{s}^*,t^*)|}$    & modified-BGG & \ & \ & 2.54 & \ & 1.66\\
     \ & model\\
     \\
     Value of & GRFDSTM & \ & \ & 10.69 & \ & 3.37\\
     $D_{0.5}$ statistic & BGG model & \ & \ & 17.43 & \ & 5.18\\
     \ & modified-BGG & \ & \ & 14.29 & \ & 5.67\\
     \ & model\\
     \hline
   \end{tabular}
   \end{center}
   \caption{Summary of results of fitting and prediction associated with datasets simulated by 
   scheme V (power transform based NLDSTM) and scheme VI (threshold NLDSTM).}
    \label{tab:table3} 
 \end{table}

We fit the GRFDSTM and the BGG model and modified-BGG model to the data simulated from the nonlinear non-Gaussian models and the findings are summarized in Table \ref{tab:table3}. Following Table \ref{tab:table3}, we see that for both the simulation schemes V and VI, almost the same percentage of test data points fall inside the $95\%$ Bayesian symmetric prediction intervals given by the BGG model and the modified-BGG model. However, the width of the $95\%$ Bayesian symmetric prediction intervals given by the GRFDSTM is much narrower than the intervals given by the two other competing models. In fact, in terms of the $D_{0.5}$ goodness of fit statistic, the GRFDSTM demonstrates significantly better performance than the two other competing models.
The observations are in keeping with the detailed diagrams depicted in Figure \ref{fig:subfigures4}, and, as before, it is
seen that GRFDSTM follows the data more closely compared to the other competing models.

\begin{figure}[H]
     \begin{center}
        \subfigure[]{%
            \label{fig:sim5s30fitted}
            \includegraphics[width=0.31\textwidth,height=0.16\textheight]{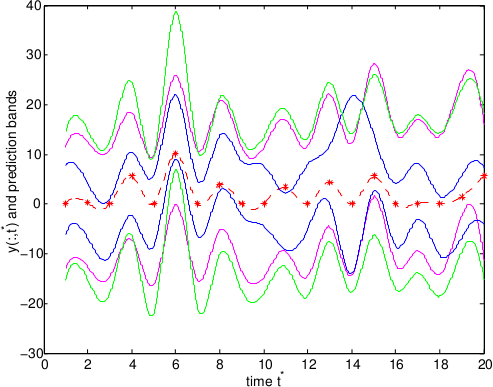}
        }
        \subfigure[]{%
           \label{fig:sim6s39fitted}
           \includegraphics[width=0.31\textwidth,height=0.16\textheight]{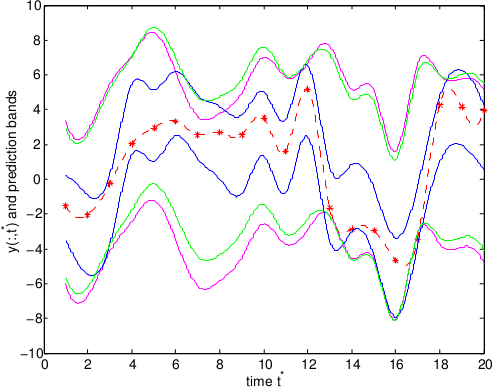}
        } 

    \end{center}
\caption{Specimen pictures showing test data and prediction bands for simulation schemes V and VI. Panel (a) displays the test data generated by the power transform based NLDSTM and the 95\% Bayesian symmetric prediction intervals given by GRFDSTM, BGG model and modified-BGG model at a particular representative spatial test data location for $t=1,2,\cdots,20$; the red starred line represents the time series of test data at test data location $\bold{s}^{∗}_{i}$ for $t=1,2,\cdots,20$ and a smooth red curve is interpolated through the test data; the blue band, the green band, and the magenta band represent the prediction intervals associated with the GRFDSTM, BGG model and the modified-BGG model, respectively. Panel (b) displays a similar plot for the threshold NLDSTM data, i.e., the data generated by simulation scheme VI.}	
    \label{fig:subfigures4}
\end{figure}

Comparative study of spatio-temporal models, however, is ambiguous and the result may strongly depend on the details of how the approaches are implemented. Moreover, the choices of the values of the hyper parameters may entirely change the result. \cite{Dou:Le:Zidek} also mentioned this issue while comparing two spatio-temporal methods. We do our best to make the comparative study worthwhile in the following way; while the properties associated with the model give us hint about what could be the feasible range of values for different hyper parameters, proposing some fixed values based on them is almost impossible. We use MCMC pilot runs with smaller datasets to find out the values of different hyper parameters. Several different sets of values of hyper parameters were chosen and for each of them, we have MCMC pilot runs based on a smaller subset of the entire dataset. Since, due to smaller sizes of the data sets the pilot runs were much faster, we were able to experiment with many such pilot MCMC runs with 
different 
combinations of the hyper parameters. We chose that combination with the smallest value of $D_{0.5}$ goodness of fit statistic. Then we performed our final MCMC computations based on the entire dataset using those selected values of the hyper parameters. This strategy is used for selecting the values of hyper parameters associated with each of the three competing models. The models selected by this strategy may not be the best ones, but the strategy provides us with candidate models, based on which a comparative study may be meaningful.

While fitting the GRFDSTM to the simulated datasets we did not have to handle missing data. In reality, however, missing data is quite common in spatio-temporal datasets. Mechanical disturbances and electronic malfunctions in measuring devices may even lead to situations where more than $50\%$ data may be missing at some specific site. Missing data problem, however, can be handled straight forwardly in the GRFDSTM, by the data augmentation technique. The key idea is to treat them as prediction problem.

\subsection{Computational Issues}

Although dynamic models are preferred over marginal models because of their superior ability to represent the temporal evolution of complex physical processes, they are associated with much higher computational complexities. Modeling of even moderate 
size spatio-temporal datasets using dynamic spatio-temporal evolution may take significant amounts of time, thus 
limiting their scope. Even the simplest of LDSTMs may be computationally very demanding owing to handling and 
inversion of large variance-covariance matrices. It is no wonder, that the problem becomes even more critical 
in the case of the GRFDSTM. Let us consider the computational cost associated with the GRFDSTM in more detail. 
Recall that we have $n$ monitoring sites where we have collected the data for $T$ consecutive time points. 
Then, in each iteration of the MCMC algorithm for the posterior computation, we need to update the $n(T+1)$ 
dimensional state vector and the smoothness and scale parameters together using TMCMC, which consist of construction of large 
covariance matrix, computing their Cholesky decomposition, calculation of the quadratic form and determinants, etc. 
The total computational cost associated with this step can be shown to be $\sim \frac{2}{3}n^3T^3$. 
This total cost is estimated as follows:
\begin{align*}
 &\texttt{Constructing the movetype:} \sim  2nT+T&\\
 &\texttt{Forming two covariance matrices from} \\
 &\texttt{squared exponential covariance kernel:}  \sim 2\times[4n^2T^2+n^2T-nT+T+4n^2-4n]&\\
 &\texttt{Forming two associated vectors:} \sim 2\times3nT &\\
 &\texttt{Cholesky decomposition of the two} \\
 &\texttt{formed covariance matrices:} \sim 2\times\frac{1}{3}n^3T^3 &\\
 &\texttt{Finding the determinant and the quadratic form:} \sim 2\times[2n^2T^2+3nT]&\\
 &\texttt{Other calculations:} \sim 54 &\\
\end{align*}
Combining them and ignoring all the terms whose total power (adding powers of $n,T$) is $\leq 4$, we get $\frac{2}{3}n^3T^3$. Apart from that, we need to update $\beta_{0g},\beta_{1g},\beta_{0f},\beta_{1f}$ and $\boldsymbol{\mu}_{0}$ using Gibbs steps and the total computational complexity associated to that is $12n^2T^2+\frac{1}{3}n^3$. Hence, ignoring lower order terms, we obtain the total computational cost for fitting the GRFDSTM to be $\sim \frac{2}{3}n^3T^3$. However, this is the cost if we use the GRFDSTM just for fitting the data and it increases substantially if we also consider prediction at new locations. Let us assume that we want to predict the whole time series of observations at each of the $n^*$ unmonitored sites. Considering $n^{\prime}=n+n^*$ one can show this amounts to the extra cost of the order $(\frac{4}{3}{n^{\prime}}^3+\frac{1}{3}n^3+2n^2n^*+2n{n^*}^2+\frac{1}{3}{n^*}^3)T^3$ (ignoring lower order terms). Combining them we get that the total cost of fitting and prediction associated 
to the 
GRFDSTM is $\sim (\frac{4}{3}{n^{\prime}}^3+n^3+2n^2n^*+2n{n^*}^2+\frac{1}{3}{n^*}^3)T^3$. Similar calculation for the BGG model (and modified-BGG model) shows that the cost associated with fitting is $n^3T^3+10n^2T^2+T^3+\frac{2}{3}n^3$ which is more than the cost associated with fitting GRFDSTM. However, if prediction is considered, the cost associated with BGG model is $(n^3+\frac{1}{3}{n^*}^3+2n{n^*}^2+2n^2n^*)T^3$, which is less compared to the cost of prediction associated with GRFDSTM. Note that, there are one time computations associated with each of these models, which can be safely ignored because of their negligible contributions to the overall computational costs. Besides computing the theoretical complexity, we also investigate the empirical computational time associated with each of the three competing models. Such comparative study of theoretical and empirical computational efficiency is meaningful since the computations associated with all of the three models are implemented through C 
programs in 
the same computer. The findings are summarized in Table \ref{tab:table2}.

\begin{table}[h]
\smaller
  \begin{center}
  \begin{tabular}{      l      c      c      r      }
    \hline
     Computational & Computational complexity & Computational cost  &  Empirical computation\\ 
     procedure & per iteration & per iteration & time (in minutes) taken\\ 
     \ & of MCMC & of MCMC & for $1000$ MCMC iterations \\
     \ & \ & ($n=50,n^*=10, T=20$) & ($n=50,n^*=10, T=20$)\\
    \hline
    \hline
    \\
    Fitting of GRFDSTM & $\frac{2}{3}n^3T^3$ & $66.7\times 10^7$ \ & 9 \\
    \\
    Fitting of BGG model & $n^3T^3+10n^2T^2+$ & $101.0\times 10^7$ & 18 \\
    and modified-BGG model & $T^3+\frac{2}{3}n^3$ & \ &  \\
    \\
    \hline
    \\
    Prediction & $(\frac{4}{3}{n^{\prime}}^3+n^3+$ & $378.7\times 10^7$ \ & 36 \\
    by GRFDSTM & \ $2n^2n^*+2n{n^*}^2+\frac{1}{3}{n^*}^3)T^3$ \ & \ \\
    \\
    Prediction by BGG model & $(n^3+\frac{1}{3}{n^*}^3+$ & $148.3\times 10^7$ & 19 \\
    and modified-BGG model & $2n{n^*}^2+2n^2n^*)T^3$ & \ &  \\
    \hline
    \hline
   \end{tabular}
   \end{center}
   \caption{Theoretical computational complexity and empirical computation time associated with fitting and prediction of the GRFDSTM and BGG model (and modified-BGG model).}
    \label{tab:table2} 
 \end{table}

The empirical computation time is derived based on the spatio-temporal data generated by simulation scheme I, i.e., the spatio-temporal white noise model. All the computations are performed on a standard Dell Inspiron laptop with Intel® Core™ i5-4200U CPU @ 1.60GHz $\times$ 4 processor. Note that there are differences between the results obtained by theoretical derivation and empirical findings. Such discrepancies are not unexpected given that in reality even inverting two matrices of same order may take different times. 
Roughly, for the simulation study we have considered, fitting of the GRFDSTM model is twice as fast as that of the BGG model. In the case of prediction, however, the result is reversed and prediction by the BGG model is twice as fast as that of GRFDSTM.

So far, the discussion regarding the computation clearly indicates that the fitting of DSTM models, as well as the GRFDSTM, may be costly. The GRFDSTM works well for small low resolution spatio-temporal dataset. In the following subsection, however we present a simulation study based on a moderate sized spatio-temporal dataset that has been generated by a highly nonlinear sptio-temporal process. Surprisingly, even in this case the GRFDSTM performs very well. That said, however, still the GRFDSTM as it currently specified, is not scalable to large spatio-temporal data i.e. when $nT$ becomes of the order $20,000$ the computation becomes intractable. For that, one need additional strategies like dimension reduction of the state vector, sparsity assumption regarding the covariance matrices, parallel processing, etc. Indeed, the analysis of massive spatial and spatio-temporal data itself is a rapidly growing area with methods combining ideas from statistics and computer science \cite{Banerjee,Banerjee:Gelfand:Finley:Sang,Cressie:Johannesson,Datta:Banerjee:Finley:Gelfand,Datta et al.,Furrer:Sain,Kaufman:Schervish:Nychka,Paciorek et al.,Wikle:Cressie}. In another working paper we are currently trying to develop such a scalable version of the GRFDSTM.

\subsection{Moderate Size Spatio-Temporal Data and the GRFDSTM} 
In this subsection we implement the GRFDSTM to a moderate size dataset, simulated by a highly nonlinear spatio-temporal dynamic model. We generate $120$ random locations in $[0,1]\times[0,1]$ and at those spatial locations we generate spatial time series data for $t=1,2,\cdots,50$ according to the following nonlinear spatio-temporal model

\begin{align*}
Y_{t}(\bold{s}_{i})=&c+d\tan(\beta_{t}(\bold{s}_{i}))+\epsilon_{t}(\bold{s}_{i})\ \ \  \text{for}\  i=1,2,\cdots,n \ \text{and}\ t=1,2,\cdots,T \\
\beta_{t}(\bold{s}_{i})=&\sum_{j=1}^{n}a_{j}\beta_{t-1}(\bold{s}_{j})+\sum_{k=1}^{n}\sum_{l=1}^{n}b_{kl}\beta_{t-1}(\bold{s}_{k})\beta^{2}_{t-1}(\bold{s}_{l})+\eta_{t}(\bold{s}_{i});\ \ \  \text{for}\  i=1,2,\cdots,n\\
\boldsymbol{\beta}_{0}(\cdot) & \sim N(\bold{0},\boldsymbol{\Sigma}) ;\ \boldsymbol{\eta}_{t}(\cdot) \stackrel{iid}{\sim}  N(\bold{0},\boldsymbol{\Sigma});\ \text{and} \ \boldsymbol{\epsilon}_{t}(\cdot) \stackrel{iid}{\sim} N(\bold{0},\boldsymbol{\Sigma});
\end{align*}
where $\boldsymbol{\Sigma}$ is a spatial variance covariance matrix generated by a squared exponential covariance kernel with smoothness parameter $\lambda=1$ and variance $\sigma^{2}=1$. Once the data is completely generated, then we randomly select $20$ spatial locations among the $120$ locations and assign them as the test data locations. Hence, the spatial time series data generated at those $20$ spatial locations for time $t=1,2,\cdots,50$ is reserved as the test dataset and the rest of the dataset is kept as the training dataset. So, now we have a training dataset of size $5000$ on which we fit the GRFDSTM and predict the values of the test data. In fact, we gave the whole posterior predictive distribution for each of those $1000$ test data points. The fit was good and more than $95\%$ of test data  are correctly captured in the respective $95\%$ Bayesian symmetric prediction intervals. Below we present some graphs showing the overall performance of the GRFDSTM. 

\begin{figure}[H]
     \begin{center}
        \includegraphics[width=0.45\textwidth,height=0.30\textheight]{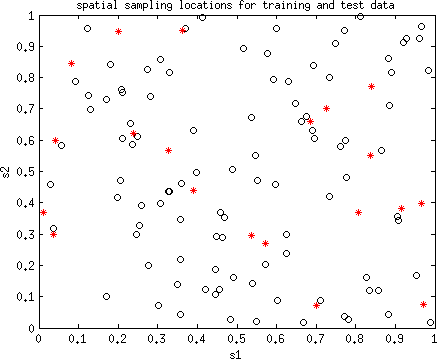}
        
    \end{center}
\caption{Spatial sampling locations in $[0,1]\times[0,1]$ for the moderate size spatio-temporal data. The black circles represent $100$ sptial training data locations and the red stars represent $20$ spatial test data locations.}	
    \label{fig:locations}
\end{figure}
As the above figure suggests, the test data locations are well representative of the training data locations. Regarding the model fitting and prediction we implement the GRFDSTM as specified in Section \ref{section:Identifiability}, i.e. we assume that all the covariance kernels are squared exponential type and $\sigma_{0}^2,\lambda_{0},\lambda_{f},\lambda_{g}$ are fixed. We estimate the rest of the parameters  ${\boldsymbol{\mu}}_{0}$, $\beta_{0g},\beta_{1g},\beta_{0f},\beta_{1f},\sigma_{f}^2,\sigma_{g}^2,\sigma_{\epsilon}^2,\sigma_{\eta}^2,\lambda_{\epsilon},\lambda_{\eta}$ using an MCMC algorithm. The prior structure is same as Section \ref{section:Prior_specification_fitting}. We take $N(\bold{0},\bold{I}_{n})$ prior for  ${\boldsymbol{\mu}}_{0}$, $N(0,1000)$ prior for each of $\beta_{0g},\beta_{0f}$ and $N(0.1,1000)$ prior for each of $\beta_{1g},\beta_{1f}$. We consider lognormal$(0.4,1.4)$ prior for each of $\sigma_{f}^2,\sigma_{g}^2,\sigma_{\epsilon}^2,\sigma_{\eta}^2$, lognormal$(0,0.08)$ prior for 
$\lambda_{\epsilon}$ and lognormal$(0,0.4)$ prior for $\lambda_{\eta}$. All the priors are independent. The prior hyper parameters are selected based on multiple MCMC pilot runs using a much smaller subset of the whole dataset and then taking the combination of prior hyper parameter values, which gives the best fit. The MCMC runs upto $30,000$ iterations and the first $14,000$ samples were discarded as burn-in samples. Below are some specimen pictures of the MCMC trace plots for some of the parameters. We also present the pictures of posterior distributions (histograms drawn based on the post burn-in samples) of the respective parameters.

\begin{figure}[H]
     \begin{center}
        \subfigure[]{%
            \label{fig:sigmae}
            \includegraphics[width=0.31\textwidth,height=0.16\textheight]{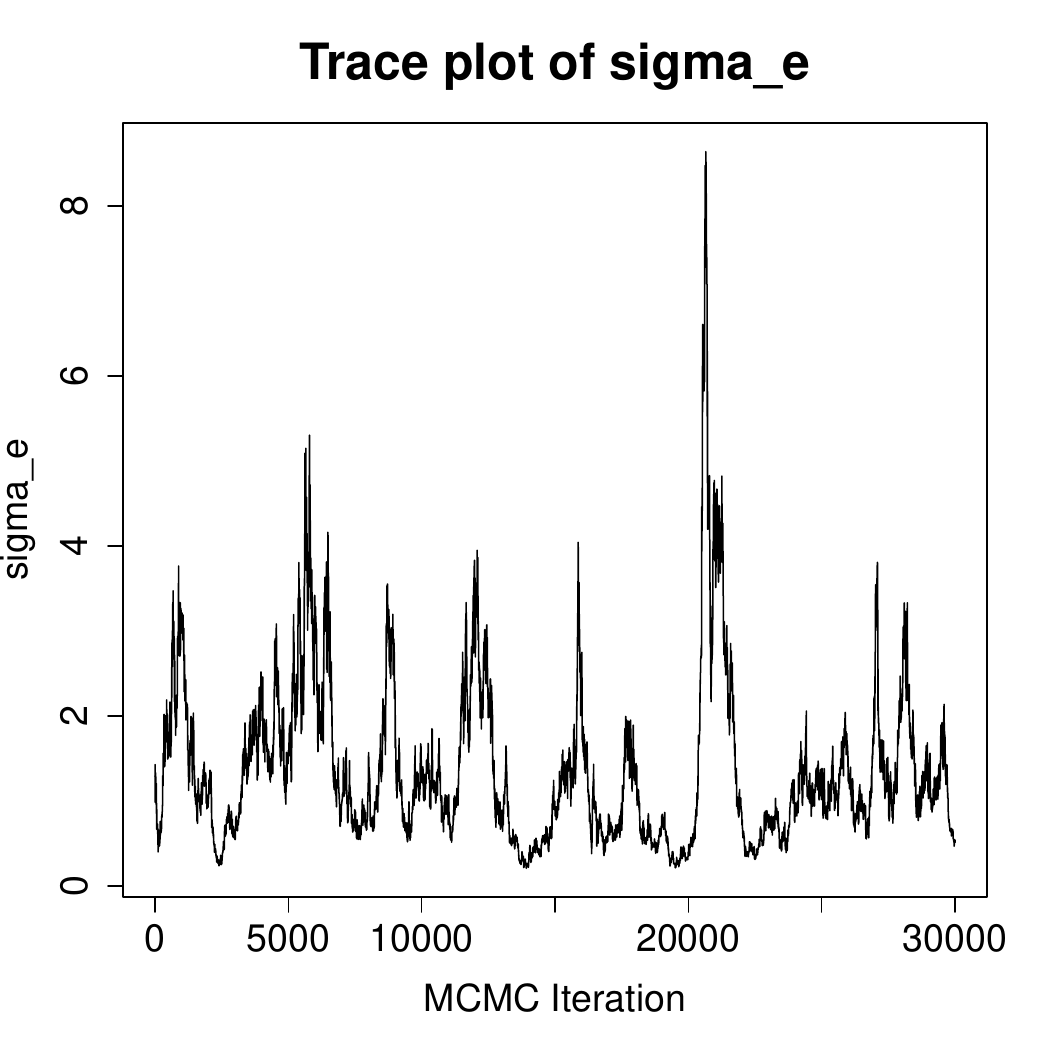}
        }%
        \subfigure[]{%
           \label{fig:sigmaf}
           \includegraphics[width=0.31\textwidth,height=0.16\textheight]{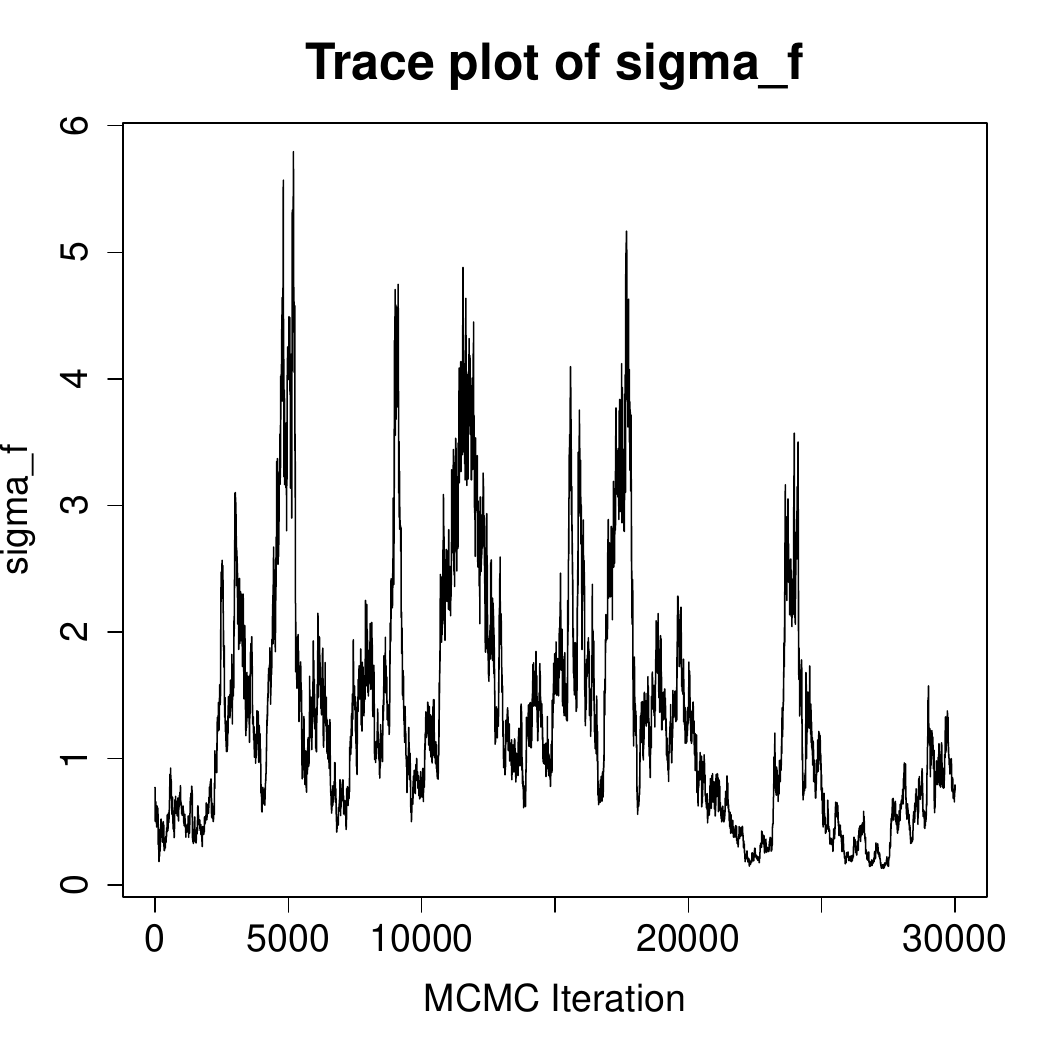}
        }%
        \subfigure[]{%
           \label{fig:smoothnesseta}
           \includegraphics[width=0.29\textwidth,height=0.16\textheight]{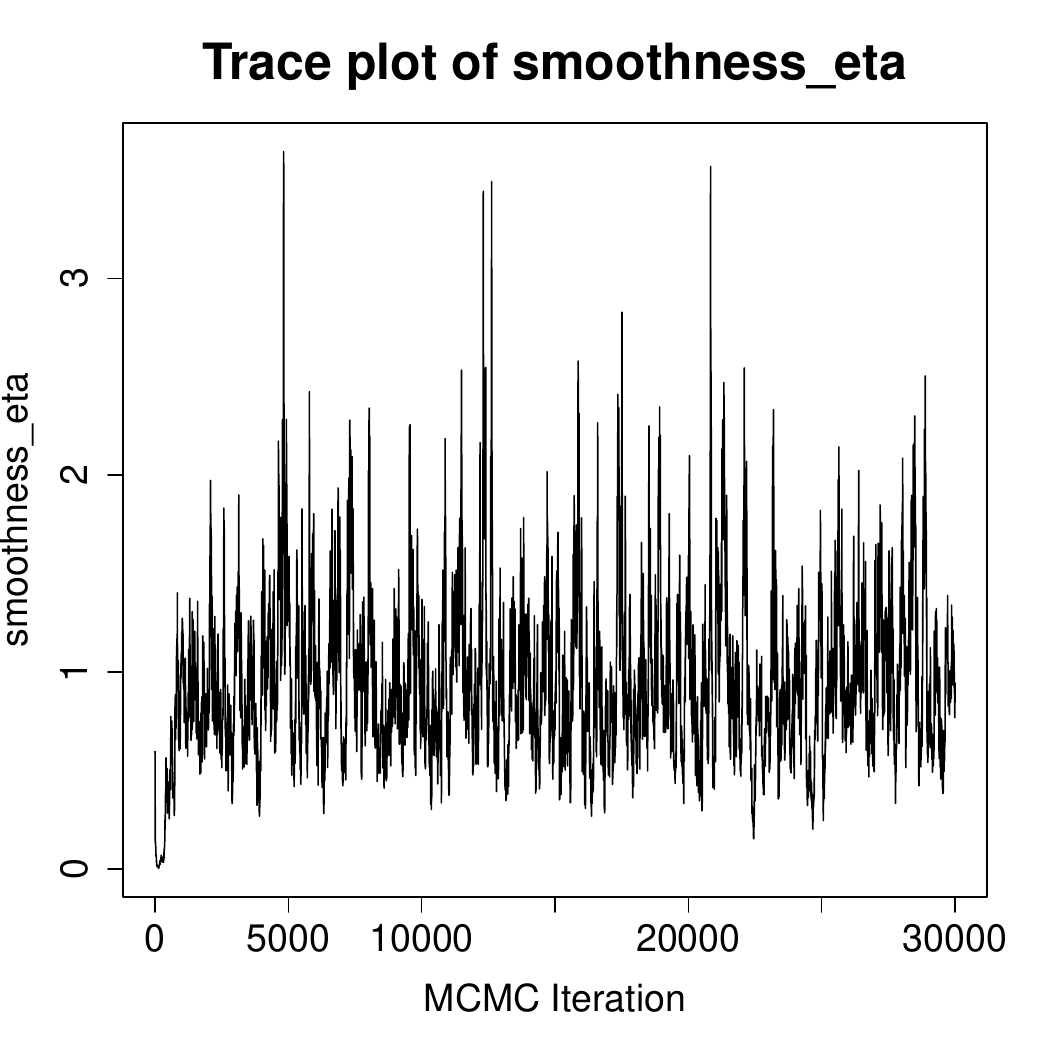}
        }\\ 

        \subfigure[]{%
            \label{fig:sigmaehist}
            \includegraphics[width=0.31\textwidth,height=0.16\textheight]{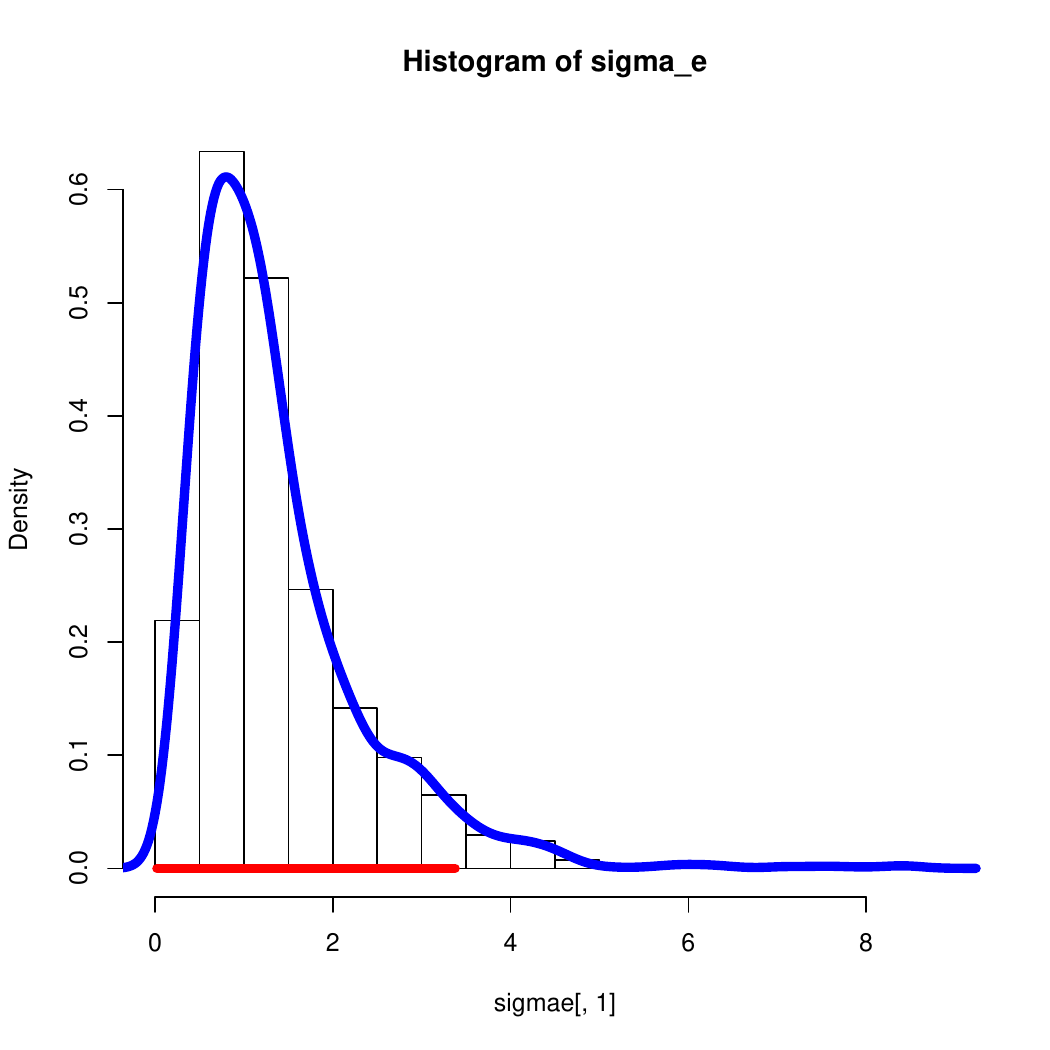}
        }%
        \subfigure[]{%
           \label{fig:sigmafhist}
           \includegraphics[width=0.31\textwidth,height=0.16\textheight]{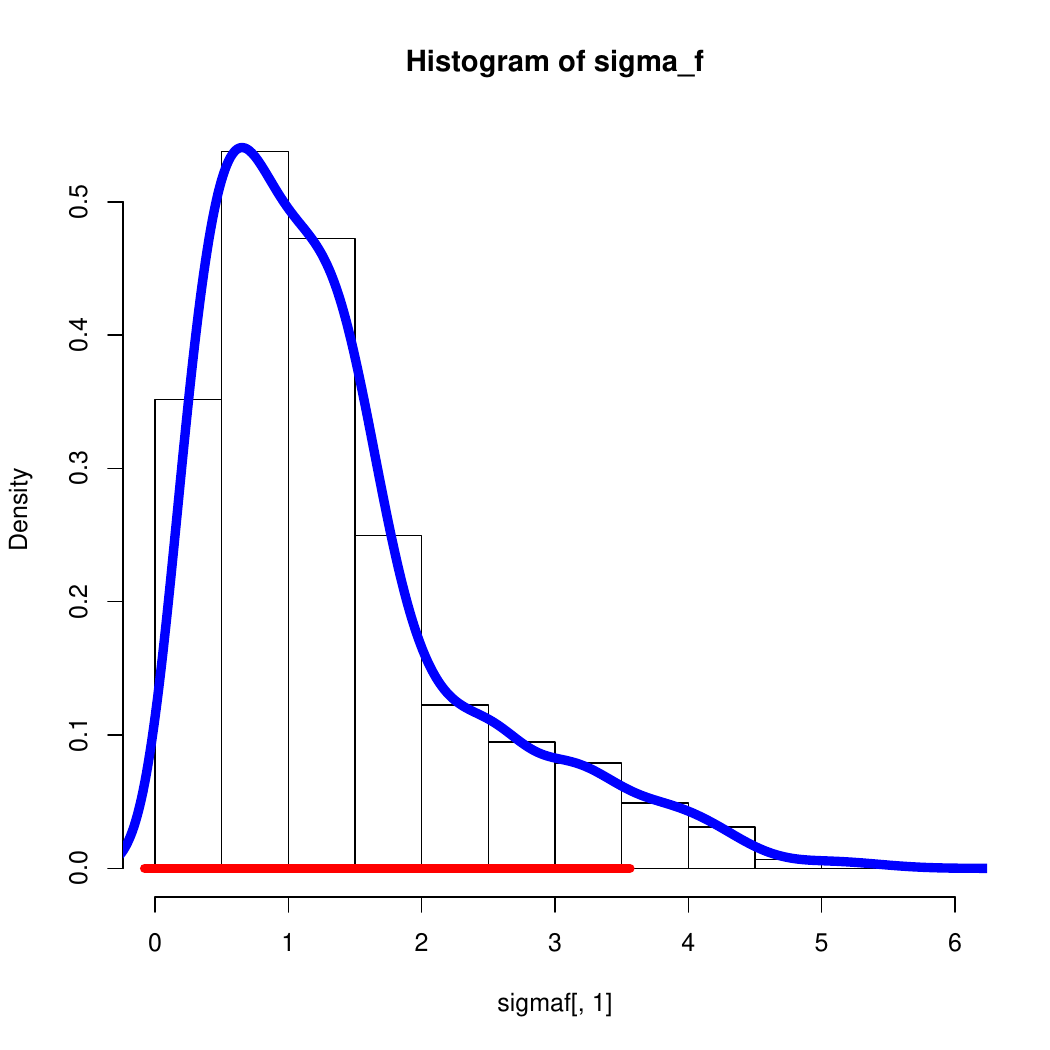}
        }%
        \subfigure[]{%
           \label{fig:smoothnessetahist}
           \includegraphics[width=0.29\textwidth,height=0.16\textheight]{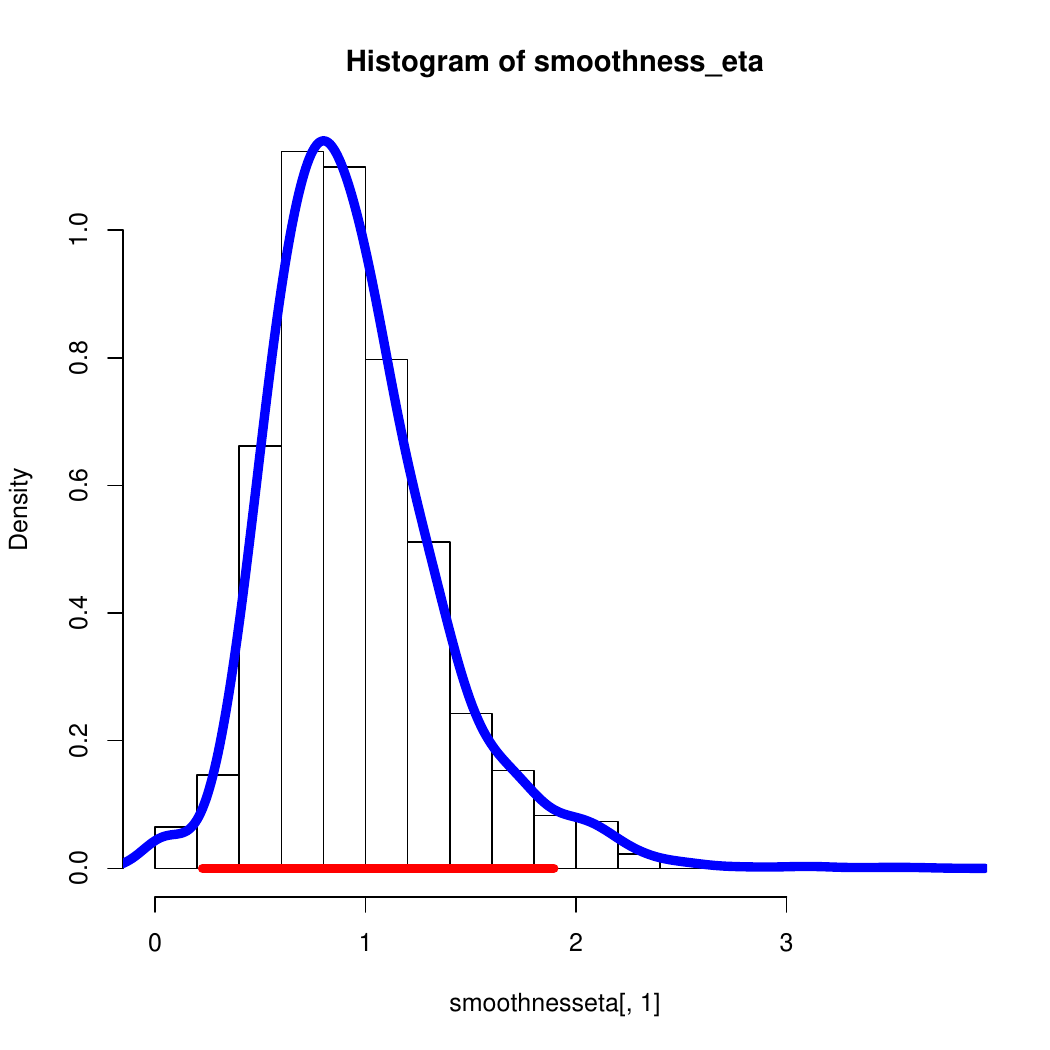}
        }\\ 
                 
    \end{center}
\caption{Panels (a), (b) and (c) correspond to the MCMC trace plots of parameter $\sigma^{2}_{\epsilon}$, $\sigma^{2}_{f}$ and $\lambda_{\eta}$ respectively. Panels (d), (e) and (f) correspond to the posterior distributions (histograms based on post burn in samples) of parameter $\sigma^{2}_{\epsilon}$, $\sigma^{2}_{f}$ and $\lambda_{\eta}$ respectively. The blue curve, superimposed on the respective histogram denotes the kernel density estimate based on Gaussian kernel and the red line denotes the HPD interval based on the kernel density estimator.}	
    \label{fig:subfigures5}
\end{figure}

Like the model parameters we also obtain the posterior distributions of the latent variables $X(\bold{s}_{i},t)$s associated with the training dataset, the latent variables $X(\bold{s}^{*}_{j},t)$s associated with the test dataset and the posterior predictive distributions of the observed test data $Y(\bold{s}^{*}_{j},t)$s. Below we present them (in forms of histograms). We also present the trace plots associated with them. Note, that there are a whole bunch of histograms and trace plots each associated with the latent and observed variables at each spatio-temporal coordinate. Here we present the histograms and trace plots associated with the latent and observed variable at particular spatio-temporal coordinate.
\begin{figure}[H]
     \begin{center}
        \subfigure[]{%
            \label{fig:x}
            \includegraphics[width=0.31\textwidth,height=0.16\textheight]{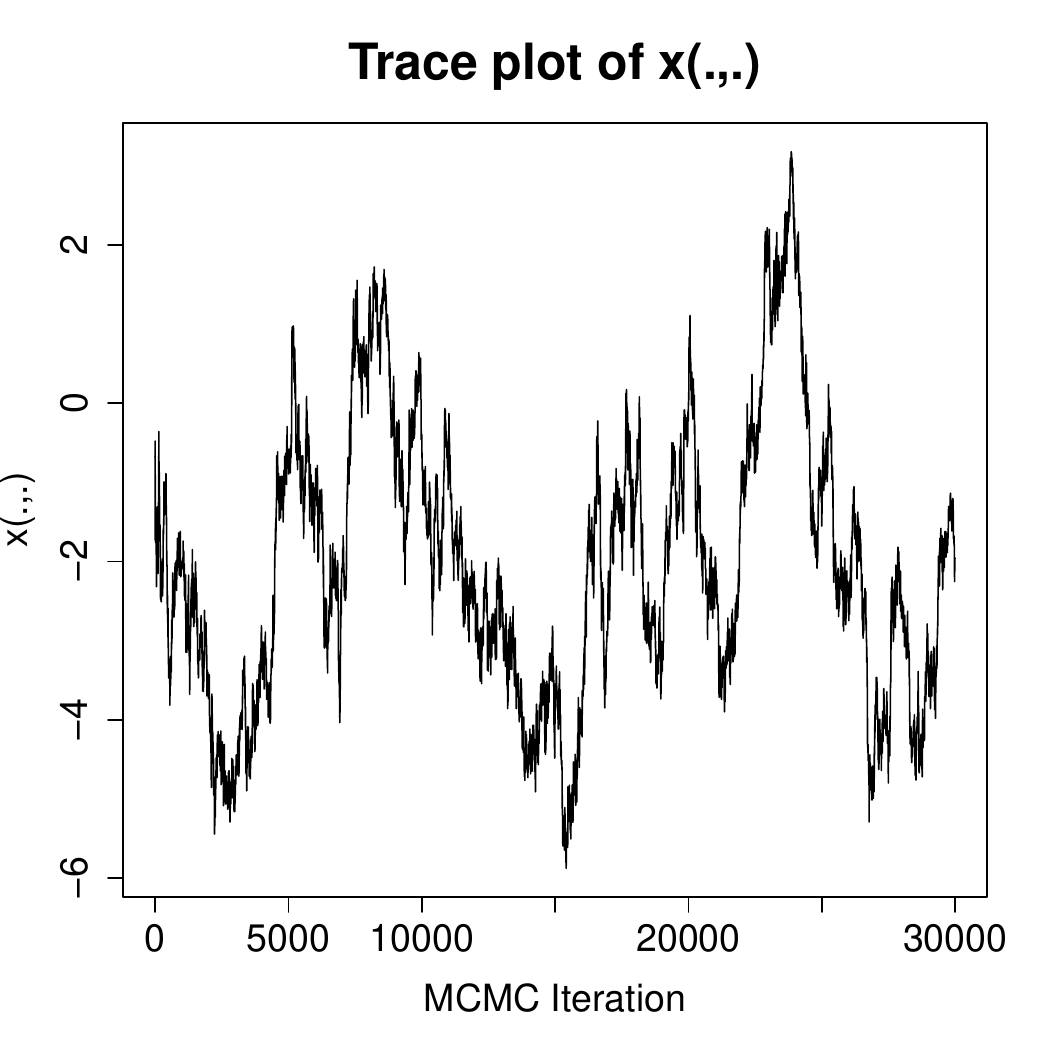}
        }%
        \subfigure[]{%
           \label{fig:xstar}
           \includegraphics[width=0.31\textwidth,height=0.16\textheight]{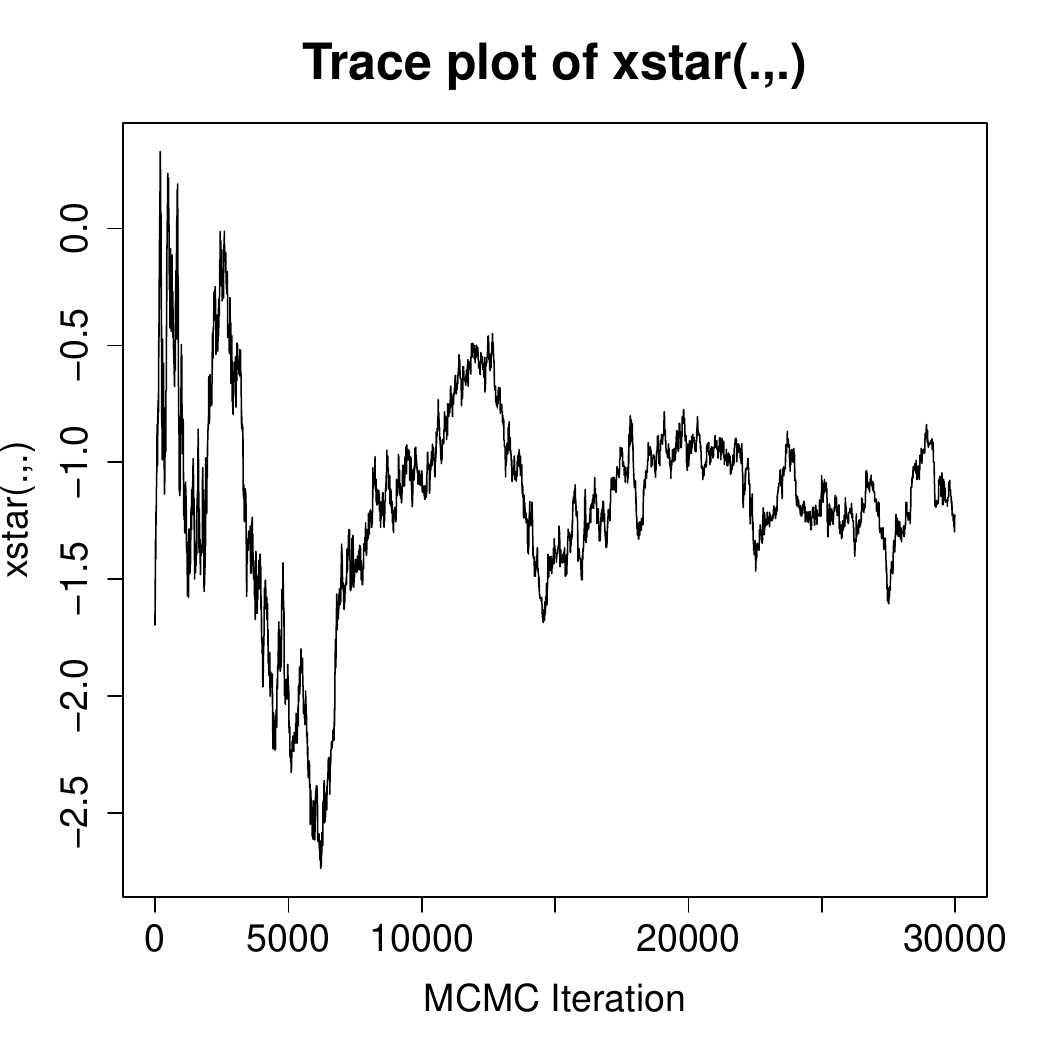}
        }%
        \subfigure[]{%
           \label{fig:ystar}
           \includegraphics[width=0.29\textwidth,height=0.16\textheight]{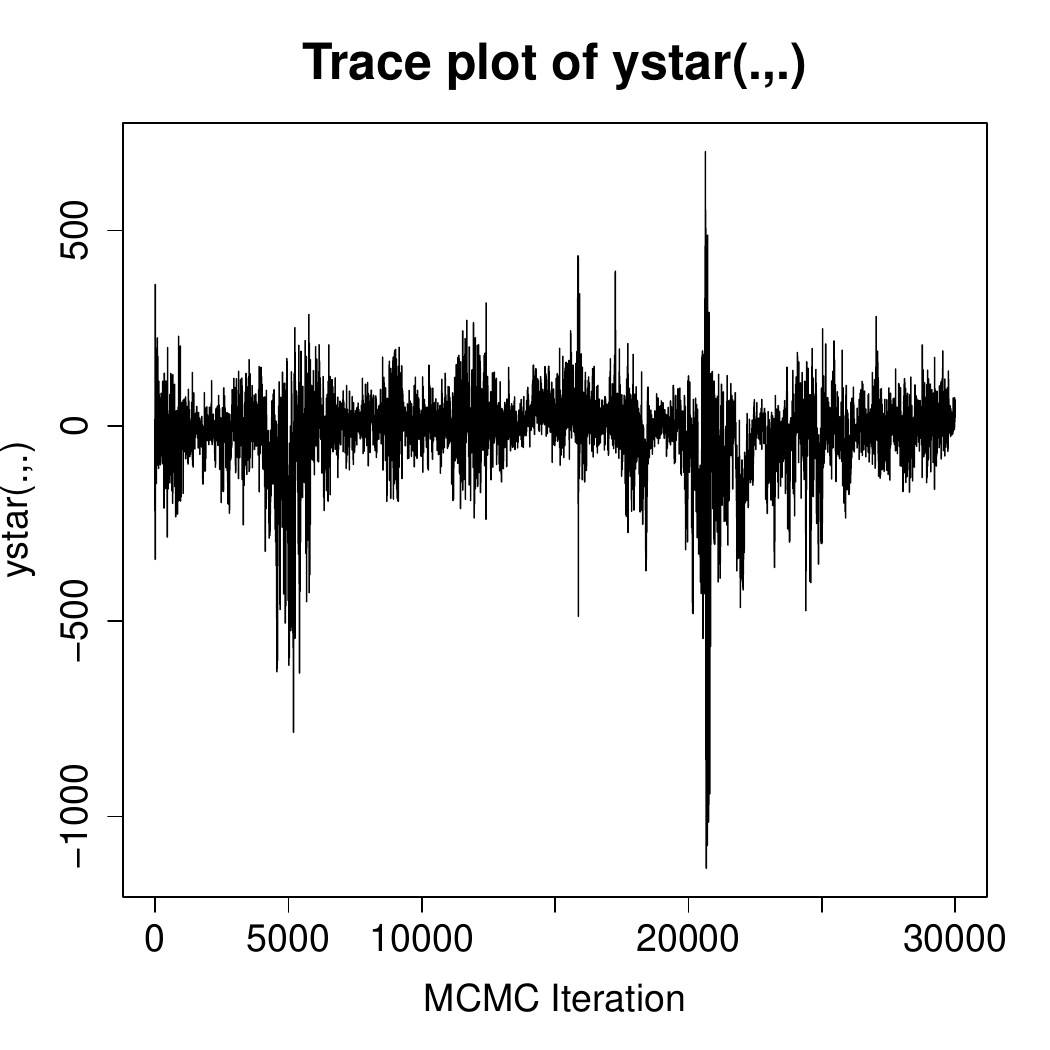}
        }\\ 

        \subfigure[]{%
            \label{fig:xhist}
            \includegraphics[width=0.31\textwidth,height=0.16\textheight]{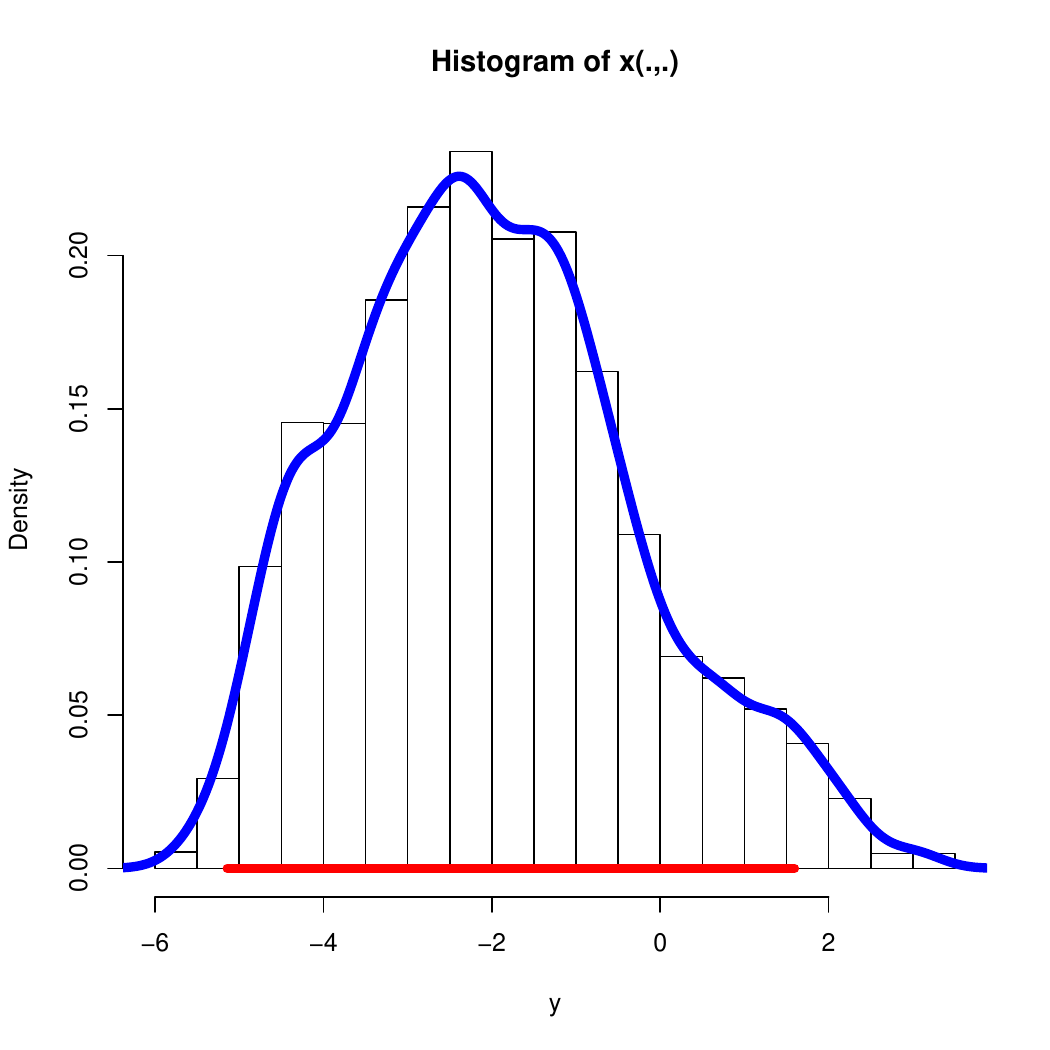}
        }%
        \subfigure[]{%
           \label{fig:xstarhist}
           \includegraphics[width=0.31\textwidth,height=0.16\textheight]{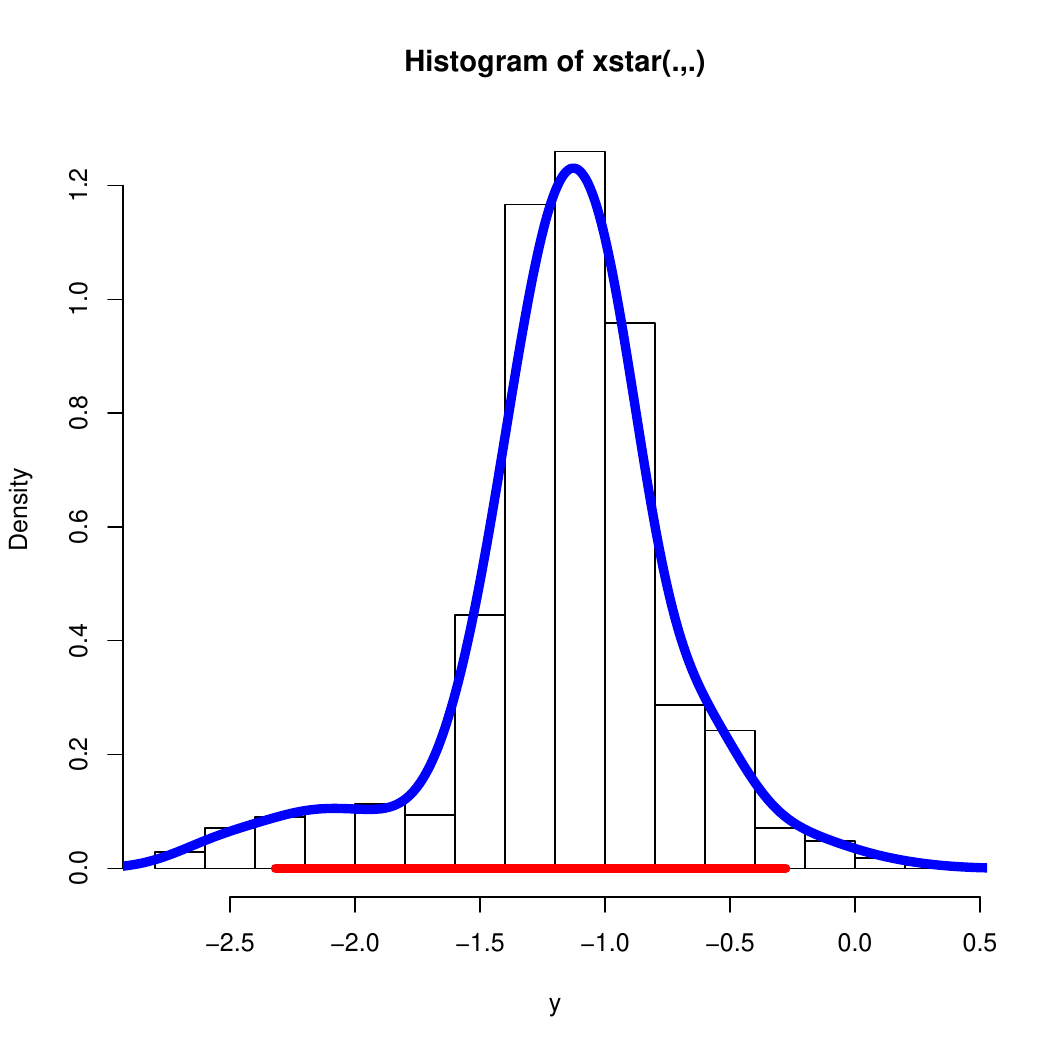}
        }%
        \subfigure[]{%
           \label{fig:ystarhist}
           \includegraphics[width=0.29\textwidth,height=0.16\textheight]{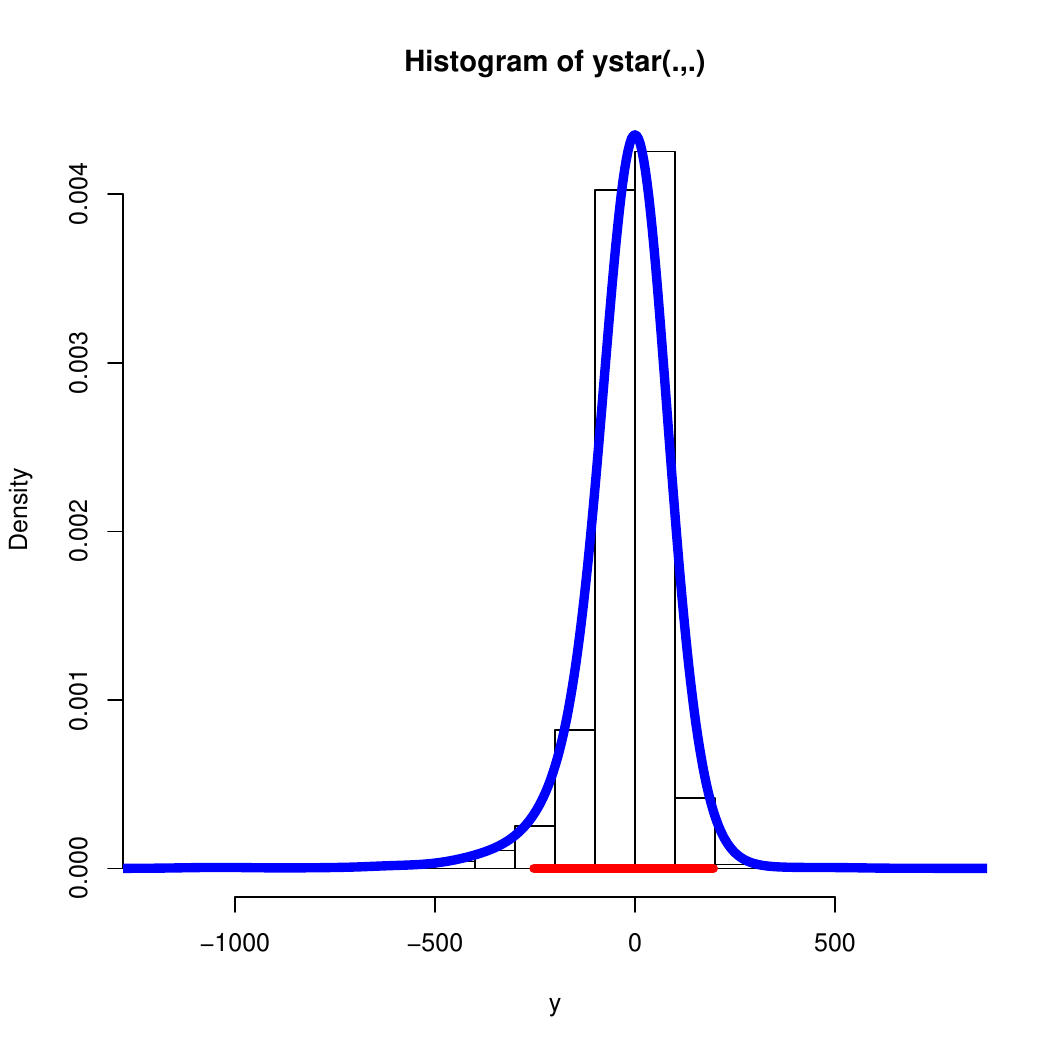}
        }\\ 
                 
    \end{center}
\caption{Panels (a), (b) and (c) correspond to the MCMC trace plots of $X(\bold{s}_{i},t)$, $X(\bold{s}^{*}_{j},t)$ and $Y(\bold{s}^{*}_{j},t)$ respectively. $X(\bold{s}_{i},t)$ is the latent variable at spatio-temporal training data location $(\bold{s}_{i},t)$ and $X(\bold{s}^{*}_{j},t)$ is the latent variable at spatio-temporal test data location $(\bold{s}^{*}_{j},t)$. $Y(\bold{s}^{*}_{j},t)$ denotes the observed variable at spatio-temporal test data location $(\bold{s}^{*}_{j},t)$. Panels (d), (e) and (f) correspond to the posterior distributions (histograms based on post burn in samples) of $X(\bold{s}_{i},t)$, $X(\bold{s}^{*}_{j},t)$ and $Y(\bold{s}^{*}_{j},t)$ respectively. The blue curve, superimposed on the respective histogram denotes the kernel density estimate based on Gaussian kernel and the red line denotes the HPD interval based on the kernel density estimator.}	
    \label{fig:subfigures6}
\end{figure}

Finally, we present the time series of test data and corresponding $95\%$ Bayesian symmetric prediction band at spatial locations $\bold{s}_{19}$, $\bold{s}_{41}$ and $\bold{s}_{49}$. The time series of test data is well within the $95\%$ Bayesian symmetric prediction band at all of the $3$ locations. Similar result is observed for other spatial test data locations as well.
\begin{figure}[H]
     \begin{center}
        \subfigure[]{%
            \label{fig:s19}
            \includegraphics[width=0.31\textwidth,height=0.16\textheight]{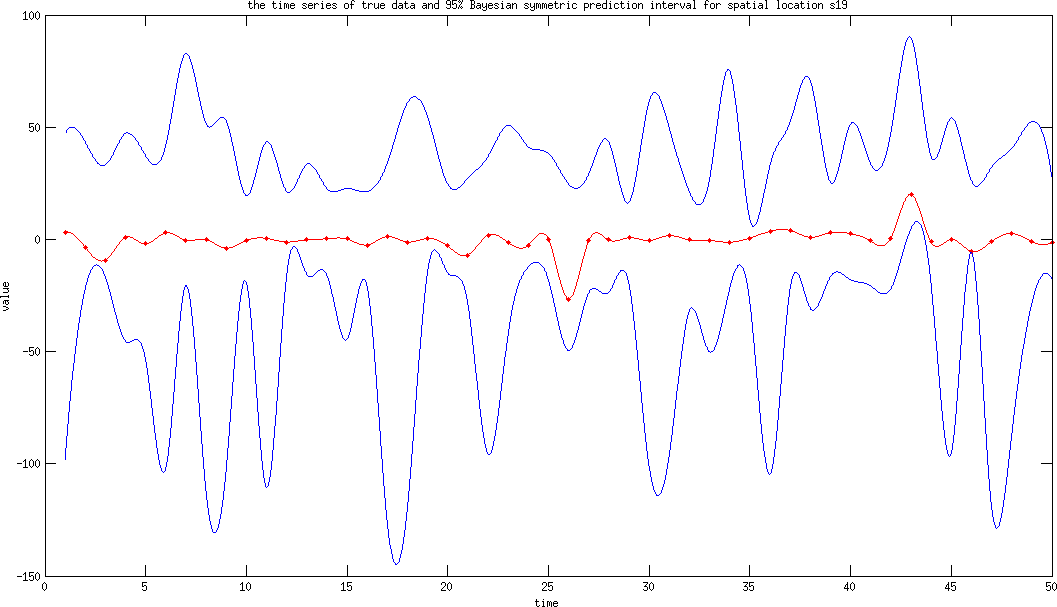}
        }%
        \subfigure[]{%
           \label{fig:s41}
           \includegraphics[width=0.31\textwidth,height=0.16\textheight]{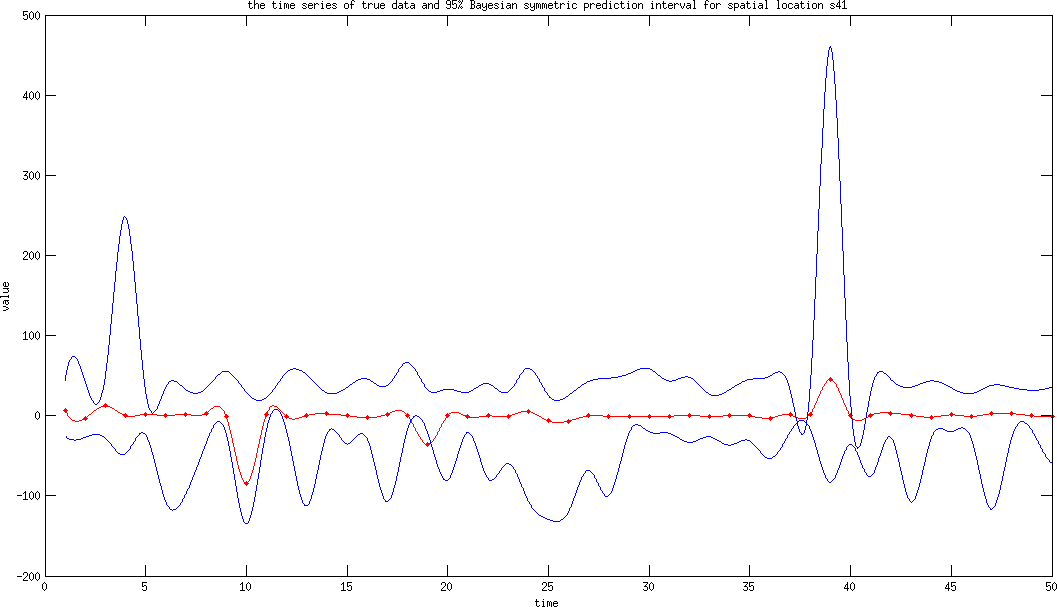}
        }%
        \subfigure[]{%
           \label{fig:s49}
           \includegraphics[width=0.29\textwidth,height=0.16\textheight]{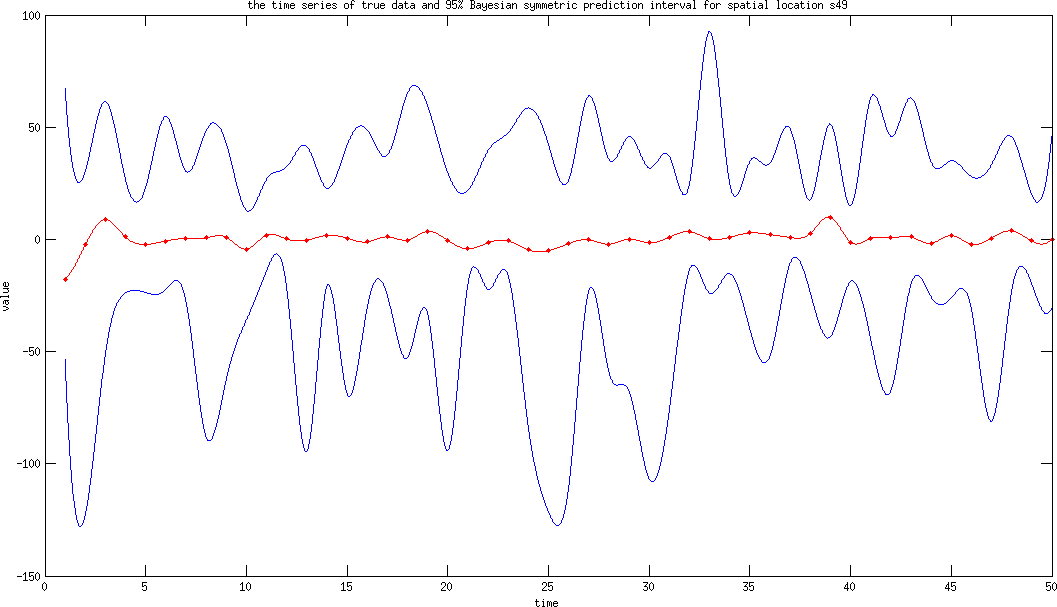}
        }\\ 

    \end{center}
\caption{Specimen pictures showing test data and prediction bands. Panel
(a) displays the test data and the 95\% Bayesian symmetric
prediction band given by the GRFDSTM at a particular representative
spatial test data location $\bold{s}_{19}$ for $t=1,2,\cdots,50$; the red starred line represents the time series of the test data at test
data location $\bold{s}_{19}$ for $t=1,2,\cdots,50$ and a smooth red curve is interpolated through the test data; the blue
band represents the prediction band associated with the GRFDSTM. Panels (b) and (c) display similar plots but for spatial test data locations $\bold{s}_{41}$ and $\bold{s}_{49}$ respectively.}	
    \label{fig:subfigures7}
\end{figure}

\subsection{Realdata Analysis : $SO_{2}$ Pollution over Europe}
Air pollution over large geographical regions is a topic of wide range of studies involving statistics and other disciplines. 
Among them, statistical modeling of pollution caused by $SO_{2}$ draws considerable attention. Here we consider a $SO_{2}$ 
pollution dataset over Europe. The dataset consists of monthly measurement of sulphur dioxide 
pollution observed at $46$ monitoring stations spread over Europe, for $60$ months, starting from January 1997 to December 2001. This dataset is a part of the data collected through the `European monitoring and evaluation programme' (EMEP) 
which co-ordinates 
the monitoring of airborne pollution over Europe. Further information is available at 
{\tt {http://www.emep.int}}, and the dataset is freely available there. Simple exploratory analysis reveals the presence of seasonal component and negligible trend component. It is important to note that instead of direct measurement what we have is the measurement of natural logarithm of $SO_{2}$ pollution for these stations. 

 \begin{figure}[H]
     \begin{center}
        \subfigure[]{%
            \label{fig:realmap}
            \includegraphics[width=0.40\textwidth,height=0.20\textheight]{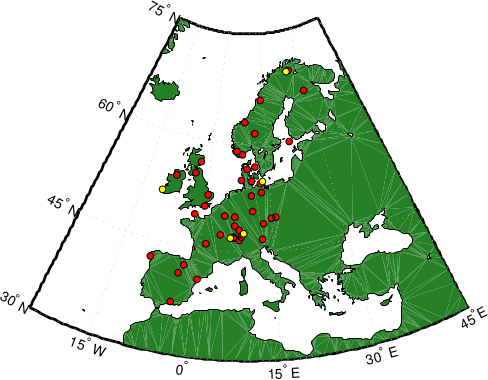}
        }
    \end{center}
    \caption{%
      Geographical location of the monitoring sites. The yellow circles indicate the sites with large proportion of missing data.
        }%
   \label{fig:subfigures5}
\end{figure}

The locations of the sites are shown in Figure \ref{fig:subfigures5}. 
Among the $46$ monitoring sites, there are $5$ sites (yellow circles in Figure \ref{fig:subfigures5}) where there are large proportions of missing data, making the data associated with those sites highly unreliable. So, we fit the GRFDSTM as specified in Section \ref{section:Prior_specification_fitting} to the available data and fully reconstruct the $SO_{2}$ pollution time series at those $5$ sites. The prior structure used for fitting the GRFDSTM is similar to that of the simulation study. We used diffused bivariate normal priors for $(\beta_{0g},\beta_{1g})$ and $(\beta_{0f},\beta_{1f})$. We consider the squared exponential covariance kernel with the following representation $c(\bold{u},\bold{v})=\sigma^2 e^{-\lambda||\bold{u}-\bold{v}||^2}$. Associated with five covariance kernels  $c_{f}(\cdot,\cdot),c_{g}(\cdot,\cdot),c_{\epsilon}(\cdot,\cdot),c_{\eta}(\cdot,\cdot), c_{0}(\cdot,\cdot)$ we have five scale parameters $\sigma_{f}^2,\sigma_{g}^2, \sigma_{\epsilon}^2,\sigma_{\eta}^2,\sigma_{0}^2$ and five 
smoothness parameters $\lambda_{f},\lambda_{g},
\lambda_{\epsilon},\lambda_{\eta},\lambda_{0}$. Recall that among them $\lambda_{f},\lambda_{g},\lambda_{0},\sigma_{0}^2$ are not really parameters but some fixed entities. For the rest of the scale and smoothness parameters, i.e. $\sigma_{f}^2,\sigma_{g}^2, \sigma_{\epsilon}^2,\sigma_{\eta}^2,\lambda_{\epsilon},\lambda_{\eta}$, we use independent lognormal priors. As discussed earlier, the lognormal priors associated with $\lambda_{\epsilon},\lambda_{\eta}$ are taken to be highly concentrated to avoid ill-conditioning of large covariance matrices, encountered during MCMC computation. The hyper parameters associated with the independent priors are selected based on minimizing $D_{0.5}$ goodness of fit statistic over multiple MCMC pilot runs on smaller datasets. 

Since the monitoring stations are spread over a large geographical region and distance stretches horizontally as latitude 
increases, the use of simple longitude and latitude as spatial coordinates would not be appropriate. 
The Lambert (or Schmidt) projection addresses this
problem by preserving the area. This projection is defined by the transformation of longitude and latitude, 
expressed in radians as $\psi $ and $\phi$, to the new co-ordinate system 
$\mathbf{s}=(2\sin(\frac{\pi}{4}-\frac{\phi}{2})\sin(\psi),-2\sin(\frac{\pi}{4}-\frac{\phi}{2})\cos(\psi))$. 
However, with respect to the temporal coordinate we simply take one month as one unit of time.

We implemented an MCMC chain with sufficiently large burn-in ($60,000$); visual inspection and tests based on CODA package suggests satisfactory convergence. Based on the $20,000$ post burn-in MCMC samples, we calculate the 95\% Bayesian symmetric prediction intervals associated with the posterior predictive distributions of $\ln{SO_{2}}$ at the $5$ sites, where we have a large proportion of missing data. 

Different scenarios are encountered at those $5$ monitoring sites, where we reconstruct the $\ln{SO_{2}}$ time series. At site $CH03$, the data is available for the first $36$ months of the time period covered. For site $NO30$, data associated with only 
the first three months is available. Data is available both at the beginning and at the end of the time period considered, for the site $IE01$ and we don't have any data for site $DE06$. Diagrams of the fully reconstructed time series at three monitoring sites are presented in Figure \ref{fig:subfigures6}. 

\begin{figure}[H]
     \begin{center}
        \subfigure[]{%
            \label{fig:DE06}
            \includegraphics[width=0.31\textwidth,height=0.16\textheight]{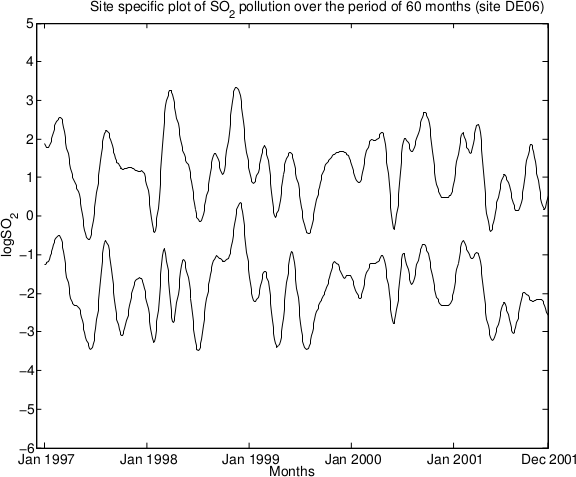}
        }%
        \subfigure[]{%
           \label{fig:IE01}
           \includegraphics[width=0.31\textwidth,height=0.16\textheight]{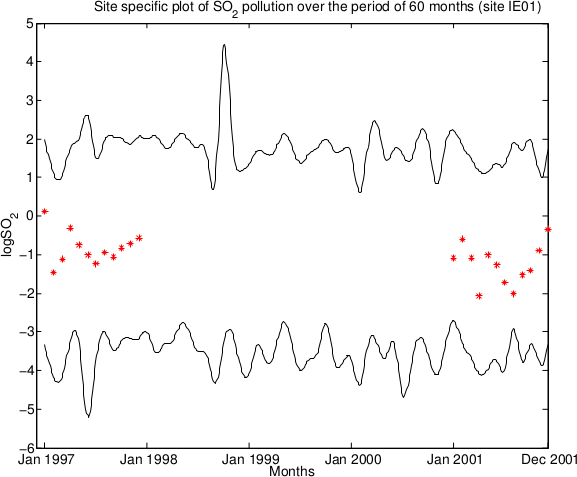}
        }%
        \subfigure[]{%
           \label{fig:NO30}
           \includegraphics[width=0.31\textwidth,height=0.16\textheight]{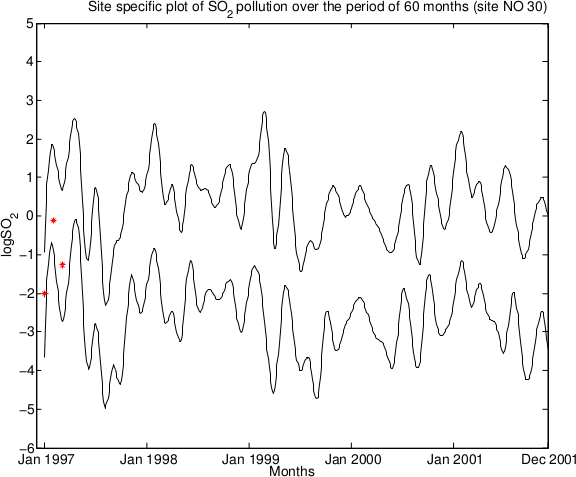}
        }\\ 

    \end{center}
    \caption{Reconstruction of time series of $logSO_{2}$ at monitoring sites $DE06, IE01$ and $NO30$. the black band represent the $95\%$ Bayesian symmetric prediction interval obtained by fitting the GRFDSTM. The true data points, whenever available, are denoted by red stars.}	
    \label{fig:subfigures6}
\end{figure}

The prediction band corresponding to the monitoring site $IE01$ is much wider and the reconstructed time series at that location is less accurate owing to the fact that its neighbouring stations are quite far away compared to the other sites. Note that in calculating the prediction interval we consider the $95\%$ Bayesian symmetric prediction interval and not the $95\%$  
highest posterior density (HPD) interval. The computational cost for HPD is much larger but the gain is little since the spatio-temporal location specific predictive distributions are almost symmetric and there is little difference between Bayesian symmetric prediction interval and the HPD interval.

\section{Discussion and Concluding Remarks}

Discrete-time spatial time series data often comes with additional covariate information. A simple modification of the GRFDSTM in the following way is capable of handling the covariates 
\begin{align*}
     Y(\bold{s},t)&=f(X(\bold{s},t),Z(\bold{s},t))+\epsilon(\bold{s},t);\\
     X(\bold{s},t)&=g(X(\bold{s},t-1))+\eta(\bold{s},t),
\end{align*}
where $Z(\bold{s},t)$ is the covariate process and $f(\cdot,\cdot)$ is now a Gaussian function on $\mathbb{R}^2$. The rest of the theory remains the same as before. In fact, we can have $k$ different covariate processes 
$Z_{1}(\bold{s},t),Z_{2}(\bold{s},t),\cdots,Z_{k}(\bold{s},t)$, in which case $f(\cdot,\cdot,\cdots,\cdot)$ will be a Gaussian function on $\mathbb{R}^{k+1}$. 
Note that, for this model, the number of regression parameters in the mean function of the modeled Gaussian
process and the number of smoothness parameters increase with the number of covariates, but since the dimensions 
of the matrices remain unchanged, the computational cost does not increase significantly. 
Another direction for extension is the time varying version of the GRFDSTM. Such modification is particularly useful when the time interval on which we are observing the data is very wide so that it is unlikely for the functions $f$ and $g$ to remain invariant with respect to time. In the context of purely temporal state space models, \cite{Ghosh14} 
consider time-varying functions $f_t(\cdot)$ and $g_t(\cdot)$, which they re-write as $f(\cdot,t)$ and $g(\cdot,t)$,
respectively. In other words, they consider the time component as an argument of their random functions $f$ and $g$, which are modeled by Gaussian functions in the usual manner. Such ideas can be easily adopted in the GRFDSTM. Currently, we are working on these extensions.

In brief, we recommend using an appropriate NLDSTM model when one clearly knows that the underlying dynamics is nonlinear and also has strong knowledge about the dynamics. The GRFDSTM with possibly nonlinear mean functional associated with $g(\cdot)$ and $f(\cdot)$ may be considered alternatively, but the performance is generally sub-optimal. However, under no or weak knowledge regarding the dynamics under the spatio-temporal process the GRFDSTM with linear mean functional would exhibit superior predictive performance. The computational cost is more than the prediction by the LDSTMs or NLDSTMs, but the gain is substantial.

\section*{Acknowledgments} The research of the first author is fully funded by CSIR SPM Fellowship SPM-07/093(0079)/2010-EMR-I, Govt. of India. We thank `European monitoring and evaluation programme' (EMEP) for the $SO_{2}$ dataset. We are thankful to Moumita Das for many fruitful discussions on this article.

\section*{Appendix}

Before proving the theorems let us make a notational clarification. The notations $[\bold{X}|\bold{Y}],[\bold{x}|\bold{y}]$ 
and $[\bold{X}=\bold{x}|\bold{Y}=\bold{y}]$ are equivalent and throughout this section they will denote the value of conditional pdf 
of $\bold{X}$ given $\bold{Y}=\bold{y}$ at $\bold{X}=\bold{x}$.

\begin{proof}[Proof of Proposition \ref{propn:measurable}] 

For this proof we assume that there exist continuous modifications of the Gaussian processes that we consider, that is,
there exist processes with sample paths that are continuous everywhere, not just almost everywhere, and
that our Gaussian processes equals such processes with probability one (see, for example, see \cite{Adler07} for details). Existence of such continous modifications are guaranteed under the correlation structure that we consider for our Gaussian processes. Note that sometimes the term continuous modification is used in a weaker sense, i.e. to denote a version that is continuous almost everywhere. However, a continuous everywhere version is obtained easily from that by mapping a continuous sample path to itself and a discontinuous sample path to zero function.

Let us first notice that, $\exists$ a probability space $(\Omega,\mathcal{F},P)$ such that
$$(g(x),X(\bold{s}_{1},0),\cdots,X(\bold{s}_{n},0)):
(\Omega,\mathcal{F})\rightarrow (C(\mathbb{R}),\mathcal{A})\bigotimes (\mathbb{R}^n,\mathcal{B}(\mathbb{R}^n))$$ 
where $C(\mathbb{R})$ is the space of all real valued continuous functions on $\mathbb{R}$ and $\mathcal{A}$ is the 
Borel sigma field obtained from the topology of compact convergence on the space $C(\mathbb{R})$. Such a joint measurability 
result need not hold unless $g(x)$ and $(X(\bold{s}_{1},0),\cdots,X(\bold{s}_{n},0))$ are independent. Now, we will show 
that $g(X(\bold{s}_{1},0))$ is a measurable real valued function or a proper random variable. To show this, first see that 
$g(X(\bold{s}_{1},0)):\Omega\rightarrow\mathbb{R}$ can be written as $T(g(\cdot),X(\bold{s}_{1},0))$ where 
$T:C(\mathbb{R})\bigotimes\mathbb{R}\rightarrow\mathbb{R}$ is a transformation such that, $T(g,x)=g(x)$ where $g$ is a real valued continuous function on $\mathbb{R}$ and $x$ is a real number.
\begin{lemma}
$T:C(\mathbb{R})\bigotimes\mathbb{R}\rightarrow\mathbb{R}$ is a continuous transformation where the topology associated with 
$C(\mathbb{R})$ is the topology of compact convergence and the topology associated with $\mathbb{R}$ is the usual Euclidean 
distance based topology on real numbers.
\label{lem:measurable}
\end{lemma}
\begin{proof}[Proof of Lemma \ref{lem:measurable}] 
Let us consider the metric $d(g,g^\prime)=\mathlarger{\sum_{i=1}^{\infty}\frac{1}{2^i}
\frac{\sup\limits_{x\in [-i,i]}\abs{g(x)-g^\prime(x)}}{1+\sup\limits_{x\in[-i,i]}\abs{g(x)-g^\prime(x)}}}$. This metric induces the 
topology of compact convergence on the space $C(\mathbb{R})$. To prove continuity of $T$ one needs to show that 
$d(g_{n},g)\rightarrow 0$ and $\abs{x_{n}-x}\rightarrow 0 \Rightarrow \abs{T(g_{n},x_{n})-T(g,x)}\rightarrow 0$.
\\ 
Let us assume that $d(g_{n},g)\rightarrow 0$ and $\abs{x_{n}-x}\rightarrow 0$. So, $\exists \ N_{0}$ and $j_{0}$ such that 
$\forall \ n\geq N_{0}$, $x_{n}\in [-j_{0},j_{0}]$.
\\ 
Now, $\mathlarger{\frac{1}{2^{j_{0}}}}\sup\limits_{x\in [-j_{0},j_{0}]}\abs{g_{n}(x)-g(x)}\leq d(g_{n},g) \Rightarrow 
\sup\limits_{x\in [-j_{0},j_{0}]}\abs{g_{n}(x)-g(x)}\rightarrow 0$ 
\\ 
and $\abs{g(x_{n})-g(x)}\rightarrow 0$ because $g$ is continuous.
\\ 
But,
\begin{align*}
\abs{g_{n}(x_{n})-g(x)} \leq \abs{g_{n}(x_{n})-g(x_{n})} + \abs{g(x_{n})-g(x)}\\
\text{So,}\ \forall \ n\geq N_{0},\  \abs{g_{n}(x_{n})-g(x)} \leq \sup\limits_{x\in [-j_{0},j_{0}]}\abs{g_{n}(x)-g(x)} + \abs{g(x_{n})-g(x)}\\
\end{align*}
The RHS goes to $0$ as $n\rightarrow\infty$. Hence, $\abs{T(g_{n},x_{n})-T(g,x)}=\abs{g_{n}(x_{n})-g(x)}\rightarrow0$.
\end{proof}
Once continuity of $T$ is proved, note that $T^{-1}(U)$, for any open set $U \subseteq \mathcal{\mathbb{R}}$, is an open set 
in the product topology on the space $C(\mathbb{R})\bigotimes\mathbb{R}$. Hence, $T^{-1}(U)$ belongs to the Borel sigma field 
generated by this product topology which is in this case equivalent to the
product sigma field $\mathcal{A}\bigotimes\mathcal{B}(\mathbb{R})$ associated with $(C(\mathbb{R}),\mathcal{A})\bigotimes (\mathbb{R},\mathcal{B}(\mathbb{R}))$. 
This equivalence holds because both of the spaces $C(\mathbb{R})$ and $\mathbb{R}$ are separable. 
But $(g(x),X(\bold{s}_{1},0))$ is measurable with respect to $(\Omega,\mathcal{F})$ and $(C(\mathbb{R}),\mathcal{A})
\bigotimes (\mathbb{R},\mathcal{B}(\mathbb{R}))$. Hence, the inverse image of $T^{-1}(U)$ with respect to $(g(x),X(\bold{s}_{1},0))$ 
is in $\mathcal{F}$. So, the inverse image of any open set $U \subseteq \mathcal{\mathbb{R}}$ with respect to 
$g(X(\bold{s}_{1},0))$ is in $\mathcal{F}$. This proves the measurability of $g(X(\bold{s}_{1},0))$.

Following exactly same argument as above we can further prove that $g(X(\bold{s}_{2},0)),\cdots,g(X(\bold{s}_{n},0))$ are jointly measurable. 
Now, as $\eta(\bold{s},t)$ is independent of $g(x)$ and $(X(\bold{s}_{1},0),\cdots,X(\bold{s}_{n},0))$, we have the joint measurability of
$(X(\bold{s}_{1},1),\cdots,X(\bold{s}_{n},1))$ (See (\ref{eqn:npr2})). Infact, we can prove that $(g(x),X(\bold{s}_{1},1),\cdots,X(\bold{s}_{n},1))$ are jointly measurable. To do it we consider $T^{\prime}:C(\mathbb{R})\bigotimes\mathbb{R}\rightarrow C(\mathbb{R})\bigotimes\mathbb{R}$ such that $T^{\prime}(g,x)=(g,g(x))$ where $g$ is a real valued continuous function on $\mathbb{R}$ and $x$ is a real number. Then similarly as in case of $T$ we can prove that $T^{\prime}$ is also a continuous map which immediately implies that $(g,g(X(\bold{s}_{1},0)))$ are jointly measurable. Then $\eta(\bold{s},t)$ being independent of $g(x)$ and $(X(\bold{s}_{1},0),\cdots,X(\bold{s}_{n},0))$, implies the joint measurability of
$(g(x),X(\bold{s}_{1},1),\cdots,X(\bold{s}_{n},1))$.
Hence, starting with the joint measurability of 
$(g(x),X(\bold{s}_{1},0),\cdots,X(\bold{s}_{n},0))$ we prove the joint measurability of $(X(\bold{s}_{1},1),\cdots,X(\bold{s}_{n},1))$ and $(g(x),X(\bold{s}_{1},1),\cdots,X(\bold{s}_{n},1))$. 
Similarly, if we start with joint measurability of $(g(x),X(\bold{s}_{1},1),\cdots,X(\bold{s}_{n},1))$ we can prove the joint measurability 
of $(X(\bold{s}_{1},2),\cdots,X(\bold{s}_{n},2))$. Thus, the joint measurability of the whole collection of state variables 
${X(\bold{s}_{i},t)}$ $\forall\ i=1,2,\cdots,n;t=0,1,\cdots,T $ is mathematically established.

Now, to prove joint measurability of the collection of observed variables
${Y(\bold{s}_{i},t)}$ $\forall\ i=1,2,\cdots,n;t=1,\cdots,T $ recall the observational equation (1). 
Since, $f(x)$ takes values in $(C(\mathbb{R}),\mathcal{A})$ just as $g(x)$, 
and also since $\epsilon(\bold{s},t)$ is independent of $f(s)$ just as $\eta(\bold{s},t)$ is independent of $g(x)$, 
all the previous arguments go through in this case and joint measurability of ${Y(\bold{s}_{i},t)}$ $\forall\ i=1,2,\cdots,n;t=0,1,\cdots,T $ 
is established.

Finally, it remains to show that a valid spatio-temporal process is induced by this model. But it is immediate from the application 
of Kolmogorov consistency theorem. The consistency conditions of the theorem are trivially satisfied by our construction and hence the result follows.
\end{proof}

\begin{proof}[Alternative proof of Proposition \ref{propn:measurable}] 

In the previous proof, we need to assume the existence of continuous modification of the underlying Gaussian processes.
Here we present an alternative proof, which is lengthier, but remains valid even if the underlying Gaussian processes do not admit a continuous modification. In fact, the alternative proof go through if the underlying process admits at most countable number of discontinuities and hence might be of independent interest in the study of iterated random functions.

Note that it is possible to represent any stochastic process $\{Z(\bold{s});\bold{s}\in T\}$, for fixed
$\bold{s}$ as a random variable $\omega\mapsto Z(\bold{s},\omega)$, where $\omega\in\Omega$;
$\Omega$ being the set of all functions from $T$ into $\mathbb R$. 
Also, fixing $\omega\in\Omega$, the function $\bold{s}\mapsto Z(\bold{s},\omega);~\bold{s}\in T$,
represents a path of $Z(\bold{s});\bold{s}\in T$. Indeed, we can identify $\omega$ with the function
$\bold{s}\mapsto Z(\bold{s},\omega)$ from $T$ to $\mathbb R$; see, for example, \cite{Oksendal00}, for
a lucid discussion.  
This latter identification will be convenient for our purpose, and we adopt this for proving our result
on measurability.

Note that the $\sigma$-algebra $\mathcal F$ induced by $Z$
is generated by sets of the form
\[
\left\{\omega:\omega(\bold{s}_1)\in B_1,\omega(\bold{s}_2)\in B_2,\ldots,\omega(\bold{s}_k)\in B_k\right\},
\]
where $B_i\subset\mathbb R;i=1,\ldots,k$, are Borel sets in $\mathbb R$.   

In our case, the Gaussian process $g(\cdot)$ can be identified with $g(x)(\omega_1)=\omega_1(x)$, for any fixed
$x\in\mathbb R$ and $\omega_1\in\Omega_1$, where $\Omega_1$ is the set of all functions from $\mathbb R$ to $\mathbb R$. 
The initial Gaussian process $X(\cdot,0)$ can be identified with
$X(\bold{s},0)(\omega_2)=\omega_2(\bold{s})$, where $\bold{s}\in\mathbb R^d$ (although, we develop the GRFDSTM assuming $d=2$ this proof is given under the more general setting of $d\geq 2$) and $\omega_2\in\Omega_2$.
Here $\Omega_2$ is the set of all functions from $\mathbb R^d$ to $\mathbb R$. Let $\mathcal F_1$ and $\mathcal F_2$
be the Borel $\sigma$-fields associated with $\Omega_1$ and $\Omega_2$, respectively.

We first show that the composition of $g(\cdot)$ with 
$X(\cdot,0)$, given by $g(X(\bold{s},0))$ is a measurable random variable for any $\bold{s}$. Since $g$ and $X(\cdot,0)$ are independent, we need to consider the 
product space $\Omega_1\otimes\Omega_2$, and noting that $g(X(\bold{s},0)(\omega_2))(\omega_1)=\omega_1(\omega_2(\bold{s}))$,
where $(\omega_1,\omega_2)\in\Omega_1\otimes\Omega_2$, need to show that sets of the form
\[
A(\bold{s}_1,\ldots,\bold{s}_k)=\left\{(\omega_1,\omega_2):\omega_1(\omega_2(\bold{s}_1))\in B_1,\omega_1(\omega_2(\bold{s}_2))\in B_2,\ldots,
\omega_1(\omega_2(\bold{s}_k))\in B_k\right\},
\]
where $B_i\subset\mathbb R;i=1,\ldots,k$, are Borel sets in $\mathbb R$, are in $\mathcal F_1\otimes\mathcal F_2$,
the product Borel $\sigma$-field associated with $\Omega_1\otimes\Omega_2$.   
For our purpose, we let $B_i$ be of the form $[a_i,b_i]$ for real values $a_i<b_i$.

Now, suppose that $(\omega_1,\omega_2)\in A(\bold{s}_1,\ldots,\bold{s}_k)$.
Then $\omega_1(\omega_2(\bold{s}_i))\in [a_i,b_i]$, which implies that $\omega_2(\bold{s}_i)$ is at most a countable
union of sets of the form $[a^{(i)}_j,b^{(i)}_j];~j\in\mathcal D_i$, where $\mathcal D_i$ is a countable set of indices.
Also, it holds that $\omega_1(x^*)\in [a_i,b_i];~\forall~x^*\in\mathbb Q\cap\left\{\underset{j\in\mathcal D_i}
\cup [a^{(i)}_j,b^{(i)}_j]\right\}$.
Here $\mathbb Q$ is the countable set of rationals in $\mathbb R$. 
If necessary, we can envisage a countable set
$\mathcal D^*$ consisting of points of discontinuities of $\omega_1$. If $\xi$ is a point of discontinuity, then
$\omega_1(\xi)$ may be only the left limit of particular sequence $\{\omega_1(\xi_{1,m});m=1,2,\ldots\}$ or only 
the right limit of a 
particular sequence $\{\omega_1(\xi_{2,m});m=1,2,\ldots\}$, or $\omega_1(\xi)$ may be an isolated point, 
not reachable by sequences
of the above forms. 
It follows that $(\omega_1,\omega_2)$ must lie in
\[
A^*(\bold{s}_1,\ldots,\bold{s}_k)=\cap_{i=1}^k\left\{(\omega_1,\omega_2):\omega_1(x)\in [a_i,b_i]~\forall~x\in
\left(\mathbb Q\cap\left\{\underset{j\in\mathcal D_i}\cup [a^{(i)}_j,b^{(i)}_j]\right\}\right)\cup\mathcal D^*, 
\omega_2(\bold{s}_i)\in\underset{j\in\mathcal D_i}\cup [a^{(i)}_j,b^{(i)}_j]\right\}.
\]

Now, if $(\omega_1,\omega_2)\in A^*(\bold{s}_1,\ldots,\bold{s}_k)$, then, noting that for any point 
$x\in \underset{j\in\mathcal D_i}\cup [a^{(i)}_j,b^{(i)}_j]$ of $\omega_1$,
$\omega_1(x)=\underset{m\rightarrow\infty}\lim \omega_1(\xi_m)$, where
$\{\xi_m;m=1,2,\ldots\}\in \mathbb Q\cap\left\{\underset{j\in\mathcal D_i}\cup [a^{(i)}_j,b^{(i)}_j]\right\}$,
it is easily seen that $(\omega_1,\omega_2)\in A(\bold{s}_1,\ldots,\bold{s}_k)$. Hence, 
$A(\bold{s}_1,\ldots,\bold{s}_k)=A^*(\bold{s}_1,\ldots,\bold{s}_k)$. 

Now observe that $A^*(\bold{s}_1,\ldots,\bold{s}_k)$ is a finite intersection of countable union of measurable sets;
hence, $A^*(\bold{s}_1,\ldots,\bold{s}_k)$ is itself a measurable set.
In other words, we have proved that $g(X(\cdot,0))$ is measurable.
Now, as $\eta(\cdot,t)$ is independent of $g(\cdot)$ and $X(\cdot,0)$, it follows from (\ref{eqn:npr2}) that
$X(\cdot,1)$ is measurable.

To prove measurability of $X(\cdot,2)$, note that
\begin{align}
X(\bold{s},2)&=g(X(\bold{s},1))+\eta(\bold{s},2)\notag\\
&=g(g(X(\bold{s},0))+\eta(\bold{s},1))+\eta(\bold{s},2).
\label{eq:2nd}
\end{align}
The process $\eta(\cdot,1)$ requires introduction an extra sample space $\Omega_3$, so that
we can identify $\eta(\bold{s},1)(\omega_3)$ as $\omega_3(\bold{s})$. With this, we can represent
$g(g(X(\bold{s},0))+\eta(\bold{s},1))$ of (\ref{eq:2nd}) as
$\omega_1(\omega_1(\omega_2(\bold{s}))+\omega_3(\bold{s}))$. 

Now, $\omega_1(\omega_1(\omega_2(\bold{s}))+\omega_3(\bold{s}))\in [a_i,b_i]$ implies that
$\omega_1(\omega_2(\bold{s}))+\omega_3(\bold{s})\in\underset{j\in\mathcal D_i}\cup[a^{(j)}_i,b^{(j)}_i]$.
If $\omega_1(\omega_2(\bold{s}))+\omega_3(\bold{s})\in [a^{(k)}_i,b^{(k)}_i]$ for some $k\in\mathcal D_i$,
then the set of solutions is 
\begin{equation}
\underset{r\in\mathbb R}\cup\left\{\omega_1(\omega_2(\bold{s}))\in [a^{(k)}_i-r,b^{(k)}_i-r],w_3(\bold{s})=r\right\},
\label{eq:uncountable}
\end{equation}
where $\omega_1(\omega_2(\bold{s}))\in [a^{(k)}_i-r,b^{(k)}_i-r]$ implies, as before, that $\omega_2(\bold{s})$
belongs to a countable union of measurable sets in $\mathbb R$. Although the set (\ref{eq:uncountable}) 
is an uncountable union, following the technique used for proving measurability of $g(X(\cdot,0))$, we will intersect the set by $\mathbb Q$, 
the (countable) set of rationals in $\mathbb R$; this will render the intersection a countable set. The proof of measurability
then follows similarly as before. 

Proceeding likewise, we can prove measurability of $X(\cdot,t)$ is measurable for $t=2,3,\ldots$. Proceeding exactly in the same
way, we can also prove that $Y(\cdot,t);~t=1,2,\ldots,T$ are measurable. 
Moreover, it can be easily seen that the same methods employed for proving the above results on measurability can be
extended in a straightforward (albeit notationally cumbersome) manner to prove that the sets of the forms
\[
\left\{X(\bold{s}_i,t_i)\in [a_i,b_i];i=1,\ldots,k\right\}\ \ \mbox{and} \ \
\left\{Y(\bold{s}_i,t_i)\in [a_i,b_i];i=1,\ldots,k\right\}
\]
are also measurable. Furthermore, it can be easily verified that $X(\bold{s},t)$ and $Y(\bold{s},t)$ satisfy Kolmogorov's consistency
criteria. In other words, $X(\bold{s},t)$ and $Y(\bold{s},t)$ are well-defined stochastic processes in both space and time. 
\end{proof}

\begin{proof}[Proof of Theorem \ref{thm:state}]
Let us first observe that conditional on $g(x)$ our latent process satisfies the Markov property. 
That is, 
\begin{align}
&[(x(\bold{s}_{1},t),\cdots,x(\bold{s}_{n},t))\mid (g(x(\bold{s}_{1},t-1)),\cdots,g(x(\bold{s}_{n},t-1))),(x(\bold{s}_{1},t-1),\notag\\
&\quad\quad\cdots,x(\bold{s}_{n},t-1)),(x(\bold{s}_{1},t-2),\cdots,x(\bold{s}_{n},t-2)),\cdots,(x(\bold{s}_{1},0),\cdots,x(\bold{s}_{n},0))]\notag\\ 
&= [(x(\bold{s}_{1},t),\cdots,x(\bold{s}_{n},t))\mid (g(x(\bold{s}_{1},t-1)),\cdots,g(x(\bold{s}_{n},t-1))),(x(\bold{s}_{1},t-1),
\cdots,x(\bold{s}_{n},t-1))]\notag\\
&\sim \mathlarger{\frac{1}{|\mathbf{\Sigma_{\eta}}|^\frac{1}{2}}}
\exp\left[-\frac{1}{2}{\begin{pmatrix}x(\bold{s}_{1},t)-g(x(\bold{s}_{1},t-1))\\
x(\bold{s}_{2},t)-g(x(\bold{s}_{2},t-1))\\ 
\vdots\\
x(\bold{s}_{n},t)-g(x(\bold{s}_{n},t-1))\end{pmatrix}}^{\prime}{{\mathbf{\Sigma}}_{\eta}}^{-1}
{\begin{pmatrix}x(\bold{s}_{1},t)-g(x(\bold{s}_{1},t-1))\\
 x(\bold{s}_{2},t)-g(x(\bold{s}_{2},t-1))\\ 
\vdots\\
x(\bold{s}_{n},t)-g(x(\bold{s}_{n},t-1))\end{pmatrix}}\right],\notag
\end{align}

\vspace{3mm}
where $[x\mid y]$ denotes the conditional density of $X$ at $x$ given $Y=y$.
Now, let us represent $g(x(\bold{s}_{i},t-1))$ by $u(i,t)$ for all $i=1,\cdots,n$ and $t=1,2,\cdots,T$. 
Then repeatedly using the Markov property we have following
\begin{align*}
[x(\bold{s}_{1},T),\cdots,x(\bold{s}_{n},T),\cdots,x(\bold{s}_{1},0),\cdots,x(\bold{s}_{n},0) & \mid g(x(\bold{s}_{1},T-1)),\cdots\\ 
\cdots,g(x(\bold{s}_{n},T-1)),\cdots,g(x(\bold{s}_{1},0))& ,\cdots,g(x(\bold{s}_{n},0))]\\
\end{align*}
\begin{align*}
\sim
[x(\bold{s}_{1},T)& ,\cdots,x(\bold{s}_{n},T)\mid g(x(\bold{s}_{1},T-1)),\cdots,g(x(\bold{s}_{n},T-1)),
x(\bold{s}_{1},T-1),\cdots,x(\bold{s}_{n},T-1)]\times\\
\dotsm \times &[x(\bold{s}_{1},1),\cdots,x(\bold{s}_{n},1)\mid g(x(\bold{s}_{1},0)),\cdots,g(x(\bold{s}_{n},0)),
x(\bold{s}_{1},0),\cdots,x(\bold{s}_{n},0)]\\
& \times [x(\bold{s}_{1},0),\cdots,x(\bold{s}_{n},0)]\\
\vspace{3mm}
\sim \mathlarger{\frac{1}{(2\pi)^\frac{nT}{2}}}&\mathlarger{\frac{1}{|\mathbf{\Sigma_{\eta}}|^\frac{T}{2}}}\prod_{t=1}^{T}
\exp\left[-\frac{1}{2}{\begin{pmatrix}x(\bold{s}_{1},t)-u(1,t)\\x(\bold{s}_{2},t)-u(2,t)\\ 
\vdots\\
x(\bold{s}_{n},t)-u(n,t)\end{pmatrix}}^{\prime}{{\mathbf{\Sigma}}_{\eta}}^{-1}
{\begin{pmatrix}x(\bold{s}_{1},t)-u(1,t)\\x(\bold{s}_{2},t)-u(2,t)\\ \vdots\\x(\bold{s}_{n},t)-u(n,t)\end{pmatrix}}\right]\\
\vspace{4mm}
\times\mathlarger{\frac{1}{(2\pi)^\frac{n}{2}}}&\mathlarger{\frac{1}{|\mathbf{\Sigma}_{0}|^\frac{1}{2}}}
\exp\left[-\frac{1}{2}{\begin{pmatrix}x(\bold{s}_{1},0)-\mu_{01}\\
x(\bold{s}_{2},0)-\mu_{02}\\ 
\vdots\\
x(\bold{s}_{n},0)-\mu_{0n}\end{pmatrix}}^{\prime}{{\mathbf{\Sigma}}_{0}}^{-1}
{\begin{pmatrix}x(\bold{s}_{1},0)-\mu_{01}\\x(\bold{s}_{2},0)-\mu_{02}\\ \vdots\\x(\bold{s}_{n},0)-\mu_{0n}\end{pmatrix}}\right]
\end{align*}
But, this is the joint density of the state variables conditioned on $g(x)$. To obtain the joint density of the state variables 
one needs to marginalize it with respect to the Gaussian process $g(\cdot)$. 
After marginalization, the joint density takes the following form:
\begin{align*}
\mathlarger{\frac{1}{(2\pi)^\frac{n(T+1)}{2}}}\mathlarger{\frac{1}{|\mathbf{\Sigma}_{0}|^\frac{1}{2}}}
\mathlarger{\frac{1}{|\mathbf{\Sigma_{\eta}}|^\frac{T}{2}}}
\exp\left[-\frac{1}{2}{\begin{pmatrix}x(\bold{s}_{1},0)-\mu_{01}\\x(\bold{s}_{2},0)-\mu_{02}\\ 
\vdots\\
x(\bold{s}_{n},0)-\mu_{0n}\end{pmatrix}}^{\prime}{{\mathbf{\Sigma}}_{0}}^{-1}
{\begin{pmatrix}x(\bold{s}_{1},0)-\mu_{01}\\x(\bold{s}_{2},0)-\mu_{02}\\ \vdots\\x(\bold{s}_{n},0)-\mu_{0n}\end{pmatrix}}\right]\\
\end{align*}

\begin{align*}
\times \mathlarger{\mathlarger{\int_{\mathbb{R}^{nT}}}}\prod_{t=1}^{T}
\exp\left[-\frac{1}{2}{\begin{pmatrix}x(\bold{s}_{1},t)-u(1,t)\\
x(\bold{s}_{2},t)-u(2,t)\\ 
\vdots\\
x(\bold{s}_{n},t)-u(n,t)\end{pmatrix}}^{\prime}
{{\mathbf{\Sigma}}_{\eta}}^{-1}
{\begin{pmatrix}x(\bold{s}_{1},t)-u(1,t)\\x(\bold{s}_{2},t)-u(2,t)\\ 
\vdots\\x(\bold{s}_{n},t)-u(n,t)\end{pmatrix}}\right]
\mathlarger{\frac{1}{(2\pi)^\frac{nT}{2}}}\mathlarger{\frac{1}{|\mathbf{\Sigma}|^\frac{1}{2}}}\\
\exp\left[-\frac{1}{2}{\begin{pmatrix}u(1,1)-\beta_{0g}-\beta_{1g}x(\bold{s}_{1},0)\\u(2,1)-\beta_{0g}-\beta_{1g}x(\bold{s}_{2},0)\\ 
\vdots\\u(n,T)-\beta_{0g}-\beta_{1g}x(\bold{s}_{n},T-1)\end{pmatrix}}^{\prime}{{\mathbf{\Sigma}}}^{-1}
{\begin{pmatrix}u(1,1)-\beta_{0g}-\beta_{1g}x(\bold{s}_{1},0)\\u(2,1)-\beta_{0g}-\beta_{1g}x(\bold{s}_{2},0)\\ 
\vdots\\u(n,T)-\beta_{0g}-\beta_{1g}x(\bold{s}_{n},T-1)\end{pmatrix}}\right]d\mathbf{u}\\
\end{align*}
where $\mathbf{\Sigma}$ is as in (3.2). 
This is nothing but a convolution of two $\mathbb{R}^{nT}$ dimensional Gaussian densities, one with mean vector $\mathbf{0}$ 
and covariance matrix ${\mathbf{I}}_{T\times T}\bigotimes\mathbf{\Sigma_{\eta}}$ and the other one with mean vector \\ 
$(\beta_{0g}+\beta_{1g}x(\bold{s}_{1},0),\cdots,\beta_{0g}+\beta_{1g}x(\bold{s}_{n},T-1))^\prime $ and covariance matrix $\mathbf{\Sigma}$.\\
Hence, the integral boils down to
\begin{align*}
\mathlarger{\frac{1}{(2\pi)^\frac{n}{2}}}\mathlarger{\frac{1}{|\mathbf{\Sigma}_{0}|^\frac{1}{2}}}
\exp\left[-\frac{1}{2}{\begin{pmatrix}x(\bold{s}_{1},0)-\mu_{01}\\x(\bold{s}_{2},0)-\mu_{02}\\ 
\vdots\\x(\bold{s}_{n},0)-\mu_{0n}\end{pmatrix}}^{\prime}{{\mathbf{\Sigma}}_{0}}^{-1}
{\begin{pmatrix}x(\bold{s}_{1},0)-\mu_{01}\\x(\bold{s}_{2},0)-\mu_{02}\\ \vdots\\x(\bold{s}_{n},0)-\mu_{0n}\end{pmatrix}}\right]
\mathlarger{\mathlarger{\frac{1}{(2\pi)^\frac{nT}{2}}\frac{1}{|\tilde{\mathbf{\Sigma}}|^\frac{1}{2}}}}\\ 
\times\exp\left[-\frac{1}{2}{\begin{pmatrix}x(\bold{s}_{1},1)-\beta_{0g}-\beta_{1g}
x(\bold{s}_{1},0)\\x(\bold{s}_{2},1)-\beta_{0g}-\beta_{1g}x(\bold{s}_{2},0)\\ 
\vdots\\x(\bold{s}_{n},T)-\beta_{0g}-\beta_{1g}
x(\bold{s}_{n},T-1)\end{pmatrix}}^{\prime}{\tilde{\mathbf{\Sigma}}}^{-1}{\begin{pmatrix}x(\bold{s}_{1},1)-\beta_{0g}-\beta_{1g}
x(\bold{s}_{1},0)\\x(\bold{s}_{2},1)-\beta_{0g}-\beta_{1g}
x(\bold{s}_{2},0)\\ \vdots\\x(\bold{s}_{n},T)-\beta_{0g}-\beta_{1g}
x(\bold{s}_{n},T-1)\end{pmatrix}}\right],
\end{align*}
where $\tilde{\mathbf{\Sigma}}$ is as in (3.2). 
\end{proof}

\begin{proof}[Proof of Theorem \ref{thm:observe}]
First, see that for fixed $x(\bold{s}_{i},t_{1})$, $Y(\bold{s}_{i},t_{1})$ is distributed as a Gaussian with mean 
$\beta_{0f}+\beta_{1f}x(\bold{s}_{i},t_{1})$ and variance ${\sigma_{f}}^2+{\sigma_{\epsilon}}^2$ where ${\sigma_{f}}^2$ and 
${\sigma_{\epsilon}}^2$ are respectively the process variance associated with the isotropic Gaussian processes $\epsilon(\cdot,t)$ 
and $f(x)$ (see (3) and (1)). 
Now, see that for fixed $x(\bold{s}_{i},t_{1})$ and $x(\bold{s}_{j},t_{2})$, 
$f(x(\bold{s}_{i},t_{1}))$ and $f(x(\bold{s}_{j},t_{2}))$ has covariance $c_{f}(x(\bold{s}_{i},t_{1}),x(\bold{s}_{j},t_{2}))$. 
Also, $\epsilon(\cdot,t_{1})$ and $\epsilon(\cdot,t_{2})$ are mutually independent spatial Gaussian processes for $t_{1}\neq t_{2}$. 
Hence, conditional on state variables the covariance between $Y(\bold{s}_{i},t_{1})$ and $Y(\bold{s}_{j},t_{2})$ is 
$c_{f}(x(\bold{s}_i,t_1),x(\bold{s}_j,t_2))+c_{\epsilon}(\bold{s}_i,\bold{s}_j)\delta(t_1-t_2)$. Here $\delta(\cdot)$ is the 
delta function i.e. $\delta(t)=1$ for $t=0$ and $=0$ otherwise.

So, the joint density of the observed variables, which is denoted by $[y(\bold{s}_{1},1),y(\bold{s}_{2},1),\cdots,y(\bold{s}_{n},T)]$, 
is given by
\begin{align*}
&[y(\bold{s}_{1},1),y(\bold{s}_{2},1),\cdots,y(\bold{s}_{n},T)]=\\
&\mathlarger{\int_{\mathbb{R}^{nT}}}[y(\bold{s}_{1},1),y(\bold{s}_{2},1),\cdots,y(\bold{s}_{n},T)\mid x(\bold{s}_{1},1),
x(\bold{s}_{2},1),\cdots,x(\bold{s}_{n},T)][x(\bold{s}_{1},1),x(\bold{s}_{2},1),\cdots,x(\bold{s}_{n},T)]d\mathbf{x}
\end{align*}
Hence, part $(a)$ follows.

For part $(b)$ note that if $\sigma_{f}^{2}=0$, the conditional density\\ 
$[y(\bold{s}_{1},1),y(\bold{s}_{2},1),\cdots,y(\bold{s}_{n},T)\mid x(\bold{s}_{1},1),x(\bold{s}_{2},1),\cdots,x(\bold{s}_{n},T)]$ 
is Gaussian with block diagonal covariance matrix ${\mathbf{I}}_{T\times T}\bigotimes\mathbf{\Sigma_{\epsilon}}$. On the other hand, 
we have already noted that if $\sigma_{g}^{2}=0$, the joint density of the state variables boils down to Gaussian 
(see the discussion following Theorem \ref{thm:state}). Let us consider only the state variables from time $t=1$ onwards. 
They jointly follow an $nT$ dimensional Gaussian distribution. It is not difficult to see that the mean vector and the covariance matrix 
of the $nT$ dimensional Gaussian distribution are of following forms:\\
the $((t-1)n+i)$ th entry of the mean vector is $\beta_{1g}^{t}\mu_{0i}+{(\beta_{1g}-1)^2}+\beta_{0g}\frac{\beta_{1g}^{t}-1}{\beta_{1g}-1}$ where $1\leq t\leq T$\\
and the $(((t_{1}-1)n+i),((t_{2}-1)n+j))$ th entry of the covariance matrix is\\
$\beta_{1g}^{t_{1}+t_{2}}\sigma_{i,j}^{0}+(\beta_{1g}^{t_{1}+t_{2}-2}+\beta_{1g}^{t_{1}+t_{2}-4}+\cdots+\beta_{1g}^{\abs{t_{1}-t_{2}}})c_{\eta}(s_{i},s_{j})$\ where $1\leq t_{1},t_{2}\leq T$ and $1\leq i,j\leq n$\\
$\sigma_{i,j}^{0}$ is the $(i,j)$ th entry of the covariance matrix $\boldsymbol{\Sigma}_{0}$.

Now, using part $(a)$ we see that the joint distribution of $Y(\bold{s}_{1},1),Y(\bold{s}_{2},1),\cdots,Y(\bold{s}_{n},T)$ is nothing but 
a convolution of two $\mathbb{R}^{nT}$-dimensional Gaussian densities. Hence, it is a Gaussian distribution whose mean vector 
has $((t-1)n+i)$ th entry as $\beta_{0f}+\beta_{1f}\left(\beta_{1g}^{t}\mu_{0i}
+\beta_{0g}\frac{\beta_{1g}^{t}-1}{\beta_{1g}-1}\right)\\
\text{where $1\leq t\leq T$}$ and the $(((t_{1}-1)n+i),((t_{2}-1)n+j))$ th entry of the covariance matrix is \\ 
$\beta_{1f}^{2}\left(\beta_{1g}^{t_{1}+t_{2}}\sigma_{i,j}^{0}+(\beta_{1g}^{t_{1}+t_{2}-2}+\beta_{1g}^{t_{1}+t_{2}-4}
+\cdots+\beta_{1g}^{\abs{t_{1}-t_{2}}})c_{\eta}(s_{i},s_{j})\right)+c_{\epsilon}(\bold{s}_{i},\bold{s}_{j})\delta(t_{1}-t_{2})$ \\
where $1\leq t_{1},t_{2}\leq T$ and $1\leq i,j\leq n$. So, part $(b)$ is proved.
\end{proof}

\begin{proof}[Proof of Theorem \ref{thm:covariance}] 
To prove part $(a)$, we first show that $E(X^2(\bold{s},t))$ is finite. Then by using the formula
\begin{align*}
E(Y^2(\bold{s},t))&=E\left(E(Y^2(\bold{s},t)|X(\bold{s},t)\right))\\
&=E\left(Var(Y(\bold{s},t)|X(\bold{s},t)\right))+E\left(E(Y(\bold{s},t)|X(\bold{s},t)\right))^2\\
&=E(\sigma_{f}^2+\sigma_{\epsilon}^2)+E(\beta_{0f}+\beta_{1f}X(\bold{s},t))^2\\
&=\sigma_{f}^2+\sigma_{\epsilon}^2+E(\beta_{0f}+\beta_{1f}X(\bold{s},t))^2
\end{align*}
we establish that $E(Y^2(\bold{s},t))$ is finite that in turn implies that $Var(Y(\bold{s},t))$ is finite.
To show $E(X^2(\bold{s},t))$ is finite we use principle of mathematical induction, i.e. we first show that $E(X^2(\bold{s},0))$ is 
finite and then we show that if\\
$E(X^2(\bold{s},0)),E(X^2(\bold{s},1)),\cdots,E(X^2(\bold{s},t-1))$ 
are finite 
then $E(X^2(\bold{s},t))$ is finite. These two steps together compel $E(X^2(\bold{s},t))$ to be finite for all $\bold{s}$ and $t$.

The first step is trivially shown as $X(\bold{s},0)$ is a  Gaussian random variable.
Now we show the second step of mathematical induction, that is, we show that if
\\ 
$E(X^2(\bold{s},0)),E(X^2(\bold{s},1)),\cdots,E(X^2(\bold{s},t-1))$ are finite then $E(X^2(\bold{s},t))$ is finite. 

Let us consider the following:
\begin{align*}
& Var(X(\bold{s},t)|X(\bold{s},t-1)=x_{t-1},X(\bold{s},t-2)=x_{t-2}\cdots,X(\bold{s},0)=x_{0})\\
&=Var(g(X(\bold{s},t-1))+\eta(\bold{s},t)|X(\bold{s},t-1)=x_{t-1},X(\bold{s},t-2)=x_{t-2}\cdots,X(\bold{s},0)=x_{0})\\
&=Var(g(X(\bold{s},t-1))|X(\bold{s},t-1)=x_{t-1},X(\bold{s},t-2)=x_{t-2}\cdots,X(\bold{s},0)=x_{0})+\sigma_{\eta}^2\\
&=Var(g(x_{t-1})|X(\bold{s},t-1)=x_{t-1},X(\bold{s},t-2)=x_{t-2}\cdots,X(\bold{s},0)=x_{0})+\sigma_{\eta}^2\\
&=Var(g(x_{t-1})|g(x_{t-2})+\eta(\bold{s},t-1)=x_{t-1},\cdots,g(x_{0})+\eta(\bold{s},1)=x_{1},X(\bold{s},0)=x_{0})+\sigma_{\eta}^2\\
&=\sigma_{g}^2-\bold{\Sigma}_{g12}^{'}(\bold{\Sigma}_{g22}+\sigma_{\eta}^2\bold{I})^{-1}\bold{\Sigma}_{g12}+\sigma_{\eta}^2\ \ 
\text{(see page 16 of \cite{Rasmussen:Williams})}
\end{align*}
where $\bold{\Sigma}_{g12}^{'}$ is the row vector $(c_{g}(x_{t-1},x_{0})\  c_{g}(x_{t-1}x_{1})\cdots\  c_{g}(x_{t-1},x_{t-2}))$ 
and $\bold{\Sigma}_{g22}$ is the variance covarince matrix $\begin{pmatrix}c_{g}(x_{0},x_{0})\ c_{g}(x_{0},x_{1})\cdots c_{g}(x_{0},x_{t-2})
\\
c_{g}(x_{1},x_{0})\ c_{g}(x_{1},x_{1})\cdots c_{g}(x_{1},x_{t-2})\\ 
\vdots\\c_{g}(x_{t-2},x_{0})\ c_{g}(x_{t-2},x_{1})\cdots c_{g}(x_{t-2},x_{t-2})\end{pmatrix}$ induced by covariance function $c_{g}(\cdot,\cdot)$. 
Now, we consider $E(Var(X(\bold{s},t)|X(\bold{s},t-1)=x_{t-1},X(\bold{s},t-2)=x_{t-2}\cdots,X(\bold{s},0)=x_{0}))$. 
We want to show that this quantity is finite. But the problem is that we have to deal with the inverse of a random matrix 
$(\bold{\Sigma}_{g22}+\sigma_{\eta}^2\bold{I}) $. Fortunately, 
the random matrix $(\bold{\Sigma}_{g22}+\sigma_{\eta}^2\bold{I}) $ is non-negative definite (nnd). Hence, 
$\sigma_{g}^2-\bold{\Sigma}_{g12}^{'}(\bold{\Sigma}_{g22}+\sigma_{\eta}^2\bold{I})^{-1}\bold{\Sigma}_{g12}
+\sigma_{\eta}^2\leq \sigma_{g}^2+\sigma_{\eta}^2$. On the other hand, this quantity being a conditional variance is always nonnegative. 
So, the following inequality holds\\
$$0\leq \sigma_{g}^2-\bold{\Sigma}_{g12}^{'}(\bold{\Sigma}_{g22}+\sigma_{\eta}^2\bold{I})^{-1}\bold{\Sigma}_{g12}
+\sigma_{\eta}^2\leq \sigma_{g}^2+\sigma_{\eta}^2.$$
Hence, it follows that
$$0\leq E(\sigma_{g}^2-\bold{\Sigma}_{g12}^{'}(\bold{\Sigma}_{g22}+\sigma_{\eta}^2\bold{I})^{-1}\bold{\Sigma}_{g12}
+\sigma_{\eta}^2)\leq \sigma_{g}^2+\sigma_{\eta}^2.$$ So, the quantity $E(Var(X(\bold{s},t)|X(\bold{s},t-1)=x_{t-1},X(\bold{s},t-2)
=x_{t-2}\cdots,X(\bold{s},0)=x_{0}))$ being equivalent to $E(\sigma_{g}^2-\bold{\Sigma}_{g12}^{'}(\bold{\Sigma}_{g22}
+\sigma_{\eta}^2\bold{I})^{-1}\bold{\Sigma}_{g12}+\sigma_{\eta}^2)$, is finite.\\

Now we consider the term $E(X(\bold{s},t)|X(\bold{s},t-1)=x_{t-1},X(\bold{s},t-2)=x_{t-2}\cdots,X(\bold{s},0)=x_{0})$.
\begin{align*}
 &E(X(\bold{s},t)|X(\bold{s},t-1)=x_{t-1},X(\bold{s},t-2)=x_{t-2}\cdots,X(\bold{s},0)=x_{0})\\
&=E(g(X(\bold{s},t-1))+\eta(\bold{s},t)|X(\bold{s},t-1)=x_{t-1},X(\bold{s},t-2)=x_{t-2}\cdots,X(\bold{s},0)=x_{0})\\
&=E(g(X(\bold{s},t-1))|X(\bold{s},t-1)=x_{t-1},X(\bold{s},t-2)=x_{t-2}\cdots,X(\bold{s},0)=x_{0})+0\\
&=E(g(x_{t-1})|X(\bold{s},t-1)=x_{t-1},X(\bold{s},t-2)=x_{t-2}\cdots,X(\bold{s},0)=x_{0})\\
&=E(g(x_{t-1})|g(x_{t-2})+\eta(\bold{s},t-1)=x_{t-1},\cdots,g(x_{0})+\eta(\bold{s},1)=x_{1},X(\bold{s},0)=x_{0})\\
&=\beta_{g0}+\beta_{g1}x_{t-1}+\bold{\Sigma}_{g12}^{'}(\bold{\Sigma}_{g22}+\sigma_{\eta}^2\bold{I})^{-1}\bold{Z(\bold{s})}\ \ 
\text{(see page 16 of \cite{Rasmussen:Williams})}
\end{align*}
where $\bold{Z(\bold{s})}^{\prime}$ is the row vector $(x_{1}-\beta_{g0}-\beta_{g1}x_{0}\  \ x_{2}-\beta_{g0}-\beta_{g1}x_{1}\  
\cdots \ x_{t-1}-\beta_{g0}-\beta_{g1}x_{t-2})$. We want to show that 
$E(E(X(\bold{s},t)|X(\bold{s},t-1)=x_{t-1},X(\bold{s},t-2)=x_{t-2}\cdots,X(\bold{s},0)=x_{0}))^2$ is finite. Equivalently, 
we want to show $E(\beta_{g0}+\beta_{g1}x_{t-1}+\bold{\Sigma}_{g12}^{'}(\bold{\Sigma}_{g22}+\sigma_{\eta}^2\bold{I})^{-1}\bold{Z(\bold{s})})^2$ 
is finite. For that it is enough to show $E(\bold{\Sigma}_{g12}^{'}(\bold{\Sigma}_{g22}+\sigma_{\eta}^2\bold{I})^{-1}\bold{Z(\bold{s})})^2$ 
is finite since our induction hypothesis already assume that $E(X(\bold{S},t-1))^2$ is finite.

Now we show that $E(\bold{\Sigma}_{g12}^{'}(\bold{\Sigma}_{g22}+\sigma_{\eta}^2\bold{I})^{-1}\bold{Z(\bold{s})})^2$ is finite. 
First note that $\bold{\Sigma}_{g12}^{'}(\bold{\Sigma}_{g22}+\sigma_{\eta}^2\bold{I})^{-1}\bold{Z(\bold{s})}$ can be expressed as a 
linear combination of the elements of $\bold{Z(\bold{s})}$ as $w_{1}z_{1}(\bold{s})+w_{2}z_{2}(\bold{s})+\cdots+w_{t-1}z_{t-1}(\bold{s})$. 
If the $w_{i}(\bold{S})$ are fixed numbers it is easy to see that $w_{1}z_{1}(\bold{s})+w_{2}z_{2}(\bold{s})+\cdots+w_{t-1}z_{t-1}(\bold{s})$ 
has finite second moment. Unfortunately, $w_{i}(\bold{S})$ are random. However we will show that they are bounded random variables and 
then using a lemma we will prove that $E(w_{1}z_{1}(\bold{s})+w_{2}z_{2}(\bold{s})+\cdots+w_{t-1}z_{t-1}(\bold{s}))^2$ is finite.

First we show that $w_{i}(\bold{S})$ are bounded random variables. Consider the spectral decomposition of the real symmetric (nnd) 
matrix $\bold{\Sigma}_{g22}$. Let us assume that $\bold{\Sigma}_{g22}=\bold{U}\bold{D}\bold{U}^{\prime}$, where $\bold{U}$ is an 
orthogonal matrix and $\bold{D}$ is the diagonal matrix whose diagonal elements are eigenvalues. Then 
\begin{align*}
 \bold{\Sigma}_{g12}^{'}(\bold{\Sigma}_{g22}+\sigma_{\eta}^2\bold{I})^{-1}\bold{Z(\bold{s})}
 &= \bold{\Sigma}_{g12}^{'}(\bold{U}\bold{D}\bold{U}^{\prime}+\sigma_{\eta}^2\bold{I})^{-1}\bold{Z(\bold{s})}\\
&=\bold{\Sigma}_{g12}^{'}(\bold{U}\bold{D}\bold{U}^{\prime}+\sigma_{\eta}^2\bold{U}\bold{U}^{\prime})^{-1}\bold{Z(\bold{s})}\\
&=\bold{\Sigma}_{g12}^{'}(\bold{U}(\bold{D}+\sigma_{\eta}^2\bold{I})\bold{U}^{\prime})^{-1}\bold{Z(\bold{s})}\\
&=\bold{\Sigma}_{g12}^{'}{\bold{U}^{\prime}}^{-1}(\bold{D}+\sigma_{\eta}^2\bold{I})^{-1}\bold{U}^{-1}\bold{Z(\bold{s})}\\
&=\bold{\Sigma}_{g12}^{'}\bold{U}(\bold{D}+\sigma_{\eta}^2\bold{I})^{-1}{\bold{U}^{\prime}}\bold{Z(\bold{s})} \ \ 
\text{(Since $\bold{U}$ is an orthogonal matrix)}.
\end{align*}

Since $\bold{U}$ is a (random) orthogonal matrix its elements are bounded random variables between $-1$ and $1$.
The (random) elements of the row vector $\bold{\Sigma}_{g12}^{'}$ are covariances induced by the isotropic covariance kernel 
$c_{g}(\cdot,\cdot)$. Hence, they are bounded random variables between $-\sigma_{g}^{2}$ and $\sigma_{g}^{2}$. Finally, the 
(random) elements of $(\bold{D}+\sigma_{\eta}^2\bold{I})^{-1}$ are bounded random variables between $0$ and $\frac{1}{\sigma_{\eta}^{2}}$. 
Hence, the (random) row vector $\bold{\Sigma}_{g12}^{'}(\bold{\Sigma}_{g22}+\sigma_{\eta}^2\bold{I})^{-1}$, being a product of some 
random matrices whose elements are bounded random variables, is itself composed of bounded random variables. So, 
its elements $w_{i}(\bold{S})$, although random, are bounded. Now, we state a crucial lemma.

\begin{lemma}
Let us assume that $X_{1},X_{2},\cdots,X_{n}$ are random variables with finite second moment and $W_{1},W_{2},\cdots,W_{n}$ are bounded 
random variables all defined on same probability space. Then the random variables $Y=(W_{1}X_{1}+W_{2}X_{2}+\cdots+W_{n}X_{n})$ also 
has finite second moment.

\begin{proof}
Let us assume that $W_{1},W_{2},\cdots W_{n}$ lie between $[-M,M]$ and $E(X_{i}^2)\leq K$ for $i=1,2\cdots,n$. 
Now $E(W_{i}X_{i})^2= E(E(W_{i}^2X_{i}^2|X_{i}))=E(X_{i}^2E(W_{i}^2|X_{i}))$. 
But $E(W_{i}^2|X_{i})\leq M^2$. So, $E(X_{i}^2E(W_{i}^2|X_{i}))\leq M^2E(X_{i}^2) \leq M^2K$.  

So, 
\begin{align*}
|E(Y^2)|&=|\sum_{i=1}^{n}E(W_{i}^2X_{i}^2)+2\sum_{1\leq i<j\leq n}E(W_{i}X_{i}W_{j}X_{j})|\\
&\leq \sum_{i=1}^{n}E(W_{i}^2X_{i}^2)+2\sum_{1\leq i<j\leq n}|E(W_{i}X_{i}W_{j}X_{j})|\\
&\leq \sum_{i=1}^{n}E(W_{i}^2X_{i}^2)+2\sum_{1\leq i<j\leq n}{E(W_{i}^2X_{i}^2)}^{\frac{1}{2}}{E(W_{j}^2X_{j}^2)}^{\frac{1}{2}}\\
&\leq nM^2K+2n(n-1)M^2K
\end{align*}

So, $Y$ has finite second moment.

\end{proof}

\end{lemma}

Once we apply the lemma to $w_{1}z_{1}(\bold{s})+w_{2}z_{2}(\bold{s})+\cdots+w_{t-1}z_{t-1}(\bold{s})$ the 
finiteness of $E(E(X(\bold{s},t)|X(\bold{s},t-1)=x_{t-1},X(\bold{s},t-2)=x_{t-2}\cdots,X(\bold{s},0)=x_{0}))^2$ is immediate. 
Then by the formula $E(X(\bold{s},t))^2=E(E(X^2(\bold{s},t)|X(\bold{s},t-1)=x_{t-1},X(\bold{s},t-2)=x_{t-2}\cdots,X(\bold{s},0)=x_{0}))=E(Var(X(\bold{s},t)|X(\bold{s},t-1)=x_{t-1},X(\bold{s},t-2)=x_{t-2}\cdots,X(\bold{s},0)
=x_{0}))+E(E(X(\bold{s},t)|X(\bold{s},t-1)=x_{t-1},X(\bold{s},t-2)=x_{t-2}\cdots,X(\bold{s},0)=x_{0}))^2$ we get 
that $E(X^2(\bold{s},t))$ is finite.
\end{proof}

We now prove part $(b)$. 
Since we have already proved in part $(a)$ that the coordinate variables of the observed spatio-temporal process have finite variances, 
now we can consider the covariance function associated with the process and study its properties.
Let us denote the covariance between $Y(\bold{s},t)$ and $Y(\bold{s}^*,t^*)$ by $c_{y}((\bold{s},t),(\bold{s}^*,t^*))$. 
Then
\begin{align}
c_{y}((\bold{s},t),(\bold{s}^*,t^*))=&E[Cov(Y(\bold{s},t),Y(\bold{s}^*,t^*)\mid 
x(\bold{s},t),x(\bold{s}^*,t^*))]\notag\\ &+Cov[E(Y(\bold{s},t)\mid x(\bold{s},t)),
E(Y(\bold{s}^*,t^*)\mid x(\bold{s}^*,t^*))]\notag\\
=E[c_{f}(X(\bold{s},t),X(\bold{s}^*,t^*))]+&c_{\epsilon}(\bold{s},\bold{s}^*)\delta(t-t^*)
+\beta_{1f}^{2}Cov[X(\bold{s},t),X(\bold{s}^*,t^*)].\notag\\
\end{align}
Now, the term $E[c_{f}(X(\bold{s},t),X(\bold{s}^*,t^*))]$ will 
be nonstationary and hence $E[c_{f}(X(\bold{s}+\bold{h},t+k),X(\bold{s}^*+\bold{h},t^*+k))]
\neq E[c_{f}(X(\bold{s},t),X(\bold{s}^*,t^*))]$. In fact, $| X(\bold{s}+\bold{h},
t+k)-X(\bold{s}^*+\bold{h},t^*+k)|\neq| X(\bold{s},t)-X(\bold{s}^*,t^*)|$ with probability 1 
because $X(\bold{s},t)$ has density with respect to Lebesgue measure and this heuristically justifies our argument. So, 
the covariance function $c_{y}(\cdot,\cdot)$ is nonstationary in both space and time.

To prove non separability, first see that $c_{f}(x(\bold{s},t),x(\bold{s}^*,t^*))$ is non separable in space and time, 
because both space and time are involved in it through $x(\bold{s},t)$. Hence, $E[c_{f}(X(\bold{s},t),X(\bold{s}^*,t^*))]$ 
is nonseparable and therefore $c_{y}(\cdot,\cdot)$ is nonseparable in space and time.


\begin{proof}[Proof of Theorem \ref{thm:closed_covariance}]
\label{sec:proof5}
First consider $Cov(X(\bold{s},t),X(\bold{s}^{*},t^{*}))$ where WLOG we assume $t>t^*$. 
Also, assume that $g^{*}(\cdot)$ is the centered Gaussian process obtained from $g(\cdot)$. Then
\begin{align*}
&Cov(X(\bold{s},t),X(\bold{s}^*,t^*))=Cov(g(X(\bold{s},t-1))+\eta(\bold{s},t),X(\bold{s}^*,t^*))\\
&=Cov(\beta_{0g}+\beta_{1g}
X(\bold{s},t-1)+g^{*}(X(\bold{s},t-1))+\eta(\bold{s},t),X(\bold{s}^*,t^*))\\
&=\beta_{1g}Cov(X(\bold{s},t-1),X(\bold{s}^*,t^*))+Cov(g^{*}(X(\bold{s},t-1)),X(\bold{s}^*,t^*))\\
\end{align*}
Repeatedly expanding the term in the same way we get
\begin{align}
&=\beta_{1g}^{t-t^*}Cov(X(\bold{s},t^*),X(\bold{s}^*,t^*))+\beta_{1g}^{t-t^*-1}
Cov(g^{*}(X(\bold{s},t^*)),X(\bold{s}^*,t^*))+\label{cov}\\
&\cdots+Cov(g^{*}(X(\bold{s},t-1)),X(\bold{s}^*,t^*))\notag
\end{align}
Just as the previous paragraph we can further see that
\begin{align}
&Cov(X(\bold{s},t^*),X(\bold{s}^*,t^*))\notag\\
&=Cov(\beta_{0g}+\beta_{1g}
X(\bold{s},t^*-1)+g^{*}(X(\bold{s},t^*-1))+\eta(\bold{s},t^*),\beta_{0g}\notag\\
&+\beta_{1g}X(\bold{s}^*,t^*-1)+g^{*}(X(\bold{s}^*,t^*-1))+\eta(\bold{s}^*,t^*))\notag\\
&=\beta_{1g}^2Cov(X(\bold{s},t^*-1),X(\bold{s}^*,t^*-1))+\beta_{1g}Cov(X(\bold{s},t^*-1),g^{*}(X(\bold{s}^*,t^*-1)))+\notag\\
&\beta_{1g}Cov(X(\bold{s}^*,t^*-1),g^{*}(X(\bold{s},t^*-1)))+Cov(g^{*}(X(\bold{s},t^*-1)),g^{*}(X(\bold{s}^*,t^*-1)))+c_{\eta}(\bold{s},\bold{s}^*)
\end{align}
Now we plan to show that terms of the types $Cov(g^{*}(X(\bold{s}^*,t^*-1)),X(\bold{s},t^*-1))$ and 
$Cov(g^{*}(X(\bold{s},t^*-1)),g^{*}(X(\bold{s}^*,t^*-1)))$ are negligible if $\sigma_{g}^2$ is small enough.
Our next lemma proves it rigorously.
\begin{lemma}
\label{lemma:small_cov}
For arbitrarily small  $\epsilon>0$, $\exists\ \delta>0$ such that $Cov(g^{*}(X(\bold{s},t-1)),X(\bold{s}^*,t^*))<\epsilon$ 
for $0<\sigma_{g}^2<\delta$.
\end{lemma}
See that it is enough to prove that $Var(g^{*}(X(\bold{s},t-1)))$ is arbitrarily small $\forall \bold{s},t$. Then Cauchy-Schwartz inequality implies $Cov^{2}(g^{*}(X(\bold{s},t-1)),g^{*}(X(\bold{s}^*,t^*)))
\leq Var(g^{*}(X(\bold{s},t-1)))Var(g^{*}(X(\bold{s}^*,t^*)))$ is arbitrarily small. Similarly, Cauchy-Schwartz inequality implies $Cov^{2}(g^{*}(X(\bold{s},t-1)),\eta(\bold{s}^*,t^*))
\leq Var(g^{*}(X(\bold{s},t-1)))Var(\eta(\bold{s}^*,t^*))=Var(g^{*}(X(\bold{s},t-1)))\sigma_{\eta}^2$ is arbitrarily small. Then we are done by the expansion 
\begin{align*}
&Cov(g^{*}(x(\bold{s},t)),x(\bold{s^*},t^*))=\\
&Cov(g^{*}(x(\bold{s},t-1)),g^{*}(x(\bold{s}^*,t^*-1))+\cdots+\beta_{1g}^{t^*}Cov(g^{*}(x(\bold{s},t-1)),g^{*}(x(\bold{s}^*,0)))\\
&+Cov(g^{*}(x(\bold{s},t-1)),\eta(\bold{s}^*,t^*-1))+\cdots+\beta_{1g}^{t^*}Cov(g^{*}(x(\bold{s},t-1)),\eta(\bold{s}^*,0))
\end{align*}

Before proceeding towards the proof we mention two results from Gaussian process (see \cite{Adler07} for details) that will be used subsequently.
\begin{result}[Borell-TIS inequality] 
\label{result:Borell-TIS}
Let us assume that $g$ is an almost surely bounded centered Gaussian process on index set 
$T\subseteq\mathbb{R}$. Define $\sigma_{T}^{2}=\sup\limits_{t\in T}E(g_{t}^2)$.\\ 
Then $P(\|g\|>s)\leq \exp(-\frac{(s-E\|g\|)^2}{2\sigma_{T}^2})$ for $s>E(\|g\|)$ where $\|g\|=\sup\limits_{t}g_{t}$.
\end{result}

\begin{result}[Dudley's metric entropy bound] 
\label{result:Dudley}
Under the assumption of the Borel-TIS inequality,\\

$E\|g\|\leq K\mathlarger{\int_{0}^{\mbox{diam}(T)}}\sqrt{H(\epsilon)}d\epsilon$,\\

where $\mbox{diam}(T)=\underset{\bold{s}_1,\bold{s}_2\in T}\sup d(\bold{s}_1,\bold{s}_2)$ is the diameter of the index set 
$T$ with respect to the canonical pseudo-metric $d$ associated with the Gaussian process $g$
given by $d(\bold{s}_1,\bold{s}_2)=\sqrt{E(g(\bold{s}_1)-g(\bold{s}_2))^2}$,
and 
$H(\epsilon)=\ln{N(\epsilon)}$ where $N(\epsilon)$ is the minimum number of $\epsilon$ balls required to cover the index set 
$T$ with respect to the canonical pseudo-metric $d$; $K$ is a universal constant.
\end{result}
With the above two results, we are ready to prove Lemma \ref{lemma:small_cov}. 

\begin{proof}[Proof of Lemma \ref{lemma:small_cov}]
Consider $Var((g^{*}(X(\bold{s},t-1))))$. Observe that
\begin{align*}
Var((g^{*}(X(\bold{s},t-1))))&\leq E((g^{*}(X(\bold{s},t-1))))^{2}\\ 
& \leq E(\sup\limits_{x}|g^{*}(x)|^{2})\ \ \ \\ 
& =\mathlarger{\int_{0}^{\infty}}P(\sup\limits_{x}|g^{*}(x)|^{2}>u)du\ \ \  (\text{by the tail sum formula})\\
&\leq 2\mathlarger{\int_{0}^{\infty}}P(\sup\limits_{x} g^{*}(x)>\sqrt{u})du\\
&= 2\mathlarger{\int_{0}^{L^{2}}}P(\sup\limits_{x} g^{*}(x)>\sqrt{u})du+2\mathlarger{\int_{L^{2}}^{\infty}}
P(\sup\limits_{x} g^{*}(x)>\sqrt{u})du\\
&(\text{where $L=\max{(E(\sup\limits_{x}g^{*}(x)),0)}$}\\
&\leq 2L^{2}+2\mathlarger{\int_{L^{2}}^{\infty}}e^{-\frac{(\sqrt{u}-L)^2}{2\sigma_{g}^{2}}}\ \ \ (\text{using Result \ref{result:Borell-TIS}})
\end{align*}

Now, using the change of variable $\sqrt{u}=z+L$ the integral $\mathlarger{\int_{L^{2}}^{\infty}}e^{-\frac{(\sqrt{u}-L)^2}{2\sigma_{g}^{2}}}$ 
can be reduced to the form
\begin{align*}
&\mathlarger{\int_{0}^{\infty}}e^{-\frac{z^2}{2\sigma_{g}^{2}}}2zdz+2L\mathlarger{\int_{0}^{\infty}}e^{-\frac{z^2}{2\sigma_{g}^{2}}}dz\\
&=2\sigma_{g}^{2}+L\sigma_{g}(\sqrt{2\pi})
\end{align*}
Hence, $Var((g^{*}(X(\bold{s},t-1))))\leq 2L^{2}+4\sigma_{g}^{2}+2L\sigma_{g}(\sqrt{2\pi})$.
\\But, $0\leq L\leq K\mathlarger{\int_{0}^{diam(T)}}\sqrt{H(\epsilon)}d\epsilon$ by Result \ref{result:Dudley} 
and it is not difficult to see that $H(\epsilon)$ is a decreasing function of $\sigma_{g}^2$. The same is true of $diam(T)$ 
when as a function of $\sigma_{g}^2$. These two facts together permit applicability of the monotone convergence theorem to yield
\begin{align*}
0\leq \lim_{\sigma_{g}^{2}\rightarrow 0^+}L\leq \lim_{\sigma_{g}^{2}\rightarrow 0^+}K\mathlarger{\int_{0}^{diam(T)}}\sqrt{H(\epsilon)}d\epsilon 
& \leq \lim_{\sigma_{g}^{2}\rightarrow 0^+}K\mathlarger{\int_{0}^{\infty}}\sqrt{H(\epsilon)}\mathbb{I}(\epsilon\leq diam(T))d\epsilon\\
&\leq K\mathlarger{\int_{0}^{\infty}}\lim_{\sigma_{g}^{2}\rightarrow 0^+}\sqrt{H(\epsilon)}\mathbb{I}(\epsilon\leq diam(T))d\epsilon=0.
\end{align*}
So, $\lim_{\sigma_{g}^{2}\rightarrow 0^+}L=0$ which in turn implies $Var((g^{*}(X(\bold{s},t-1))))$ can be made 
arbitrarily small by making $\sigma_{g}^{2}$ small. This proves Lemma \ref{lemma:small_cov}.
\end{proof}

Arguing similarly one can also show that for arbitrarily small $\epsilon>0$, $\exists\ \delta>0$ such that 
$Cov(g^{*}(X(\bold{s},t^*-1)),g^{*}(X(\bold{s}^*,t^*-1)))<\epsilon$ for $0<\sigma_{g}^2<\delta$. Moreover, 
see that the bound is uniform in $\bold{s}$ and $t$. Since $|\beta_{1g}|<1$, using the bound repeatedly in (15), we obtain
\begin{align*}
&|Cov(X(\bold{s},t),X(\bold{s}^*,t^*))-\beta_{1g}^{t-t^*}Cov(X(\bold{s},t^*),X(\bold{s}^*,t^*))|\\
&\leq \frac{\epsilon}{1-|\beta_{1g}|}
\end{align*}
Similarly, using the bound repeatedly in (16), we obtain
\begin{align*}
&|Cov(X(\bold{s},t^*),X(\bold{s}^*,t^*))-Cov(X(\bold{s},0),X(\bold{s}^*,0))-[\frac{1-\beta_{1g}^{2(t^*+1)}}{1-\beta_{1g}^2}]c_{\eta}(\bold{s},\bold{s}^*)|\\
&\leq \left[\frac{\epsilon}{1-|\beta_{1g}|}+\frac{\epsilon}{1-|\beta_{1g}|}+\frac{\epsilon}{1-|\beta_{1g}|}\right].
\end{align*}
Combining them we get
\begin{align*}
&|Cov(X(\bold{s},t),X(\bold{s}^*,t^*))-\beta_{1g}^{t-t^*}Cov(X(\bold{s},0),X(\bold{s}^*,0))-\beta_{1g}^{t-t^*}[\frac{1-\beta_{1g}^{2(t^*+1)}}{1-\beta_{1g}^2}]c_{\eta}(\bold{s},\bold{s}^*)|\\
&\leq \frac{\epsilon}{1-|\beta_{1g}|}\left[1+3|\beta_{1g}|^{t-t^*}\right] \leq \frac{4\epsilon}{1-|\beta_{1g}|}
\end{align*}
Now plugging in this approximation 
in the expression for $c_{y}((\bold{s},t),(\bold{s}^*,t^*))$ we get
the desired result.
So, Theorem 3.5 
is finally proved.

\end{proof}

\begin{proof}[Proof of Theorem \ref{thm:sample_path}]
Part $(a)$: Arguing in the similar lines as in the proof of Proposition \ref{propn:measurable}, one can show that $\exists$ a probability space $(\Omega,\mathcal{F},P)$ such that for $\omega\in \Omega$, $X(\bold{s},0)(\omega),\eta(\bold{s},t)(\omega),
\epsilon(\bold{s},t)(\omega)$ are continuous functions in $\bold{s}$ where $t=1,2,3\cdots$ and $g(x)(\omega),f(x)(\omega)$ 
are continuous functions in $x$. Then by the property of composition of two functions $X(\bold{s},1)(\omega)
=g(X(\bold{s},0)(\omega))(\omega)+\eta(\bold{s},1)(\omega)$ is a continuous function in $\bold{s}$. Proceeding recursively, 
one can prove that $X(\bold{s},t)(\omega)$ is a continuous function in $\bold{s}$ for any $t$. Once we show $X(\bold{s},t)(\omega)$ 
is a continuous function, we prove $Y(\bold{s},t)(\omega)=f(X(\bold{s},t)(\omega))(\omega)+\epsilon(\bold{s},t)(\omega)$ 
is a continuous function in $\bold{s}$. So, part $(a)$ is proved.

Part $(b)$: Proof of part $(b)$ is similar to that of part $(a)$. Firstly, we state a simple lemma.
\begin{lemma}
Let us consider two real valued functions $u(z)$ and $v(x,y)$ such that both of them are $k$ times differentiable. 
Then the composition function $u(v(x,y))$ is also $k$ times differentiable.
\end{lemma}

\begin{proof}
Proof of this lemma is basically a generalization of chain rule for multivariate functions and can be found in advanced 
multivariate calculus text books. We give a brief sketch of the proof. First we clarify the term $k$ times differentiable 
for the function $v(x,y)$. It means all mixed partial derivatives of $v(x,y)$ of order $k$ exist. We prove the lemma using 
mathematical induction. Firstly, We show that the lemma is true for $k=1$ and then we show that if the lemma is true for $k-1$ 
then it must be true for $k$ as well. 

That the lemma is true for $k=1$ easily follows from the chain rule for multivariate functions. Now we prove the second step.
By the induction hypothesis the lemma is true for the $k-1$ case and $u(z)$ and $v(x,y)$ are $k$ times differentiable. 
We want to show that $u(v(x,y))$ is also $k$ times differentiable. Without loss of generality, we consider the mixed partial 
derivative $\frac{\partial}{\partial x^{k_{1}}\partial y^{k_{2}}}\left(u(v(x,y))\right)$ where $k_{1}+k_{2}=k$ and show that it exists. 
Observe that the partial derivative is equivalent to $\frac{\partial}{\partial x^{k_{1}}\partial y^{k_{2}-1}}
\left(u^{\prime}(v(x,y))(\frac{\partial}{\partial y}v(x,y))\right)$ provided it exists. Since, by the induction hypothesis the 
lemma is true for the $k-1$ case and $u^{\prime}(z)$ and $v(x,y)$ are $k-1$ times differentiable, the composition of them 
$u^{\prime}(v(x,y))$ is also $k-1$ times differentiable. On the other hand, $\frac{\partial}{\partial y}v(x,y)$ is 
also $k-1$ times differentiable. So, the product of them $u^{\prime}(v(x,y))(\frac{\partial}{\partial y}v(x,y))$ is also $k-1$ 
times differentiable. Hence the partial derivative $\frac{\partial}{\partial x^{k_{1}}\partial y^{k_{2}-1}}\left(u^{\prime}
(v(x,y))(\frac{\partial}{\partial y}v(x,y))\right)$ exists. Equivalently, $\frac{\partial}{\partial x^{k_{1}}\partial y^{k_{2}}}
\left(u(v(x,y))\right)$ exists. Similarly one can prove the existence of other mixed partial derivatives of $u(v(x,y))$ of order $k$. 
Hence, by induction the proof follows.
\end{proof}

Part $(b)$: From the condition of the theorem it is clear that $\exists$ a probability space $(\Omega,\mathcal{F},P)$ such that for $\omega\in \Omega$, $X(\bold{s},0)(\omega),\eta(\bold{s},t)(\omega),\epsilon(\bold{s},t)(\omega)$ 
are $k$ times differentiable functions in $\bold{s}$ where $t=1,2,3\cdots$ and $g(x)(\omega),f(x)(\omega)$ are $k$ times differentiable 
functions in $x$. Then by the above lemma $X(\bold{s},1)(\omega)=g(X(\bold{s},0)(\omega))(\omega)+\eta(\bold{s},1)(\omega)$ is 
a $k$ times differentiable function in $\bold{s}$. The rest of the proof is exactly similar as in part $(a)$. 

\end{proof}

\begin{proof}[Proof of Propositon \ref{propn:identifiability}] 

Borrowing the notations from Theorem \ref{thm:state} and Theorem \ref{thm:observe} let us denote the pdf  $[\bold{Y=y}|\boldsymbol{\theta}_{1}]$ by \\  $\mathlarger{\mathlarger{\int}_{\mathbb R^{nT}}\frac{1}{(2\pi)^\frac{nT}{2}}\frac{1}{|{\mathbf\Sigma}_{1,f,\epsilon}(\bold{x})|^\frac{1}{2}}}\times $ 

\vspace{3mm}

$\exp\left[-\frac{1}{2}{\begin{pmatrix}y(\bold{s}_{1},1)-\beta_{0f}-\beta_{1f}x(\bold{s}_{1},1)\\y(\bold{s}_{2},1)-\beta_{0f}-\beta_{1f}x(\bold{s}_{2},1)\\ \vdots\\y(\bold{s}_{n},T)-\beta_{0f}-\beta_{1f}x(\bold{s}_{n},T)\end{pmatrix}}^{\prime}{{\mathbf\Sigma}_{1,f,\epsilon}}(\bold{x})^{-1}{\begin{pmatrix}y(\bold{s}_{1},1)-\beta_{0f}-\beta_{1f}x(\bold{s}_{1},1)\\y(\bold{s}_{2},1)-\beta_{0f}-\beta_{1f}x(\bold{s}_{2},1)\\ \vdots\\y(\bold{s}_{n},T)-\beta_{0f}-\beta_{1f}x(\bold{s}_{n},T)\end{pmatrix}}\right]\mathlarger{ h_{1}(\mathbf{x}) \,d\mathbf{x}} $

\vspace{3mm}

where $\mathbf{\Sigma}_{1,f,\epsilon}(\bold{x})$ is the matrix $\mathbf{\Sigma}_{f,\epsilon}$ associated with parameter value $\boldsymbol{\theta}=\boldsymbol{\theta}_{1}$ and $\bold{x}$. $h_{1}(\mathbf{x})$ is the pdf \\ $[x(\bold{s}_{1},1),x(\bold{s}_{2},1),\cdots,x(\bold{s}_{n},T)|\boldsymbol{\theta}]$ associated with parameter value $\boldsymbol{\theta}=\boldsymbol{\theta}_{1}$. Similarly the pdf  $[\bold{Y=y}|\boldsymbol{\theta}_{2}]$ is denoted by $\mathlarger{\mathlarger{\int}_{\mathbb R^{nT}}\frac{1}{(2\pi)^\frac{nT}{2}}\frac{1}{|{\mathbf\Sigma}_{2,f,\epsilon}(\bold{x})|^\frac{1}{2}}}\times $ 

\vspace{3mm}

$\exp\left[-\frac{1}{2}{\begin{pmatrix}y(\bold{s}_{1},1)-\beta_{0f}-\frac{\beta_{1f}}{c}x(\bold{s}_{1},1)\\y(\bold{s}_{2},1)-\beta_{0f}-\frac{\beta_{1f}}{c}x(\bold{s}_{2},1)\\ \vdots\\y(\bold{s}_{n},T)-\beta_{0f}-\frac{\beta_{1f}}{c}x(\bold{s}_{n},T)\end{pmatrix}}^{\prime}{{\mathbf\Sigma}_{2,f,\epsilon}}(\bold{x})^{-1}{\begin{pmatrix}y(\bold{s}_{1},1)-\beta_{0f}-\frac{\beta_{1f}}{c}x(\bold{s}_{1},1)\\y(\bold{s}_{2},1)-\beta_{0f}-\frac{\beta_{1f}}{c}x(\bold{s}_{2},1)\\ \vdots\\y(\bold{s}_{n},T)-\beta_{0f}-\frac{\beta_{1f}}{c}x(\bold{s}_{n},T)\end{pmatrix}}\right]\mathlarger{ h_{2}(\mathbf{x}) \,d\mathbf{x}} $

\vspace{3mm}

where $\mathbf{\Sigma}_{2,f,\epsilon}(\bold{x})$ is the matrix $\mathbf{\Sigma}_{f,\epsilon}$ associated with parameter value $\boldsymbol{\theta}=\boldsymbol{\theta}_{2}$ and $\bold{x}$. $h_{2}(\mathbf{x})$ is the pdf \\ $[x(\bold{s}_{1},1),x(\bold{s}_{2},1),\cdots,x(\bold{s}_{n},T)|\boldsymbol{\theta}]$ associated with parameter value $\boldsymbol{\theta}=\boldsymbol{\theta}_{2}$. 

\vspace{3mm}

Now recall that $h_{1}(\mathbf{x})$ is given by  

$ \mathlarger{\mathlarger{{\int}_{\mathbb R^{n}} \frac{1}{(2\pi)^\frac{n}{2}}\frac{1}{|{\mathbf{\Sigma}}_{0}|^\frac{1}{2}}}}\exp\left[-\frac{1}{2}{\begin{pmatrix}x(\bold{s}_{1},0)-\mu_{01}\\x(\bold{s}_{2},0)-\mu_{02}\\ \vdots\\x(\bold{s}_{n},0)-\mu_{0n}\end{pmatrix}}^{\prime}{{\mathbf{\Sigma}}_{0}}^{-1}{\begin{pmatrix}x(\bold{s}_{1},0)-\mu_{01}\\x(\bold{s}_{2},0)-\mu_{02}\\ \vdots\\x(\bold{s}_{n},0)-\mu_{0n}\end{pmatrix}}\right] \mathlarger{\mathlarger{\frac{1}{(2\pi)^\frac{nT}{2}}\frac{1}{|\tilde{\mathbf{\Sigma}}_{1}(\bold{x})|^\frac{1}{2}}}}\times$

\vspace{3mm}

$\exp\left[-\frac{1}{2}{\begin{pmatrix}x(\bold{s}_{1},1)-\beta_{0g}-\beta_{1g}
x(\bold{s}_{1},0)\\x(\bold{s}_{2},1)-\beta_{0g}-\beta_{1g}
x(\bold{s}_{2},0)\\ \vdots\\x(\bold{s}_{n},T)-\beta_{0g}-\beta_{1g}
x(\bold{s}_{n},T-1)\end{pmatrix}}^{\prime}{\tilde{\mathbf{\Sigma}}}_{1}(\bold{x})^{-1}{\begin{pmatrix}x(\bold{s}_{1},1)-\beta_{0g}-\beta_{1g}
x(\bold{s}_{1},0)\\x(\bold{s}_{2},1)-\beta_{0g}-\beta_{1g}
x(\bold{s}_{2},0)\\ \vdots\\x(\bold{s}_{n},T)-\beta_{0g}-\beta_{1g}
x(\bold{s}_{n},T-1)\end{pmatrix}}\right]$

\vspace{3mm}

$ dx(\bold{s}_{1},0)dx(\bold{s}_{2},0)\cdots dx(\bold{s}_{n},0) $
 
\vspace{3mm} 

where $\tilde{\mathbf{\Sigma}}_{1}(\bold{x})$ is the $\tilde{\mathbf{\Sigma}}$ matrix associated with associated with parameter value $\boldsymbol{\theta}=\boldsymbol{\theta}_{1}$ and $\bold{x}$. Similarly, we may define $\tilde{\mathbf{\Sigma}}_{2}(\bold{x})$ as the $\tilde{\mathbf{\Sigma}}$ matrix associated with associated with parameter value $\boldsymbol{\theta}=\boldsymbol{\theta}_{2}$ and $\bold{x}$. Then, it is easy to see that $\tilde{\mathbf{\Sigma}}_{2}(c\bold{x})=\tilde{\mathbf{\Sigma}}_{1}(\bold{x})$, and that in turn implies $h_{2}(c\mathbf{x})=\frac{h_{1}(\mathbf{x})}{c^nT}$. \\
Hence, $[\bold{Y=y}|\boldsymbol{\theta}_{2}]=\mathlarger{\mathlarger{\int}_{\mathbb R^{nT}}\frac{1}{(2\pi)^\frac{nT}{2}}\frac{1}{|{\mathbf\Sigma}_{2,f,\epsilon}(\bold{x})|^\frac{1}{2}}}\times $
\vspace{3mm}

$\exp\left[-\frac{1}{2}{\begin{pmatrix}y(\bold{s}_{1},1)-\beta_{0f}-\frac{\beta_{1f}}{c}x(\bold{s}_{1},1)\\y(\bold{s}_{2},1)-\beta_{0f}-\frac{\beta_{1f}}{c}x(\bold{s}_{2},1)\\ \vdots\\y(\bold{s}_{n},T)-\beta_{0f}-\frac{\beta_{1f}}{c}x(\bold{s}_{n},T)\end{pmatrix}}^{\prime}{{\mathbf\Sigma}_{2,f,\epsilon}}(\bold{x})^{-1}{\begin{pmatrix}y(\bold{s}_{1},1)-\beta_{0f}-\frac{\beta_{1f}}{c}x(\bold{s}_{1},1)\\y(\bold{s}_{2},1)-\beta_{0f}-\frac{\beta_{1f}}{c}x(\bold{s}_{2},1)\\ \vdots\\y(\bold{s}_{n},T)-\beta_{0f}-\frac{\beta_{1f}}{c}x(\bold{s}_{n},T)\end{pmatrix}}\right]\mathlarger{ h_{2}(\mathbf{x}) \,d\mathbf{x}} $

\vspace{3mm}

$=\mathlarger{\mathlarger{\int}_{\mathbb R^{nT}}\frac{1}{(2\pi)^\frac{nT}{2}}\frac{1}{|{\mathbf\Sigma}_{2,f,\epsilon}(c\bold{x})|^\frac{1}{2}}}\times $
\vspace{3mm}

$\exp\left[-\frac{1}{2}{\begin{pmatrix}y(\bold{s}_{1},1)-\beta_{0f}-\frac{\beta_{1f}}{c}cx(\bold{s}_{1},1)\\y(\bold{s}_{2},1)-\beta_{0f}-\frac{\beta_{1f}}{c}cx(\bold{s}_{2},1)\\ \vdots\\y(\bold{s}_{n},T)-\beta_{0f}-\frac{\beta_{1f}}{c}cx(\bold{s}_{n},T)\end{pmatrix}}^{\prime}{{\mathbf\Sigma}_{2,f,\epsilon}}(c\bold{x})^{-1}{\begin{pmatrix}y(\bold{s}_{1},1)-\beta_{0f}-\frac{\beta_{1f}}{c}cx(\bold{s}_{1},1)\\y(\bold{s}_{2},1)-\beta_{0f}-\frac{\beta_{1f}}{c}cx(\bold{s}_{2},1)\\ \vdots\\y(\bold{s}_{n},T)-\beta_{0f}-\frac{\beta_{1f}}{c}cx(\bold{s}_{n},T)\end{pmatrix}}\right]\mathlarger{ h_{2}(c\mathbf{x}) c^nT\,d\mathbf{x}} $

\vspace{3mm}

$=\mathlarger{\mathlarger{\int}_{\mathbb R^{nT}}\frac{1}{(2\pi)^\frac{nT}{2}}\frac{1}{|{\mathbf\Sigma}_{1,f,\epsilon}(\bold{x})|^\frac{1}{2}}}\times $ 

\vspace{3mm}

$\exp\left[-\frac{1}{2}{\begin{pmatrix}y(\bold{s}_{1},1)-\beta_{0f}-\beta_{1f}x(\bold{s}_{1},1)\\y(\bold{s}_{2},1)-\beta_{0f}-\beta_{1f}x(\bold{s}_{2},1)\\ \vdots\\y(\bold{s}_{n},T)-\beta_{0f}-\beta_{1f}x(\bold{s}_{n},T)\end{pmatrix}}^{\prime}{{\mathbf\Sigma}_{1,f,\epsilon}}(\bold{x})^{-1}{\begin{pmatrix}y(\bold{s}_{1},1)-\beta_{0f}-\beta_{1f}x(\bold{s}_{1},1)\\y(\bold{s}_{2},1)-\beta_{0f}-\beta_{1f}x(\bold{s}_{2},1)\\ \vdots\\y(\bold{s}_{n},T)-\beta_{0f}-\beta_{1f}x(\bold{s}_{n},T)\end{pmatrix}}\right]\mathlarger{ h_{2}(c\mathbf{x}) c^nT\,d\mathbf{x}} $

\vspace{3mm}

$=\mathlarger{\mathlarger{\int}_{\mathbb R^{nT}}\frac{1}{(2\pi)^\frac{nT}{2}}\frac{1}{|{\mathbf\Sigma}_{1,f,\epsilon}(\bold{x})|^\frac{1}{2}}}\times $ 

\vspace{3mm}

$\exp\left[-\frac{1}{2}{\begin{pmatrix}y(\bold{s}_{1},1)-\beta_{0f}-\beta_{1f}x(\bold{s}_{1},1)\\y(\bold{s}_{2},1)-\beta_{0f}-\beta_{1f}x(\bold{s}_{2},1)\\ \vdots\\y(\bold{s}_{n},T)-\beta_{0f}-\beta_{1f}x(\bold{s}_{n},T)\end{pmatrix}}^{\prime}{{\mathbf\Sigma}_{1,f,\epsilon}}(\bold{x})^{-1}{\begin{pmatrix}y(\bold{s}_{1},1)-\beta_{0f}-\beta_{1f}x(\bold{s}_{1},1)\\y(\bold{s}_{2},1)-\beta_{0f}-\beta_{1f}x(\bold{s}_{2},1)\\ \vdots\\y(\bold{s}_{n},T)-\beta_{0f}-\beta_{1f}x(\bold{s}_{n},T)\end{pmatrix}}\right]\mathlarger{ h_{1}(c\mathbf{x}) \,d\mathbf{x}} $

\vspace{3mm}

$=[\bold{Y=y}|\boldsymbol{\theta}_{1}]$.
 
\end{proof}

\begin{proof}[Proof of Theorem \ref{thm:mathematical_identifiability}]
 Assume $[Y(\bold{s}_{1},1),\cdots, Y(\bold{s}_{n},T)|\boldsymbol{\theta}_{1}] \stackrel{d}{=} [Y(\bold{s}_{1},1),\cdots, Y(\bold{s}_{n},T)|\boldsymbol{\theta}_{2}]$. We want to show $\boldsymbol{\theta}_{1}=\boldsymbol{\theta}_{2}$. Since, in this case $\lambda_{f},\lambda_{g},\sigma_{f}^2,\sigma_{g}^2,\lambda_{0},\sigma_{0}^2,\lambda_{\eta},\sigma_{\eta}^2,\beta_{0f},\beta_{1f},{\boldsymbol{\mu}}_{0}$ are fixed so it is enough to show $(\beta_{0g1},\beta_{1g1},\lambda_{\epsilon1},\sigma_{\epsilon1}^2)=(\beta_{0g2},\beta_{1g2},\lambda_{\epsilon2},\sigma_{\epsilon2}^2)$.
 
 From the proof of part (a) of Theorem \ref{thm:observe} we get that for fixed $x(\bold{s},t)$, $Y(\bold{s},t)$ is distributed as a Gaussian with mean $\beta_{0f}+\beta_{1f}x(\bold{s},t)$. Hence, $E(Y(\bold{s},t))=E_{X(\bold{s},t)}(E(Y(\bold{s},t)|X(\bold{s},t))=E(\beta_{0f}+\beta_{1f}X(\bold{s},t))=\beta_{0f}+\beta_{1f}E(X(\bold{s},t))$. Now, put $t=1$ to get $E(Y(\bold{s},1))=\beta_{0f}+\beta_{1f}E(X(\bold{s},1))$. But, $E(X(\bold{s},1))=E_{X(\bold{s},0)}(E(X(\bold{s},1)|X(\bold{s},0))=E_{X(\bold{s},0)}(g(X(\bold{s},0))+\eta(\bold{s},0))=\beta_{0g}+\beta_{1g}\mu_{0}(\bold{s})$. Combining these results we get
 \begin{equation}
  E(Y(\bold{s},1))=\beta_{0f}+\beta_{1f}(\beta_{0g}+\beta_{1g}\mu_{0}(\bold{s}))
 \end{equation}

 Again from the proof of part (a) of Theorem \ref{thm:observe} we get that $Cov(Y(\bold{s},t),Y(\bold{s}^{\prime},t^{\prime})|X(\bold{s},t),X(\bold{s}^{\prime},t^{\prime}))=c_{f}(x(\bold{s},t),x(\bold{s}^{\prime},t^{\prime}))+c_{\epsilon}(\bold{s},\bold{s}^{\prime})\delta(t-t^{\prime})$ where $\delta(\cdot)$ is the 
delta function i.e. $\delta(t)=1$ for $t=0$ and $=0$ otherwise. Since, by assumptions of the Theorem \ref{thm:mathematical_identifiability} all the covariance kernels are squared exponential type this expression becomes $\sigma^{2}_{f}e^{-\lambda_{f}(x(\bold{s},t)-x(\bold{s}^{\prime},t^{\prime}))^{2}}+\sigma^{2}_{\epsilon}e^{-\lambda_{\epsilon}||\bold{s}-\bold{s}^{\prime}||^{2}}\delta(t-t^{\prime})$. Now, put $t=t^{\prime}=1$ to get $Cov(Y(\bold{s},1),Y(\bold{s}^{\prime},1)|X(\bold{s},1),X(\bold{s}^{\prime},1))=\sigma^{2}_{f}e^{-\lambda_{f}(x(\bold{s},1)-x(\bold{s}^{\prime},1))^{2}}+\sigma^{2}_{\epsilon}e^{-\lambda_{\epsilon}||\bold{s}-\bold{s}^{\prime}||^{2}}$. Hence,
\begin{align*}
&Cov(Y(\bold{s},1),Y(\bold{s}^{\prime},1))\\
&=E_{X(\bold{s},1),X(\bold{s}^{\prime},1)}(Cov(Y(\bold{s},1),Y(\bold{s}^{\prime},1)|X(\bold{s},1),X(\bold{s}^{\prime},1)))+\\
&Cov_{X(\bold{s},1),X(\bold{s}^{\prime},1)}(E(Y(\bold{s},1)|X(\bold{s},1)),E(Y(\bold{s}^{\prime},1)|X(\bold{s}^{\prime},1)))\\
=&\sigma^{2}_{f}E(e^{-\lambda_{f}(x(\bold{s},1)-x(\bold{s}^{\prime},1))^{2}})+\sigma^{2}_{\epsilon}e^{-\lambda_{\epsilon}||\bold{s}-\bold{s}^{\prime}||^{2}}+Cov_{X(\bold{s},1),X(\bold{s}^{\prime},1)}(E(Y(\bold{s},1)|X(\bold{s},1)),E(Y(\bold{s}^{\prime},1)|X(\bold{s}^{\prime},1)))\\
=&\sigma^{2}_{f}E(e^{-\lambda_{f}(x(\bold{s},1)-x(\bold{s}^{\prime},1))^{2}})+\sigma^{2}_{\epsilon}e^{-\lambda_{\epsilon}||\bold{s}-\bold{s}^{\prime}||^{2}}+Cov(\beta_{0f}+\beta_{1f}X(\bold{s},1),\beta_{0f}+\beta_{1f}X(\bold{s}^{\prime},1))\\
=&\sigma^{2}_{f}E(e^{-\lambda_{f}(x(\bold{s},1)-x(\bold{s}^{\prime},1))^{2}})+\sigma^{2}_{\epsilon}e^{-\lambda_{\epsilon}||\bold{s}-\bold{s}^{\prime}||^{2}}+\beta^{2}_{1f}Cov(X(\bold{s},1),X(\bold{s}^{\prime},1))\\
\end{align*}
Now, 
\begin{align*}
 &Cov(X(\bold{s},1),X(\bold{s}^{\prime},1))\\
&=E_{X(\bold{s},0),X(\bold{s}^{\prime},0)}(Cov(X(\bold{s},1),X(\bold{s}^{\prime},1)|X(\bold{s},0),X(\bold{s}^{\prime},0)))+\\
&Cov_{X(\bold{s},0),X(\bold{s}^{\prime},0)}(E(X(\bold{s},1)|X(\bold{s},0)),E(X(\bold{s}^{\prime},1)|X(\bold{s}^{\prime},0)))\\
=&\sigma^{2}_{g}E(e^{-\lambda_{g}(x(\bold{s},0)-x(\bold{s}^{\prime},0))^{2}})+\sigma^{2}_{\eta}e^{-\lambda_{\eta}||\bold{s}-\bold{s}^{\prime}||^{2}}+Cov_{X(\bold{s},0),X(\bold{s}^{\prime},0)}(E(X(\bold{s},1)|X(\bold{s},0)),E(X(\bold{s}^{\prime},1)|X(\bold{s}^{\prime},0)))\\
=&\sigma^{2}_{g}E(e^{-\lambda_{g}(x(\bold{s},0)-x(\bold{s}^{\prime},0))^{2}})+\sigma^{2}_{\eta}e^{-\lambda_{\eta}||\bold{s}-\bold{s}^{\prime}||^{2}}+Cov(\beta_{0g}+\beta_{1g}X(\bold{s},0),\beta_{0g}+\beta_{1g}X(\bold{s}^{\prime},0))\\
=&\sigma^{2}_{g}E(e^{-\lambda_{g}(x(\bold{s},0)-x(\bold{s}^{\prime},0))^{2}})+\sigma^{2}_{\eta}e^{-\lambda_{\eta}||\bold{s}-\bold{s}^{\prime}||^{2}}+\beta^{2}_{1g}Cov(X(\bold{s},0),X(\bold{s}^{\prime},0))\\
=&\sigma^{2}_{g}E(e^{-\lambda_{g}(x(\bold{s},0)-x(\bold{s}^{\prime},0))^{2}})+\sigma^{2}_{\eta}e^{-\lambda_{\eta}||\bold{s}-\bold{s}^{\prime}||^{2}}+\beta^{2}_{1g}\sigma^{2}_{0}e^{-\lambda_{0}||\bold{s}-\bold{s}^{\prime}||^{2}}\\
\end{align*}
Combining these results we get
 \begin{align}
  &Cov(Y(\bold{s},1),Y(\bold{s}^{\prime},1))\notag \\ 
  =&\sigma^{2}_{f}E(e^{-\lambda_{f}(x(\bold{s},1)-x(\bold{s}^{\prime},1))^{2}})+\sigma^{2}_{\epsilon}e^{-\lambda_{\epsilon}||\bold{s}-\bold{s}^{\prime}||^{2}}+\beta^{2}_{1f}\sigma^{2}_{g}E(e^{-\lambda_{g}(x(\bold{s},0)-x(\bold{s}^{\prime},0))^{2}}) \notag \\
  &+\beta^{2}_{1f}\sigma^{2}_{\eta}e^{-\lambda_{\eta}||\bold{s}-\bold{s}^{\prime}||^{2}}+\beta^{2}_{1f}\beta^{2}_{1g}\sigma^{2}_{0}e^{-\lambda_{0}||\bold{s}-\bold{s}^{\prime}||^{2}}
 \end{align}
Now, see that $[Y(\bold{s}_{1},1),\cdots, Y(\bold{s}_{n},T)|\boldsymbol{\theta}_{1}] \stackrel{d}{=} [Y(\bold{s}_{1},1),\cdots, Y(\bold{s}_{n},T)|\boldsymbol{\theta}_{2}]$ implies $E_{\boldsymbol{\theta}_{1}}(Y(\bold{s}_{i},1))=E_{\boldsymbol{\theta}_{2}}(Y(\bold{s}_{i},1))$. Hence, by equation (12) 
\begin{equation}
 (\beta_{0f}-\beta_{0f})+\beta_{1f}(\beta_{0g1}-\beta_{0g2})+\beta_{1f}(\beta_{1g1}\mu_{0}(\bold{s}_{i})-\beta_{1g2}\mu_{0}(\bold{s}_{i}))=0\notag
\end{equation} 
\begin{equation}
 \Leftrightarrow\beta_{1f}(\beta_{0g1}-\beta_{0g2})+\beta_{1f}(\beta_{1g1}\mu_{0}(\bold{s}_{i})-\beta_{1g2}\mu_{0}(\bold{s}_{i}))=0
\end{equation} 
Similarly, noting that the term $\sigma^{2}_{f}E(e^{-\lambda_{f}(x(\bold{s},1)-x(\bold{s}^{\prime},1))^{2}})+\beta^{2}_{1f}\sigma^{2}_{g}E(e^{-\lambda_{g}(x(\bold{s},0)-x(\bold{s}^{\prime},0))^{2}})$ remains same under $\boldsymbol{\theta}_{1}$ and $\boldsymbol{\theta}_{2}$; $Cov_{\boldsymbol{\theta}_{1}}(Y(\bold{s}_{i},1),Y(\bold{s}_{j},1))=Cov_{\boldsymbol{\theta}_{2}}(Y(\bold{s}_{i},1),Y(\bold{s}_{j},1))$ implies (by equation (13))
\begin{equation}
 \beta^{2}_{1f}\sigma^{2}_{0}e^{-\lambda_{0}||\bold{s}_{i}-\bold{s}_{j}||^{2}}(\beta^{2}_{1g1}-\beta^{2}_{1g2})+(\sigma^{2}_{\epsilon1}e^{-\lambda_{\epsilon1}||\bold{s}_{i}-\bold{s}_{j}||^{2}}-\sigma^{2}_{\epsilon2}e^{-\lambda_{\epsilon2}||\bold{s}_{i}-\bold{s}_{j}||^{2}})=0
\end{equation}
Now we, solve equation (14) and (15) to show that $(\beta_{0g1},\beta_{1g1},\lambda_{\epsilon1},\sigma_{\epsilon1}^2)=(\beta_{0g2},\beta_{1g2},\lambda_{\epsilon2},\sigma_{\epsilon2}^2)$. Note that, equation (15) can be represented as $ax^{t}+by^{t}=cz^{t}$ where $a=\beta^{2}_{1f}\sigma^{2}_{0}(\beta^{2}_{1g1}-\beta^{2}_{1g2}), x=e^{-\lambda_{0}},b=\sigma^{2}_{\epsilon1},y= e^{-\lambda_{\epsilon1}},c=\sigma^{2}_{\epsilon2}$ and $z=e^{-\lambda_{\epsilon2}}$. Then by assumption (C) $\exists \  t_{1}=d_{1}^{2},t_{2}=d_{2}^{2},t_{3}=d_{3}^{2}$ such that 
\begin{align}
 ax^{t_{1}}+by^{t_{1}}&=cz^{t_{1}}\\
 ax^{t_{2}}+by^{t_{2}}&=cz^{t_{2}}\\
 ax^{t_{3}}+by^{t_{3}}&=cz^{t_{3}}
\end{align}
WLOG we assume $t_{1}<t_{2}<t_{3}$. There are three possibilities, i.e. $a=0$, $a>0$ and $a< 0$.\\
Case I: Assume $a=0$. Then equation (16) $\Rightarrow by^{t_{1}}=cz^{t_{1}}$ \ $\Rightarrow by^{t_{1}}y^{t_{2}-t_{1}}=cz^{t_{1}}y^{t_{2}-t_{1}}$ \ $\Rightarrow by^{t_{2}}=cz^{t_{1}}y^{t_{2}-t_{1}}$. Now, using equation (17) we get that  
$cz^{t_{2}}=cz^{t_{1}}y^{t_{2}-t_{1}}$ \ $\Rightarrow cz^{t_{2}-t_{1}}=cy^{t_{2}-t_{1}}$ \ $\Rightarrow z^{t_{2}-t_{1}}=y^{t_{2}-t_{1}}$ (since by the assumption $c=\sigma^{2}_{\epsilon2}\neq 0$). So, $z=y$ \ $\Rightarrow \lambda_{\epsilon1}=\lambda_{\epsilon2}$. Now, $a=0$ \ $\Rightarrow \beta^{2}_{1g1}=\beta^{2}_{1g2}$. Moreover, using $z=y$ in equation (16) \ $\Rightarrow bz^{t_{1}}=cz^{t_{1}}$\ $\Rightarrow b=c$\ $\Rightarrow \sigma^{2}_{\epsilon1}=\sigma^{2}_{\epsilon2}$. \\
Case II: Assume $a> 0$. We show it leads to contradiction. Since, by assumption $x\neq y$ and $x \neq z$, the case II can be exhaustively represented by $3$ sub cases; sub case IIa : $x < min(y,z)$, sub case IIb : $x > max(y,z)$ and sub case IIc : $min(y,z) < x < max(y,z)$. \\
Subcase IIa: If $x<y\leq z$ then see that $cz^{t_{2}}=ax^{t_{2}}+by^{t_{2}}=ax^{t_{1}}x^{t_{2}-t_{1}}+by^{t_{1}}y^{t_{2}-t_{1}}<(ax^{t_{1}}+by^{t_{1}})z^{t_{2}-t_{1}}=cz^{t_{1}}z^{t_{2}-t_{1}}=cz^{t_{2}}$; contradiction. Or, if $x<z<y$; then $x^{t_{1}}<z^{t_{1}}<y^{t_{1}}$. Hence, $z^{t_{1}}=\frac{a}{c}x^{t_{1}}+\frac{b}{c}y^{t_{1}}$ implies $\frac{a}{c}+\frac{b}{c}=1$. Then by Jensen's inequality, $z^{t_{2}}={z^{t_{1}}}^{\frac{t_{2}}{t_{1}}}=(\frac{a}{c}x^{t_{1}}+\frac{b}{c}y^{t_{1}})^{\frac{t_{2}}{t_{1}}}< \frac{a}{c}(x^{t_{1}})^{\frac{t_{2}}{t_{1}}}+\frac{b}{c}(y^{t_{1}})^{\frac{t_{2}}{t_{1}}}=\frac{a}{c}x^{t_{2}}+\frac{b}{c}y^{t_{2}}=z^{t_{2}}$; contradiction.\\
Subcase IIb: If $x>y\geq z$ then see that $cz^{t_{2}}=ax^{t_{2}}+by^{t_{2}}=ax^{t_{1}}x^{t_{2}-t_{1}}+by^{t_{1}}y^{t_{2}-t_{1}}>(ax^{t_{1}}+by^{t_{1}})z^{t_{2}-t_{1}}=cz^{t_{1}}z^{t_{2}-t_{1}}=cz^{t_{2}}$; contradiction. Or, if $x>z>y$; then $x^{t_{1}}>z^{t_{1}}>y^{t_{1}}$. Hence, $z^{t_{1}}=\frac{a}{c}x^{t_{1}}+\frac{b}{c}y^{t_{1}}$ implies $\frac{a}{c}+\frac{b}{c}=1$. Then again by Jensen's inequality, we get contradiction.\\
Subcase IIc: If $y<x<z$ then see that $cz^{t_{2}}=ax^{t_{2}}+by^{t_{2}}=ax^{t_{1}}x^{t_{2}-t_{1}}+by^{t_{1}}y^{t_{2}-t_{1}}<(ax^{t_{1}}+by^{t_{1}})z^{t_{2}-t_{1}}=cz^{t_{1}}z^{t_{2}-t_{1}}=cz^{t_{2}}$; contradiction. Or, if $z<x<y$; then $cz^{t_{2}}=ax^{t_{2}}+by^{t_{2}}=ax^{t_{1}}x^{t_{2}-t_{1}}+by^{t_{1}}y^{t_{2}-t_{1}}>(ax^{t_{1}}+by^{t_{1}})z^{t_{2}-t_{1}}=cz^{t_{1}}z^{t_{2}-t_{1}}=cz^{t_{2}}$; contradiction.\\
Case III: Assume $a< 0$. Now, we put $-a$ in place of $a$ and then symmetry and subcase IIb leads to contradiction.\\
With that we establish that $a$ must be $0$ and $\lambda_{\epsilon1}=\lambda_{\epsilon2},\sigma^{2}_{\epsilon1}=\sigma^{2}_{\epsilon2}$ and $\beta^{2}_{1g1}=\beta^{2}_{1g2}$. \\
Now we show that $\beta_{1g1}=-\beta_{1g2}$ leads to contradiction implying $\beta_{1g1}=\beta_{1g2}$. We substitute $\beta_{1g1}=-\beta_{1g2}$ in equation (14) to get $\beta_{1f}(\beta_{0g1}-\beta_{0g2})+2\beta_{1f}(\beta_{1g1}\mu_{0}(\bold{s}_{i}))=0$ i.e. $\mu_{0}(\bold{s}_{i})=\frac{\beta_{0g1}-\beta_{0g2}}{-2\beta_{1g1}}$; contradiction since by assumption (C) of Theorem \ref{thm:mathematical_identifiability} $\exists$ at least one pair $i,j$ such that $\mu_{0}(\bold{s}_{i}) \neq \mu_{0}(\bold{s}_{j})$. So, $\beta_{1g1}=\beta_{1g2}$. Now, we replace that in equation (14) and get $\beta_{1f}(\beta_{0g1}-\beta_{0g2})=0$ which implies $\beta_{0g1}=\beta_{0g2}$.
\end{proof}

\end{document}